\DocumentMetadata{
    pdfstandard=A-2b,
    pdfversion=1.7,
}

\PassOptionsToPackage{rgb}{xcolor}

\documentclass[sigconf, nonacm, authorversion, screen, pdfa]{acmart}

\usepackage[
	type={CC},
	modifier={by-nc-nd},
	version={4.0},
]{doclicense}

\usepackage[utf8]{inputenc}

\usepackage[
	linesnumbered,
	lined,
	algoruled,
	noend,
]{algorithm2e}
\SetAlCapNameSty{textbf}

\usepackage{dirtytalk}%
\usepackage{listings}
\usepackage{enumitem}
\usepackage{multirow}
\usepackage{subcaption}
\usepackage{xspace}

\providecolor{MaterialPurple800}{HTML}{6A1B9A}
\providecolor{MaterialGreen500}{HTML}{4CAF50}

\lstdefinelanguage{pgconstraints}{%
	alsoletter={
	},
	morekeywords={
		GRAPH,
		TYPE,
		CREATE,
		LOOSE,
		STRICT,
		OPEN,
		OPTIONAL,
		FOR,
		IN,
		EXCLUSIVE,
		MANDATORY,
		SINGLETON,
		IDENTIFIER,
		WITHIN,
		WHERE,
		NULL,
		AND,
		OR,
		MATCH,
		FILTER,
		IS,
		EXISTS,
		ALL,
		ANY,
		SHORTEST,
		TRAIL,
		ACYCLIC,
		LET,
		USE,
		RETURN,
		AS,
		THEN,
		INTERSECT,
		UNION,
		EXCEPT,
		NOT,
		STRING,
		UINT16,
		UINT32,
	},
	morecomment=[l]{//},
	sensitive=true,
}

\lstdefinestyle{pgconstraints}{%
	language=pgconstraints,
	basicstyle=\footnotesize\ttfamily,%
	keywordstyle={\color{MaterialPurple800}\upshape\bfseries},
	backgroundcolor=\color{white},%
	commentstyle=\color{MaterialGreen500}\upshape,
	texcl=true,%
	numbers=none,
	numberstyle=\tiny\upshape,
	tabsize=2,
	escapechar=|,
	mathescape,
	float,
	framerule=1pt,%
	xleftmargin=0pt,%
	xrightmargin=0pt,%
	breaklines=true,
}
\graphicspath{{figures/plots/},{figures/generated}}

\newcommand{\refapp}[2]{Appendix~\ref{#1}}

\newcommand{\deffunc}[3]{#1\colon #2 \to #3}
\newcommand{\defpartfunc}[3]{#1\colon #2 \rightharpoonup #3}
\newcommand{\tuple}[1]{\overline{#1}} %
\newcommand{\abs}[1]{\lvert #1\rvert}

\newcommand{\range}[2][1]{\{#1,\ldots,#2\}}

\newcommand{\nodeset}{N}
\newcommand{\edgeset}{E}
\newcommand{\graph}{G}
\newcommand{\weightof}[1]{w(#1)}
\newcommand{\exgraph}{\graph_{\texttt{e}}}

\newcommand{\hgraph}{\mathcal{H}}
\newcommand{\hedgeset}{\mathcal{E}}
\newcommand{\hnodeset}{\mathcal{W}}

\newcommand{\labelset}{\mathcal{L}}
\newcommand{\keyset}{\mathcal{K}}
\newcommand{\valueset}{\mathcal{V}}
\newcommand{\idset}{\mathcal{I}}
\newcommand{\pgkey}[1]{\texttt{#1}} %
\newcommand{\pglabel}[1]{\texttt{#1}}

\newcommand{\pgrelmap}{\rho}
\newcommand{\pglabelmap}{\lambda}
\newcommand{\pgpropmap}{\pi}

\newcommand{\pgtuple}{(\nodeset,\edgeset,\pgrelmap,\pglabelmap,\pgpropmap)}

\newcommand{\prop}[1]{\texttt{#1}}
\newcommand{\propof}[2]{#1\texttt{.}\prop{#2}}

\newcommand{\match}{\mu}

\newcommand{\gpceplabel}[1]{\raisebox{-2pt}[0pt][0pt]{\ensuremath{\scriptstyle{\,{:}\,\pglabel{#1}}}}}
\newcommand{\gpcnp}[2][]{(#1\,{:}\,\pglabel{#2})}
\newcommand{\gpcep}[1]{\xrightarrow{\gpceplabel{#1}}}
\newcommand{\gpceprev}[1]{\xleftarrow{\gpceplabel{#1}}}
\newcommand{\gpcrep}[3][0]{\big[#3\big]^{#2}}

\newcommand{\constraintset}{\Sigma}
\newcommand{\constraint}{\sigma}
\newcommand{\cimp}{\Rightarrow}

\newcommand{\autA}{\mathcal{A}}

\newcommand{\stateA}{q} %
\newcommand{\statesetA}{Q}%
\newcommand{\transitionRel}{\Delta}%

\newcommand{\istateA}{\stateA_{0}}%
\newcommand{\astateset}{F} %
\newcommand{\astateA}{\stateA_{f}}%

\newcommand{\query}{q}

\newcommand{\complexityclass}[1]{\textnormal{\textsf{#1}}}

\makeatletter
\newcommand{\pgnid}[1]{#1}
\newcommand{\pgeid}[1]{#1}
\newcommand{\pathlikenode}[1]{\pgnid{#1}\setpathlikeentityedge}
\newcommand{\pathlikeedge}[1]{\pgeid{#1}\setpathlikeentitynode}
\newcommand{\pathlikeentity}[1]{\pathlikenode{#1}}
\newcommand{\setpathlikeentitynode}{\renewcommand{\pathlikeentity}[1]{\pathlikenode{##1}}}
\newcommand{\setpathlikeentityedge}{\renewcommand{\pathlikeentity}[1]{\pathlikeedge{##1}}}
\newcommand{\pathlikeseparator}{}
\newcommand{\enablepathlikeseparator}[1]{\renewcommand{\pathlikeseparator}{#1}}
\newcommand{\disablepathlikeseparator}{\renewcommand{\pathlikeseparator}{}}
\newcommand{\pathlike}[2][]{\setpathlikeentitynode\disablepathlikeseparator\@for\entity:=#2\do{\pathlikeseparator\pathlikeentity{\entity}\enablepathlikeseparator{#1}}}
\makeatother

\newcommand{\pgpath}[1]{\pathlike{#1}}
\newcommand{\pgerr}[1]{\{\pathlike[,]{#1}\}}

\newcommand{\NP}{\complexityclass{NP}}

\newcommand{\highlightcell}[1]{\hspace{-4pt}\begingroup\setlength{\fboxrule}{1pt}\fcolorbox{blue}{white}{\bfseries #1}\endgroup}

\begin{document}

\title{Repairing Property Graphs under PG-Constraints}

\author{Christopher Spinrath}
\affiliation{%
	\institution{Lyon 1 University, Liris CNRS}
	\city{Lyon}
	\country{France}
}
\email{christopher.spinrath@liris.cnrs.fr}

\author{Angela Bonifati}
\orcid{https://orcid.org/0000-0002-9582-869X}
\affiliation{%
	\institution{Lyon 1 University, Liris CNRS \& IUF}
	\city{Lyon}
	\country{France}
}
\email{angela.bonifati@univ-lyon1.fr}

\author{Rachid Echahed}
\affiliation{%
	\institution{CNRS LIG, Univ. Grenoble Alpes}
	\city{Grenoble}
	\country{France}
}
\email{rachid.echahed@imag.fr}

\begin{abstract}
Recent standardization efforts for graph databases lead to standard query languages like GQL and SQL/PGQ, and constraint languages like Property Graph Constraints (PG-Constraints). In this paper, we embark on the study of repairing property graphs under PG-Constraints. %
We identify a significant subset of PG-Constraints, encoding denial constraints and including recursion as a key feature, while still permitting automata-based structural analyses of errors.
We present a comprehensive repair pipeline for these constraints to repair Property Graphs, involving changes in the graph topology and leading to node, edge and, optionally, label deletions.
We investigate three algorithmic strategies for the repair procedure, based on Integer Linear Programming (ILP), a naive, and an LP-guided greedy algorithm.
Our experiments on various real-world datasets reveal that repairing with label deletions can achieve a 59\% reduction in deletions compared to node/edge deletions.
Moreover, the LP-guided greedy algorithm offers a runtime advantage of up to 97\% compared to the ILP strategy, while matching the same quality.%
 \end{abstract}

\maketitle

\pagestyle{plain}
\begingroup
\renewcommand\thefootnote{}\footnote{\noindent
	\raggedright Copyright~\copyright2026~\authors
	\doclicenseThis
}\addtocounter{footnote}{-1}%
\endgroup

\begingroup\small\noindent\raggedright\textbf{Related Version:}\\
This paper, without the appendix, has been accepted for publication in Volume 19 of PVLDB and will be presented at the 52nd International Conference on Very Large Data Bases (VLDB 2026).
\endgroup

\vspace{.3cm}
\begingroup\small\noindent\raggedright\textbf{Artifact Availability:}\\
The source code, data, and/or other artifacts have been made available at \url{https://doi.org/10.5281/zenodo.18301604}.
\endgroup

\section{Introduction}\label{section:intro}
Repairing errors and inconsistencies is an elemental task when working with large datasets.
Applications arise naturally as part of data integration scenarios or if requirements -- e.g., the schema -- of databases change.
Consequently, error repairing has been studied extensively in the literature, for various database models \cite{Chu2016}.

Property graph languages spanning query languages, i.e.\ GQL and SQL/PGQ \cite{Deutsch2022,GQLStandard}, and schema and constraint languages, e.g.\ PG-Schema and PG-Keys \cite{Angles2021,Angles2023}, have recently been standardized.

To the best of our knowledge, repairing techniques for property graphs under these schema and constraint languages have not been investigated yet. %
We embark on the study of \emph{qualitative} error repairing \cite{Chu2016} for \emph{property graphs}, that is, error repairing with respect to \emph{constraints}; or more precisely, a subset of PG-Constraints~\cite{Angles2021,Angles2023}, which are part of the PG-Schema formalism and extend PG-Keys.%

To motivate our work, we illustrate property graphs and PG-Constraints by means of a running example.
\begin{example}\label{example:intro-constraint-desc}
Figure~\ref{figure:example-pg-intro} depicts the property graph \(\exgraph\) modelling an organisation with persons, tasks, and documents for planning activities.
Persons can work on tasks, which reference documents required for the task.
Both, nodes and edges, can have multiple labels, and can be equipped with key-value pairs, called \emph{properties}.
For instance, the node \(d_3\) has the two labels \pglabel{document} and \pglabel{important}, as well as the two properties \prop{\#pages} and \prop{access\_level}.
\begin{figure*}
		\includegraphics[width=.85\linewidth]{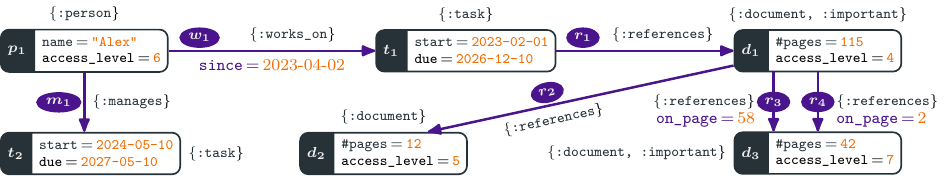}\\[.5em]

		\begin{minipage}{.9175\textwidth}
		\hrule
		~\\
		\emph{\small PG-Constraint in GQL}
		\begin{lstlisting}[numbers=none, gobble=6,basicstyle=\scriptsize\ttfamily]
			FOR x, y MATCH z = (x:person)-[:works_on]->(u:task)(-[:references]->(:document))+(y:document & important) FILTER u.start <= NOW()
			MANDATORY x.access_level, y.access_level FILTER x.access_level >= y.access_level
		\end{lstlisting}
		\hrule
	\end{minipage}
	\caption{%
		A property graph modelling an organisation with persons, tasks, and documents, along with a PG-Constraint%
	}
	\label{figure:example-pg-intro}
\end{figure*}

\noindent	Let us consider the following constraint for the property graph \(G_e\).
	\begin{quote}
		\say{%
			If a person works on a task, which has started and which references directly or indirectly, i.e.\ recursively, an important document, then the person's access level is at least as high as the (required) access level of the referenced document.%
		}
	\end{quote}
	The constraint describes a \emph{graph pattern} in its precondition:
	The part \say{a person works on a task} refers to a node labelled \pglabel{person} and a node labelled \pglabel{task} which are connected by an edge labelled \pglabel{works\_on}.
	A \emph{match} of this partial pattern in \(\exgraph\) consists of the nodes \(p_1\), \(t_1\), and the edge \(w_1\).
	The described pattern also extends to documents which are referenced -- directly or indirectly -- by the task.
	The match for the partial pattern above could be extended in several ways: \(d_1\), \(d_2\), and \(d_3\) are all documents referenced by \(t_1\), where \(d_2\) and \(d_3\) are referenced indirectly.
	In general, reference chains might be of arbitrary length.
	However, the document \(d_2\) is not labelled \pglabel{important} and thus it should not be matched.

	Lastly, the constraint compares property values, namely the \prop{start} property of \(t_1\) to assert that the task has started, and the \prop{access\_level} properties of \(p_1\) and, depending of which node is matched, \(d_1\) or \(d_3\). A possible specification of this constraint as a \emph{PG-Constraint} \cite{Angles2021} is illustrated in Figure~\ref{figure:example-pg-intro}.

	\(G_e\) violates the above constraint since \(p_1\), \(w_1\), \(t_1\), \(r_1\), \(d_1\), \(r_4\), and~\(d_3\) (with \(d_3\) being the important document in question) form a match, the task \(t_1\) has started, but \(d_3\) has a higher access level than \(p_1\).
\end{example}

\paragraph{Contributions}
We identify \emph{regular graph pattern calculus (RGPC) constraints} as a subset of PG-Constraints, and introduce an automata model for its patterns to reason about the structure of errors.
At the same time, our subset entails major features of PG-Constraints; especially, recursion, complex restrictions of labels for nodes and edges, as well as comparability of node properties.
They can express denial constraints, which capture, among others, functional dependencies, including the constraint from Example~\ref{example:intro-constraint-desc}.
For specifying patterns we use a subset of the graph pattern calculus GPC \cite{Francis2023a}, which captures the core features for patterns in GQL and SQL/PG.

	For repairing databases various repair models have been studied in the literature \cite{Fan2008}, including the addition of new data \cite{Arenas1999}, changing data values \cite{Wijsen2005,Bohannon2005,Lopatenko2007} -- i.e. values of properties -- and deletion~\cite{DBLP:journals/iandc/ChomickiM05}.
	The choice of a repair model depends on the setting.

In this paper, we assume the repair model of deletions \cite[see, e.g.][]{Fan2008,DBLP:journals/pvldb/LiuSGGKR24,DBLP:journals/iandc/ChomickiM05} in which property graphs are repaired by deleting objects from it, i.e.\ nodes, edges, or labels.
Deletions make sense as a repair model in several applications where
inserting new data is undesirable or too risky, for example, due to privacy or security reasons.
More concretely, in a social network inserting a friendship could expose private data, but deleting one can be rectified.
In a property graph modelling a supply chain, inserting a \pglabel{sourced\_from} relationship between a manufacturer and a product which the manufacturer cannot actually supply can be disastrous for a company, while deleting such a relationship would highlight the issue of unavailability.
In a scenario as modelled by the property graph shown in Figure~\ref{figure:example-pg-intro}, increasing, that is changing, the access level of a person or document is a security risk.
Recall that the property graph shown in Figure~\ref{figure:example-pg-intro} does not satisfy the constraint from Example~\ref{example:intro-constraint-desc}, since the person represented by \(p_1\) works on the task \(t_1\) which references the important document \(d_3\).
But the access level of \(p_1\) is not high enough to access \(d_3\).
One way to repair the property graph is to delete the edge \(w_1\), effectively removing \(p_1\) from working on \(t_1\).
Another possibility is to remove the label \pglabel{important} from \(d_3\).

In this paper, we propose a comprehensive pipeline (cf., Figure~\ref{figure:pipeline}) featuring different graph repair options.
The repair method follows a \emph{holistic} approach, meaning it detects and repairs
all errors simultaneously. This ensures
that the fewest possible objects are removed, and their relationships are respected:
For example, the deletion of a node entails deleting all its incident edges.
However, managing multiple errors -- especially those involving long, arbitrary-length paths rather than fixed-length tuples -- can be
time-consuming.
To mitigate this, we propose three algorithms for tackling the underlying combinatorial
problem of determining which objects to delete.
The first is a naive greedy
algorithm, which sequentially repairs each error by deleting an
object from it, unless it has already
been resolved through prior deletions. The second formulates the
problem as an integer linear program (ILP), using a
solver to derive an optimal set of objects for deletion.
The third one is an LP-guided greedy algorithm which uses non-integer solutions of the LP-relaxation of the ILP to guide the greedy algorithm.

Finally, we gauge the effectiveness of our repair pipeline and analyse the trade-offs associated with the different algorithms and options through experiments conducted on real-world datasets.
They reveal that the LP-guided greedy algorithm can outperform the other algorithms, while maintaining the quality, i.e.\ number of deletions, of the ILP strategy.
Permitting label deletions can reduce the number of deletions by up to \(59\%\), but increases the runtime (by up to \(5\times\)).
Moreover, our experiments reveal that approximate versions of our algorithms can be up to \(89\%\) faster, while yielding a good approximation or even a repair.

\paragraph{Structure}
In Section~\ref{section:setting} we recall some basics on property graphs, and present our subset of the graph pattern calculus (GPC).
We then introduce our RGPC constraints in Section~\ref{section:constraints}.
In Section~\ref{section:cleaning} we present our pipeline and algorithm(s) for error repairing, and in Section~\ref{section:experiments} our experiments.
Finally, we discuss further related work in Section~\ref{section:related-work}, and conclude in Section~\ref{section:conclusion}.
\section{Preliminaries}\label{section:setting}

In this section, we provide basics on the data model used in property graphs, and introduce the subset of the pattern calculus GPC \cite{Francis2023a} which we use as an ingredient for our constraints in Section~\ref{section:constraints}.

\subsection{Property Graphs}\label{section:property-graphs}
By \(\labelset\), \(\keyset\), \(\valueset\), and \(\idset\) we denote countably infinite, pairwise disjoint sets of \emph{labels}, \emph{keys} (of properties), \emph{property values}, and \emph{identifiers} (\emph{ids} for short), respectively.

A \emph{property graph} is a tuple \(\graph = \pgtuple\), where
\begin{itemize}
	\item \(\nodeset\subset\idset\) is a finite set of node identifiers;
	\item \(\edgeset\subset\idset\) is a finite set of edge identifiers disjoint from \(\nodeset\);
	\item \(\deffunc{\pgrelmap}{\edgeset}{(\nodeset \times \nodeset)}\) defines the endpoints, i.e.\ the source and target, for all edges;
	\item \(\deffunc{\pglabelmap}{(\nodeset\cup\edgeset)}{2^{\labelset}}\) is a labelling function that assigns a finite set of labels to node and edge identifiers; and
	\item \(\defpartfunc{\pgpropmap}{((\nodeset\cup\edgeset)\times \keyset)}{\valueset}\) is a partial function assigning values to pairs of (node or edge) identifiers and keys.
\end{itemize}
For convenience, we will refer to node and edge identifiers simply as the \emph{nodes} and \emph{edges} of the property graph.
Also, we use the term \emph{object} to refer to nodes, edges, and labels on specific nodes or edges.

	For example, the nodes of the property graph \(\exgraph\) depicted in Figure~\ref{figure:example-pg-intro} are \(p_1,t_1,t_2,d_1,d_2,d_3\) and the edges are \(w_1,m_1,r_1,r_2,r_3,r_4\).
	The endpoints of the edge \(w_1\) are \(p_1\) (the source) and \(t_1\) (the target).
We note that it is possible that a node or an edge does not have any labels.
In that case, \(\pglabelmap\) assigns the empty set (of labels).

A \emph{path} in a property graph \(\graph = \pgtuple\) is a non-empty alternating sequence
\(
v_0e_1v_1\ldots e_nv_n
\)
where \(v_0,\ldots,v_n\in\nodeset\) are nodes and \(e_1,\ldots,e_n\in\edgeset\) are edges such that \(\pgrelmap(e_i) = (v_{i-1},v_i)\) holds for all \(i\in\range{n}\).
For example, the sequence
\(\pgpath{p_1,w_1,t_1,r_1,d_1,r_3,d_3}\)
is a path in \(\exgraph\).
We emphasize that every path starts and ends with a node.
The \emph{length} of a path \(v_0e_1v_1\ldots e_nv_n\) is the number of edges \(n\).
An \empty{empty path} consists of exactly one node and has length~\(0\).
Two paths \(v_0e_1\ldots e_nv_n\) and \(v_0'e_1'\ldots e_m'v_m'\) can be \emph{concatenated} if \(v_n = v_0'\) holds.
The resulting path is then \(v_0e_1\ldots e_nv_ne_1'\ldots e_m'v_m'.\)
The \emph{trace} \(\pglabelmap(P)\) of a path \(P = v_0e_1v_1\ldots e_nv_n\) is \(\pglabelmap(P) = \pglabelmap(v_0)\pglabelmap(e_1)\pglabelmap(v_1)\ldots\pglabelmap(e_n)\pglabelmap(v_n)\).
For example, the trace of the path \(\pgpath{p_1,w_1,t_1,r_1,d_1,r_3,d_3}\) is
\begin{multline*}%
	\{\pglabel{\small person}\}
	\{\pglabel{\small works\_on}\}
	\{\pglabel{\small task}\}
	\{\pglabel{\small references}\}\\
	\{\pglabel{\small document},\pglabel{\small important}\}
	\{\pglabel{\small references}\}
	\{\pglabel{\small document},\pglabel{\small important}\}.
\end{multline*}

\subsection{Regular Graph Pattern Calculus}
In the following we introduce the syntax and semantics of a proper subset of GPC, a graph pattern calculus, which captures the core features of the pattern sublanguage shared by GQL and SQL/PGQ~\cite{Francis2023a}. We present our subset, which we call \emph{RGPC} (short for \emph{regular GPC}), by means of examples. %

\begin{example}\label{example:gpc-simple}
	The following RGPC pattern matches all pairs \(x,y\) of nodes such that \(x\) has the label \(\pglabel{document}\) and there is a path \(z\) from \(x\) to \(y\) on which every edge has the label \(\pglabel{references}\).
	\[
		z = \gpcnp[x]{document} \left[\gpcep{references}\right]^* (y)
	\]
\end{example}
We explain the components of the pattern in Example~\ref{example:gpc-simple}, as well as some variations, in a bottom-up fashion.
The basic building blocks of RGPC patterns are \emph{node} and \emph{edge patterns}.
A node pattern has the form \(\gpcnp[x]{\(\varphi\)}\) where \(x\) is a \emph{node variable} and \(\varphi\) is a \emph{label expression} which restricts the labels of the nodes matching the pattern.
Both \(x\) and \(\varphi\) are optional.
The node patterns in Example~\ref{example:gpc-simple} are \(\gpcnp[x]{document}\), and \((y)\).
In the latter pattern the labels of \(y\) are not restricted, the node may have any label(s).
An example for a node pattern without variable is \(\gpcnp{a}\).
If \(x\) and \(\varphi\) are both omitted, we just write \(()\).
In general, \(\varphi\) can be any propositional formula connecting labels from the set \(\labelset\) with the operators \(\land\) (conjunction), \(\lor\) (disjunction), and \(\neg\) (negation).
For example, \(\gpcnp{(a \(\land\) b) \(\lor\) \(\neg\)c}\) matches every node which has labels \(\pglabel{a}\) and \(\pglabel{b}\) or not label \(\pglabel{c}\).

Similarly to node patterns, edge patterns have the form \(\smash{\gpcep{\(\varphi\)}}\), although we do not consider variables in edge patterns.
In Example~\ref{example:gpc-simple}, there is one edge pattern, namely \(\smash{\gpcep{references}}\).

\paragraph{Path Patterns.}
Node and edge patterns can be combined into \emph{path patterns} using \emph{concatenation}, \emph{union}, and \emph{repetition}.

For example, a \emph{concatenation} of two node and an edge pattern is
\[
\gpcnp[x]{document} \gpcep{references} (y)
\]
which matches all pairs \(x,y\) of nodes such that \(x\) has the label \(\pglabel{document}\) and there is an edge labelled \(\pglabel{references}\) from \(x\) to \(y\).

The \emph{union} of two patterns can be built if they \emph{do not contain node variables}.
For instance, \(\gpcep{a} \cup \gpcep{c}\gpcep{c}\) states that there is an edge labelled \pglabel{a}, \emph{or} a path of length two on which all edges are labelled \pglabel{c}.

For \emph{repetitions} we consider the classic Kleene-star operator. For example,
\(
\gpcrep[1]{*}{%
	\gpcep{a}
}
\)
specifies that the path specified by \(\gpcep{a}\) can be repeated arbitrarily often (including \(0\) times).
We also write \(\gpcrep{+}{\gpcep{a}}\)  as shorthand for \(\gpcep{a}\gpcrep{*}{\gpcep{a}}\).
Similar to union, repetition is limited to patterns that do not contain any node
variables. Additionally, patterns below a repetition must not match a path of length \(0\).

Of course, concatenation, union, and repetition can be nested.
To improve readability (round brackets are used in node patterns) we use square brackets for disambiguation.

\paragraph{Graph Patterns.}
Finally, \emph{graph patterns} are conjunctions of path patterns, associated with pairwise distinct \emph{path variables}.\footnote{In the original definition of GPC, path variables are optional, but here they are obligatory.
They also do not have to be distinct. We demand them to be distinct to avoid \say{hidden} path equality constraints.}
For example, the following graph pattern consists of two path patterns
\[
z_a = \gpcnp[x]{b} \gpcep{a} (y),\ z_c = (y) \gpcep{c}^{+} (x)
\]
\noindent with variables \(z_a\) and \(z_c\). It specifies that there are nodes \(x, y\) in the property graph such that \(x\) is labelled \pglabel{b}, there is an edge \((x,y)\) labelled \(a\), and a path from \(y\) to \(x\) on which every edge is labelled \(c\).
Let us point out that path and node variables are disjoint.%

\begin{example}\label{example:graph-pattern}
	Consider the following graph pattern where \pglabel{p}, \pglabel{w}, \pglabel{t}, \pglabel{d}, \pglabel{r} and \pglabel{i} are shorthands for the labels \pglabel{person}, \pglabel{works\_on}, \pglabel{task}, \pglabel{document}, \pglabel{references} and \pglabel{important}, respectively.
	\[
		z =
		\gpcnp[x]{p} \gpcep{w} \gpcnp{t}
		\gpcrep{+}{%
			\gpcep{r}\gpcnp{d}
		}
		\gpcnp[y]{d \(\land\) i}
	\]
	It specifies the topology from Example~\ref{example:intro-constraint-desc}:
	person \(x\) works on a task which (possibly indirectly) references an important document~\(y\).
	Note that all nodes on the subpath from the task to \(y\) must be labelled \(\pglabel{document}\), except for the first node, representing the task.
\end{example}

\paragraph{Semantics.}
A \emph{match} for a (graph) pattern \(\query\) in a property graph \(\graph\) is a mapping \(\match\) that maps the node and path variables of \(\query\) to nodes and paths in \(\graph\), respectively,  which constitutes an answer to~\(\query\) on~\(\graph\). We refer to the literature \cite{Francis2023a} for a formal definition.%
\begin{example}\label{example:match}
	A match for the pattern in Example~\ref{example:graph-pattern} in the property graph depicted in Figure~\ref{figure:example-pg-intro} is the mapping \(\match\) with \(\mu(x) = p_1\), \(\mu(y) = d_3\), and \(\mu(z) = \pgpath{p_1,w_1,t_1,r_1,d_1,r_3,d_3}\).
\end{example} %
\section{Constraints for Property Graphs}\label{section:constraints}
With the preliminaries in place, we are now ready to introduce the constraints we study in this paper.
We first explain the idea by means of an example, and then make the terms more precise.

\begin{example}\label{example:constraint-simple-natural}
	We consider the following constraint as a running example in this section:
		\say{%
			Every document has a greater access level than any other document referenced (directly or indirectly) by it, that has not been edited since \(2020\).%
		}
\end{example}

The core ingredient of a constraint is an RGPC graph pattern,
incorporating characteristics of established graph query languages
such as GQL and SQL/PGQ, notably recursion.
This allows us to express the topology described in Example~\ref{example:constraint-simple-natural}:
The constraint affects all pairs of documents, say \(x\) and \(y\), such that \(x\) references (directly or indirectly) \(y\).
The following RGPC pattern models this topology.
\[
	z = \gpcnp[x]{document} \left[\gpcep{references}\right]^+ \gpcnp[y]{document}
\]

The RGPC pattern is complemented with two more ingredients to yield a constraint.
Both ingredients can refer to node variables.

The first can be understood as an additional filter on the matches of the pattern.
For instance, the predicate \(x \neq y\) ensures that \(x\) and \(y\) in the above pattern are matched to different documents.
Furthermore, the constraint in Example~\ref{example:constraint-simple-natural} only affects documents \(y\) which have not been edited since \(2020\).
Assuming nodes labelled \pglabel{document} have a property for the last time it was edited, i.e.\ a timestamp, this can be ensured by the predicate \(\propof{y}{edited} \le \texttt{2020-12-31}\).
In general, it is possible to supply any number of predicates (or none), in which case all of them have to hold for a match to pass the filter.

The last ingredient also consists of predicates and states which conditions have to hold, for every match of the pattern that passes the filter.
For our example, the condition is the set consisting of the predicate \(\propof{x}{access\_level} \ge \propof{y}{access\_level}\).

Overall the constraint described in Example~\ref{example:constraint-simple-natural} can be written as follows using the shorthands from Example~\ref{example:graph-pattern}.
\begin{multline*}
	z = \gpcnp[x]{d} \left[\gpcep{r}\right]^+ \gpcnp[y]{d};
	\{x \neq y,~\propof{y}{edited} \le \texttt{2020-12-31}\}\\
	\cimp
	\{\propof{x}{access\_level} \ge \propof{y}{access\_level}\}
\end{multline*}

In general, the constraints we study in this paper have the shape
\( \query; F \cimp C\)
where \(\query\) is an RGPC graph pattern, and \(F\) and~\(C\) are sets of predicates modelling the filter and conditions of the constraint, respectively.
More precisely, we consider all comparison predicates of GQL \cite[Section~19.3]{GQLStandard}:
	\(x_i = y_j\), \(x_i \neq y_j\), \(\propof{x_i}{propA} \oplus \propof{y_j}{propB}\), and \(\propof{x_i}{propA} \oplus c\) where \(x_i\) and \(y_j\) are node variables that occur in~\(q\), \(\oplus\in\{=,\neq,\le,\ge,<,>\}\), and~\(c\) is a value from~\(\valueset\).

The semantics of these predicates are as usual, with the noteworthy trait that all mentioned properties have to exist.
For example, the filter and condition in the constraint above imply the existence of the \prop{edited} and \prop{access\_level} property, respectively.
A formal definition is available in \refapp{section:definition-literal-semantics}{Appendix~A.2}.
A match \(\match\) satisfies a set of predicates if it satisfies all predicates in the set.
A property graph \(\graph\) satisfies a constraint \(\query;F \cimp C\) if every match \(\match\) in \(\graph\) for \(\query\) satisfies \(C\) or not \(F\).

\begin{example}\label{example:constraints-intro}
	The following constraint \(\constraint\) is described informally in Example~\ref{example:intro-constraint-desc}.
	\begin{multline*}
		\underbrace{%
			z =
			\gpcnp[x]{p} \gpcep{w} \gpcnp[u]{t}
			\gpcrep{+}{%
				\gpcep{r}\gpcnp{d}
			}
			\gpcnp[y]{d \(\land\) i}%
		}_{\query};\\
			\{\propof{u}{start} \le \texttt{now}\}\vphantom{\big[}
		\cimp
			\{\propof{x}{access\_level} \ge \propof{y}{access\_level}\}
	\end{multline*}
	Here \(\query\) is the RGPC pattern explained in Example~\ref{example:graph-pattern}, except for the additional node variable \(u\).
	We use \say{\texttt{now}} as a place holder for the current timestamp.
	The property graph illustrated in Figure~\ref{figure:example-pg-intro} does \emph{not} satisfy the constraint, because the match in Example~\ref{example:match} satisfies the filter (the start date of the task is in the past), but not \( \{\propof{x}{access\_level} \ge \propof{y}{access\_level}\}\).
	Indeed, \(x\) is matched to \(p_1\) and \(y\) to \(d_3\) and we have
	\(\propof{p_1}{access\_level} = 6 \not\ge 7 = \propof{d_3}{access\_level}.\)
\end{example}

As discussed in the introduction, the constraint \(\sigma\) from Example~\ref{example:constraints-intro} can be expressed as the PG-Constraint in Figure~\ref{figure:example-pg-intro}: The RGPC pattern is translated into a GQL, ASCII-art-like pattern, and the filter and consequence are expressed as \lstinline|FILTER| clauses.
Analogously, all RGPC constraints can be written as PG-Constraints.

\begin{example}\label{example:denial-constraint}
	So far our examples were implication dependencies, a subclass of denial constraints.
	Using an unsatisfiable predicate like \(x\neq x\), which we simply denote as \(\texttt{false}\), we can also express denial constraints which do not fall into this subclass.
	For instance, the following denial constraint states that there are no cyclic references between documents.
	\[
		\gpcnp[x]{d}\gpcrep{+}{\gpcep{r}\gpcnp{d}}(x); \emptyset \cimp \{\texttt{false}\}
	\]
\end{example}
\section{Repairing Property Graphs}\label{section:cleaning}

	This section is dedicated to the presentation of our repair pipeline, which is depicted in Figure~\ref{figure:pipeline}.
	We start by introducing the notion of a \emph{repair of a property graph} in Section~\ref{section:cleaning:repairing}.
	In Section~\ref{section:repair-pipeline} we present the \emph{base pipeline}, that is, our pipeline without any of the optional steps it features.
	They are introduced in Sections~\ref{section:label-repairs} and~\ref{section:neighbourhood-repairs}.%

\subsection{Property Graph Repairs}\label{section:cleaning:repairing}
As discussed in the introduction, motivated by applications where changing or adding objects is too risky due to, e.g.\ privacy or security reasons, we repair property graphs by deleting objects -- i.e.\ nodes, edges, or labels -- from it.

\begin{example}\label{example:repair-intro}
	Consider the property graph \(\exgraph\) depicted in Figure~\ref{figure:example-pg-intro} and the constraint \(\constraint\) from Example~\ref{example:constraints-intro}.
	The property graph~\(\exgraph\) does \emph{not} satisfy the constraint, because \say{Alex} works on task \(t_1\) which indirectly references, via document \(d_1\), the important document \(d_3\).
	However, \say{Alex} has a strictly smaller access level than \(d_3\).

	To repair~\(\exgraph\), under the repair model of deletions, it suffices to delete one of the nodes \(p_1, t_1, d_1\), or \(d_3\), or one of the edges \(w_1\) or \(r_1\) between them.
	Note that deleting either \(r_3\) or \(r_4\) is \emph{not} sufficient (but deleting both is).
	Moreover, all of these options will change the \emph{topology} of~\(\exgraph\).

	Another option is to delete a label of a node or edge, for instance, here it would suffice to remove the label \lstinline|:important| from \(d_3\).
	This option does not change the topology of the graph at all.
\end{example}

To properly specify repairs, we will employ \emph{subgraphs}.
Intuitively, a subgraph is a property graph obtained by deleting nodes, edges, labels (or properties) from a given property graph.

\begin{definition}[Subgraph]
	Let \(\graph = \pgtuple\) be a property graph.
	A property graph \(\graph' = (\nodeset',\edgeset',\pgrelmap',\pglabelmap',\pgpropmap')\) is called a \emph{subgraph} of \(\graph\) if the following five conditions hold~: (i) \(\nodeset'\subseteq\nodeset\),
	(ii) \(\edgeset'\subseteq\edgeset\),
	(iii) \(\pgrelmap'(o') = \pgrelmap(o') \cap \left(\nodeset'\times\nodeset'\right)\) for all \(o'\in\edgeset'\),
	(iv) \(\pglabelmap'(o') \subseteq \pglabelmap(o')\) for all \(o'\in\nodeset'\cup\edgeset'\) and
    (v) If \(\pgpropmap'(o',k)\) is defined for a property key \(k\) and \(o'\in\nodeset'\cup\edgeset'\), then so is \(\pgpropmap(o', k)\) and \(\pgpropmap'(o',k) = \pgpropmap(o',k)\).
	\noindent
	That is, all nodes, edges, labels, and properties (as well as their values) of \(\graph'\) are also present in \(\graph\).
	If all nodes and edges in \(\graph'\) have exactly the same labels and property key/value pairs as in \(\graph\), i.e.\ if \(\pglabelmap'(o') = \pglabelmap(o')\) and \(\pgpropmap'(o',k) = \pgpropmap(o',k)\), for all \(o'\in\nodeset'\cup\edgeset'\) and property keys \(k\), then we call \(\graph'\) a \emph{topological subgraph} of \(\graph\).
	We have that \(\graph'\) is a \emph{proper subgraph} of \(\graph\), if \(\graph \neq \graph'\) holds.
\end{definition}

For example, removing the label \pglabel{:important} from node \(d_3\) in the property graph \(\exgraph\) depicted in Figure~\ref{figure:example-pg-intro} results in a proper subgraph, say~\(\exgraph^\star\).
The subgraph \(\exgraph^\star\) is \emph{not} topological because the node \(d_3\) is also present in the original graph \(\exgraph\) but the labels of \(d_3\) in \(\exgraph^\star\) are not the same as in \(\exgraph\).
Another example is the topological subgraph obtained from
\(\exgraph\)
by deleting the edges \(r_3\) and \(r_4\).

For being a \emph{repair}, a subgraph has to satisfy all constraints, and, in the spirit of previous work \cite[see, e.g.][Definition~2.2]{DBLP:journals/iandc/ChomickiM05}, it should be maximal, i.e. there should be no unnecessary deletions.
\begin{definition}[Repair]\label{definition:repair}
	Given a property graph \(\graph\) and a set \(\constraintset\) of constraints, a \emph{repair} of \(\graph\) is a subgraph \(\graph'\) of \(\graph\) that
	\begin{enumerate}[label={\alph*)}]
		\item satisfies \(\constraintset\); and\label{def:repair-satisfaction}
		\item is \emph{maximal}, that is, there is \emph{no} subgraph \(\graph''\) of \(\graph\) such that \(\graph'\) is a proper subgraph of \(\graph''\) and \(\graph''\) satisfies \(\constraintset\).\label{definition:repair-cond-b}\label{def:repair-maximality}
	\end{enumerate}
	When a subgraph \(\graph_a\) of \(\graph\) satisfies Condition~\ref{def:repair-satisfaction} but not necessarily Condition~\ref{def:repair-maximality}, we call \(\graph_a\) an \emph{approximate} repair.
	A \emph{topological (approximate) repair} is defined analogously to a(n approximate) repair, except that all subgraphs involved are topological subgraphs.
\end{definition}
\begin{example}[Continuation of Example~\ref{example:repair-intro}]\label{example:repair}
	Let \(\exgraph'\) be the subgraph of the property graph \(\exgraph\) obtained by deleting the edge \(r_1\) between \(t_1\) and \(d_1\).
	We observe that \(\exgraph'\) is a topological repair.
	But it is \emph{not} a repair, because the subgraph \(\exgraph''\) obtained from \(\exgraph\) by removing (just) the \lstinline|:references| label from \(r_1\) is also a repair, and \(\exgraph'\) is a subgraph of \(\exgraph''\).
	Thus, Condition~\ref{definition:repair-cond-b} of Definition~\ref{definition:repair} is violated.
	It is, however, an approximate repair, since it satisfies Condition~\ref{def:repair-satisfaction}.

	Similarly, removing the node \(d_3\) is not a repair, since it suffices to remove either the \lstinline|:document| or the \lstinline|:important| label.
	Observe that removing a node always implies removing all edges involving this node as well: otherwise, the result is not a well-defined property graph.
	Thus, the removal of \(d_3\) leads to the removal of the edges \(r_3\) and \(r_4\).
	But then the resulting subgraph is also \emph{not} a topological repair: it is possible to add back \(d_3\), as long as the edges are not added back.
	We discuss this effect in more detail later on.
\end{example}

\subsection{The Base Pipeline}\label{section:repair-pipeline}
We are now ready to present our repair pipeline, which is illustrated in Figure~\ref{figure:pipeline}.
It takes a set \(\constraintset\) of RGPC constraints and a property graph \(\graph\) as input, and has~\(6\) steps, the last of which carries out the deletion operations.
Two of these steps are optional, namely steps~\(2\) and~\(3\): they lead to label removals and a controllable trade-off between quality and runtime, respectively.
In the following, we first discuss the \emph{base pipeline}, consisting of steps~\(1\),~\(4\),~\(5\), and~\(6\).
\begin{figure*}
	\includegraphics[width=\textwidth]{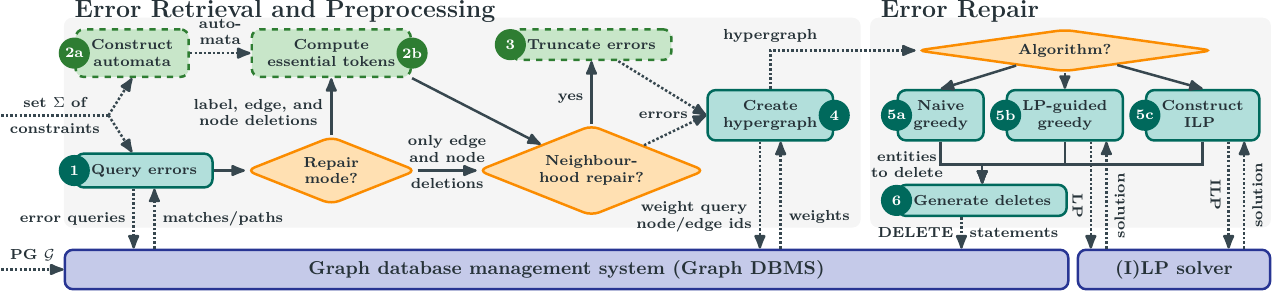}
	\caption{%
		Our repair pipeline for property graphs consists of 6 steps, where steps 2 and 3 are optional.
		2a and 2b are either both enabled or not, while 5a, 5b and 5c are alternatives.
		Solid edges indicate control flow, while dotted edges indicate communication.
	}
	\label{figure:pipeline}
\end{figure*}

\paragraph{Error Retrieval (Step~1)}
Intuitively, we understand an error as a set of objects witnessing that a constraint is not satisfied by the property graph.
Since all path patterns of an RGPC constraint are associated with a path variable, an error can be obtained from the range of a match which satisfies the filter but not the consequence of the constraint.

Consider again the constraint from Example~\ref{example:constraints-intro} and the property graph depicted in Figure~\ref{figure:example-pg-intro}.
The property graph does not satisfy the constraint, because there are two matches of the constraint's pattern which pass the filter and do not satisfy the condition of the constraint:
The first match corresponds to the path \(\pgpath{p_1,w_1,t_1,r_1,d_1,r_3,d_3}\) (cf.\ Example~\ref{example:match}), and the second one to the path \(\pgpath{p_1,w_1,t_1,r_1,d_1,r_4,d_3}\).
As illustrated in Figure~\ref{figure:errors-to-hg}~(\(1\)), (\(2\)), the sets of nodes and edges occurring on these paths, namely \(O_1 = \pgerr{p_1,w_1,t_1,r_1,d_1,r_3,d_3}\) and \(O_2 = \pgerr{p_1,w_1,t_1,r_1,d_1,r_4,d_3}\) constitute \emph{topological errors}.

\begin{definition}[Topological Error]
	Let \(\graph = \pgtuple\) be a property graph, \(\query; F \cimp C\) be a constraint with path variables \(z_1,\ldots,z_k\), and \(\match\) be a match for \(\query\) in \(\graph\) that witnesses \(\graph\) \emph{not} satisfying the constraint.
	The set \(O_\match\), consisting of all nodes and edges, $o \in \nodeset\cup\edgeset$, for which there is an $i$, \(1\le i\le k\), such that $o$ occurs in \(\match(z_i)\), is called a \emph{topological error} of the constraint in \(\graph\).
\end{definition}

Error detection by means of graph queries has already been introduced in previous work \cite{Angles2021}.
We put this approach into practice and delegate error detection and retrieval to a graph database system by rewriting every constraint into an \emph{error query}, which asks for all paths involved in an error.
For example, the constraint from Example~\ref{example:constraints-intro} is rewritten into the following GQL error query.
\begin{lstlisting}[numbers=none]
MATCH z = (x:person)-[:works_on]->(u:task)
          (-[:references]->(:document))+
          (y:document & important)
FILTER u.start IS NOT NULL AND u.start <= NOW()
AND (x.access_level < y.access_level
     OR x.access_level IS NULL OR y.access_level IS NULL)
RETURN z
\end{lstlisting}

\paragraph{Hypergraph Construction (Step~4)}
For repairing the property graph, the idea is to pick one object from each error and delete it.
Since every error corresponds to a match, deleting a single object from it suffices to eliminate the match; hence, the resulting property graph satisfies the constraint(s).
However, the outcome is not necessarily a repair.
For instance, selecting \(r_1\) and \(r_4\) from the errors \(O_1\) and \(O_2\) discussed above, respectively, does not yield a repair, because it is possible to add \(r_4\) back without violating the constraint.
This is because \(r_1\) occurs in both errors.
We already discussed in Example~\ref{example:repair} that removing nodes can have a similar effect.

To address this issue it helps to understand a collection of (topological) errors as a hypergraph, called \emph{conflict hypergraph}.
A \emph{hypergraph} is a pair \(\hgraph = (\hnodeset,\hedgeset)\) consisting of a set \(\hnodeset\) of vertices\footnote{We call them vertices to differentiate them from the nodes of property graphs.} and a set \(\hedgeset\subseteq 2^{\hnodeset}\) of hyperedges.
The conflict hypergraph for our running example is depicted in Figure~\ref{figure:errors-to-hg}~(\(3\)).
\begin{definition}[Topological Conflict Hypergraph]
	Given a property graph \(\graph = \pgtuple\), a set \(\constraintset\) of constraints,
	the \emph{topological conflict hypergraph} of \(\graph\) w.r.t.\ \(\constraintset\) is the hypergraph \(\hgraph = (\hnodeset, \hedgeset)\) whose edges are precisely all topological errors of all constraints in \(\constraintset\) in \(\graph\), and where
	\(\hnodeset =  \bigcup_{O\in \hedgeset} O\) is the union of all these edges.
\end{definition}

\begin{figure}
	\centering
	\includegraphics[width=.9\linewidth]{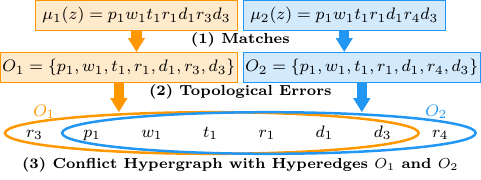}
	\caption{Different encodings of topological errors}\label{figure:errors-to-hg}
\end{figure}

\paragraph{Computing a Repair (Step~5c)}
	To repair a property graph, the idea is to select vertices from the conflict hypergraph with high-inbetweenness, i.e.\ vertices which participate in a high amount of errors.
	At the same time we want to select a minimal number of vertices.
This amounts to computing a minimal vertex cover:%

A \emph{vertex cover} of a hypergraph \(\hgraph = (\hnodeset,\hedgeset)\) is a set \(V\subseteq \hnodeset\) such that \(V\cap O \neq \emptyset\), for each \(O\in\hedgeset\).
A vertex cover \(V\) is minimal if there is no vertex cover \(V'\) with \(V'\subsetneq V\).

Intuitively, a minimal vertex cover tells us which objects to delete to obtain a repair.
For example, a vertex cover for the conflict hypergraph with hyperedges \(O_1\) and \(O_2\) from above is \(V = \{r_1\}\).
Trivially, no subset of \(V\) is a vertex cover, and removing \(r_1\) yields a repair, as discussed above.
The set \(V \cup \{r_4\} = \{r_1, r_4\}\) is not a minimal vertex cover, since it contains \(V\).
And indeed, removing \(r_1\) and \(r_4\) does not yield a repair for the same reason: \(r_4\) can be added (back).

However, this does not address the issue of node removal implying removal of incident edges.
For instance, \(\{d_3\}\) is a minimal vertex cover, but as discussed in Example~\ref{example:repair}, removing \(d_3\) implies the removal of \(r_3\) and  \(r_4\), and does not yield a repair.
	Therefore, the deletion of nodes should be avoided, unless it is necessary, e.g.\ because an error consists only of nodes.
We realize this by assigning a weight to each vertex of the conflict hypergraph \(\hgraph\).
The weight \(\weightof{n}\) of a node \(n\) of the property graph is the number of edges incident to \(n\) plus one.
The weight \(\weightof{e}\) of an edge \(e\) is simply \(1\).
We are then interested in a \emph{minimum weight vertex cover} \cite[see e.g.][]{Xiao2024} of~\(\hgraph\).
The weight of a vertex cover \(V\) is \(\weightof{V} = \sum_{v\in V}\weightof{v}\).
A \emph{minimum weight vertex cover}~\(V\) is a vertex cover whose weight~\(\weightof{V}\) is the minimum among all vertex covers (of \(\hgraph\)).
Note that, if all vertices have the same weight, minimum weight vertex covers coincide with minimum vertex covers\footnote{\(V\) is a minimum vertex cover, if there is no vertex cover \(V'\) with \(\abs{V'} < \abs{V}\).}, which is a stronger notion than that of a minimal vertex cover.
Each minimum vertex cover is in turn a minimal vertex cover (but not vice versa).
We have the following result, whose proof is available in \refapp{section:error-repairing-proofs}{Appendix~B.2}.

\begin{proposition}\label{result:topological-repair-ind-set}
	Given a minimum weight vertex cover \(V\), removing\footnote{Recall that removing a node implies removing all its incident edges in~\(G\).} all objects in \(V\) from \(\graph\) yields a topological repair of \(\graph\).
\end{proposition}

The following integer linear program (ILP) with integer variables \(x_v\), for each vertex \(v\in\hnodeset\), encodes the problem of finding a minimum weight vertex cover.
\begin{align*}
	\textnormal{minimize } \sum_{v\in \hnodeset}\weightof{v}\cdot x_v
	&&\textnormal{subject to } \sum_{v\in O} x_v \ge 1 \; \textnormal{ for all } O\in\hedgeset\\
	&&\textnormal{and } x_v \in \{0, 1\} \; \textnormal{ for all } v\in\hnodeset
\end{align*}
Given a solution for this ILP, that is, a function that assigns all variables \(x_v\) into \(\{0,1\}\), corresponds to a minimum weight vertex cover \(V\): It consists of all vertices \(v\) for which \(x_v\) is mapped to \(1\).

Overall, a repair can be computed by rewriting every constraint into an error query, querying for errors, and then solving the ILP above.
Since there are at most exponentially many hyperedges and integer linear programming is in \(\NP\), the complete procedure runs in non-deterministic exponential time (data complexity).

\paragraph{Greedy Repairs (Steps~5a and~5b).}
Despite ILP solvers being reasonably fast in practice, it can be desirable to trade the quality of a repair -- in our case, the number of objects deleted -- for a better runtime.
A common approach to achieve this are so called greedy algorithms, of which we consider two.

The first we call \emph{naive greedy algorithm}.
It computes a vertex cover \(V_g\) with a small weight in two phases: the \emph{selection} and the \emph{trimming} phase.
In the selection phase the algorithm iterates over all hyperedges to select candidates for the vertex cover: If the current hyperedge contains a vertex, which has already been selected and has minimal weight among the vertices in the hyperedge, the hyperedge is skipped.
Otherwise, the algorithm selects an arbitrary vertex from the hyperedge with minimal weight among the vertices in the hyperedge.
The selected vertices form a vertex cover \(V_g\), because every hyperedge contains at least one selected vertex.

However, the algorithm may end up selecting too many vertices, and thus \(V_g\) might not be \emph{minimal}.
Thus, this phase yields an approximate repair.
To obtain a repair, the algorithm attempts, in the trimming phase, to reduce the number of selected vertices again.
It does this by checking, for every selected vertex \(v\), ordered by weight in descending order, whether there is a hyperedge \(E\) for which \(v\) is the only selected vertex, i.e.\ for which \(V_g \cap E = \{v\}\) holds.
If there is \emph{no} such hyperedge \(E\), then \(v\) has to be removed from~\(V_g\).

Consider again the two errors \(O_1 = \pgerr{p_1,w_1,t_1,r_1,d_1,r_3,d_3}\) and \(O_2 = \pgerr{p_1,w_1,t_1,r_1,d_1,r_4,d_3}.\)
The naive greedy algorithm might compute \(V_g = \{r_3, r_4\}\) in the first phase.
In the second phase, \(V_g\) is not changed, because \(r_3\) and \(r_4\) do not occur in \(O_2\) and \(O_1\), respectively.
Thus, the vertex cover \(V_g\) is minimal because neither \(r_3\) nor \(r_4\) can be removed.
However, it is not a minimum weight vertex cover, since, e.g., the minimum vertex cover \(\{r_1\}\) has a smaller weight.\footnote{It is also possible that the naive greedy algorithm picks \(r_1\) instead of \(r_3\) and \(r_4\) in which case it would yield the minimum weight vertex cover \(\{r_1\}\).}
While the naive greedy algorithm does not necessarily compute a minimum weight vertex cover, picking vertices with minimal weight is sufficient to compute a (topological) repair.
Since the algorithm iterates over all errors, it runs in exponential time.

\begin{proposition}\label{result:topological-repair-greedy}
	Given a minimal vertex cover \(V_g\) computed by the naive greedy algorithm, removing all objects in \(V_g\) from \(\graph\) yields a topological repair of \(\graph\).
\end{proposition}
A proof is available in \refapp{section:error-repairing-proofs}{Appendix~B.2}.

As for the second greedy algorithm, we consider an \emph{LP-guided greedy algorithm} which combines a solution for the \emph{LP-relaxation} of the ILP with the naive greedy algorithm.
Here the LP-relaxation is the linear program (LP) obtained from the ILP by replacing the conditions \(x_v\in\{0,1\}\) with inequalities \(0\le x_v \le 1\) and allowing the \(x_v\) to assume real values.
It is well-known that LPs can be solved in polynomial time, unlike ILPs.
The algorithm then selects any vertex \(v\) for which \(x_v > 0\) holds as a candidate, and then proceeds as the naive greedy algorithm.

\subsection{Optional Step~2: Label Deletions}\label{section:label-repairs}%
We continue with the description of optional Step~\(2\) in Figure~\ref{figure:pipeline} of our repair pipeline, which allows us to obtain non-topological repairs by deleting labels (in addition to nodes and edges).

Our repair pipeline -- whether the ILP or a greedy algorithm is used -- deletes nodes only if errors contain isolated nodes (aka paths of length \(0\)).
Otherwise, an edge is deleted, since edges have minimal weight \(1\).
This can lead to the presence of (potentially many) isolated nodes in a repair, when all incident edges of a node are deleted.
To avoid this case, we study repairs which are not necessarily topological.
More precisely, we extend our approach to include the removal of labels to obtain a repair, as discussed in Example~\ref{example:repair-intro}.
Instead of deleting (many) incident edges of a node, it can suffice to remove a single label of a node.
Note that it is not always possible to obtain a repair by only removing labels, i.e., if a constraint does not refer to any label.
Thus, it may still be necessary to delete nodes and/or edges, in addition to deleting labels.

The idea is to extend topological errors \(O_\match\), which initially contain the nodes and edges of the property graph occurring in the co-domain of a match \(\match\), by pairs \((o,\ell)\) of objects \(o\), i.e.\ a node or an edge, and labels \(\ell\).
We call these pairs \((o,\ell)\) \emph{tokens}.
The intended meaning is that the error can be repaired by removing the label \(\ell\) from node (or edge) \(o\).
The edges of the (not necessarily topological) conflict hypergraph are then all extended sets -- which we will simply call \emph{errors} in the following.
We note that, in general, there are multiple extensions of a single topological error.
In particular, there are more errors than topological errors.

In the following we discuss how to compute errors given a match.
We start with the simple case of a single constraint with a single path pattern.
Consider once again the constraint from Example~\ref{example:constraints-intro}.
Recall that it consists of one path pattern with path variable \(z\).
We can observe that removing the \(\pglabel{person}\) label from the node \(p_1\), or removing the \(\pglabel{works\_on}\) from the edge \(w_1\) yields a repair.
Indeed, every path matching the pattern of the constraint has to start with a node labelled \(\pglabel{person}\) followed by an edge labelled \(\pglabel{works\_on}\), and \(p_1\) and \(w_1\) are the only objects with these labels.
Thus, \((p_1,\pglabel{person})\) is a token we are looking for.
The same is true for \((w_1,\pglabel{works\_on})\).
We say that these tokens are \emph{essential}.

\paragraph{Automata Construction (Step~2a).}
To compute essential tokens automatically for arbitrary RGPC path patterns, we construct automata,
which we call \emph{RGPC automata} and introduce next, by means of an example.
Consider the path pattern
\[
	\gpcnp[x]{p} \gpcep{w} \gpcnp[u]{t}
	\gpcrep{+}{%
		\gpcep{r}\gpcnp{d}
	}
	\gpcnp[y]{d \(\land\) i}
\]
from Example~\ref{example:constraints-intro}.
The RGPC automaton for this path pattern is illustrated in Figure~\ref{figure:example-automaton}.
\begin{figure}%
	\includegraphics[width=\columnwidth, trim=.25cm 0cm 0cm .1cm, clip]{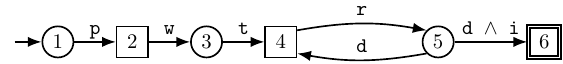}%
	\caption{%
		Automaton for the RGPC pattern from Example~\ref{example:constraints-intro}%
	}\label{figure:example-automaton}%
\end{figure}
Note that the transitions are labelled with label expressions, i.e.\ propositional formula connecting labels -- instead of just (single) labels.
For instance, the transition from state~\(5\) to~\(6\) is labelled with \(\pglabel{d}~\land~\pglabel{i}\), which is the label expression of the node pattern \(\gpcnp[y]{d \(\land\) i}\).
Thus, RGPC automata belong to the family of so called \emph{symbolic} automata \cite[see, e.g.,][]{Cimatti2011}.
A transition can be taken if the label expression is satisfied by the \emph{set} of labels of a node or edge.
For instance, the transition with the formula \texttt{d~\(\land\)~i} actually represents several transitions:  one for each set of labels that contains~\(\pglabel{d}\) and~\(\pglabel{i}\).
Notably, the transition from State~\(5\) to State~\(4\) can also be taken if \(\pglabel{d}\) and \(\pglabel{i}\) are both present.
This allows the automaton to read traces of paths, e.g.\ the trace \(
	\{\pglabel{p}\}
	\{\pglabel{w}\}
	\{\pglabel{t}\}
	\{\pglabel{r}\}
	\{\pglabel{d},\pglabel{i}\}
	\{\pglabel{r}\}
	\{\pglabel{d},\pglabel{i}\}
\) of the path \(\pgpath{p_1,w_1,t_1,r_1,d_1,r_3,d_3}\) from the end of Section~\ref{section:property-graphs}.

Furthermore, unlike classical automata, the RGPC automaton illustrated in Figure~\ref{figure:example-automaton} has two kinds of states: Those drawn as a circle and those drawn as a square, indicating that the next transition reads the label(s) of a node or the label(s) of an edge, respectively.
For instance, the outgoing transition of State~\(1\) originates from the node pattern \(\gpcnp{p}\), and the outgoing transition of State~\(2\) from the edge pattern \(\gpcep{w}\).
In particular, the RGPC automaton reads node and edge labels in alternation, and can thus read the trace of a path.

Finally, note that variables of the path pattern do not play a role for the automaton construction: We use RGPC automata to rescan paths returned during the error retrieval step (Step~1).
We thus already know that such a path is a match, and we are only interested in identifying labels for deletions.

The formal RGPC automata construction is available in \refapp{section:automata-construction}{Appendix~B.1}.%

\begin{algorithm}\small
	\caption{Computing Errors with Essential Tokens for a Single Path}\label{algo:step-3-single-path}
	\SetKw{Yield}{yield}
	\KwIn{Match \(\match\), path pattern \(z = \mathcal{P}\), \(\match(z) = v_0e_1v_1\ldots e_nv_n\)}
	\KwOut{Set of errors \(\mathcal{E}_\mu\) with essential tokens}

	\(\autA \gets\) automaton for \(\mathcal{P}\)\;
	\(O_\match \gets \{v_0,e_1,v_1,\ldots,e_n,v_n\}\) \tcp*{topological error}
	\(\mathcal{E}_\mu \gets \emptyset \) \tcp*{set of errors}
	\ForEach{accepting run \(r\) of \(\mathcal{A}\) on path \(\match(z)\)}%
	{%
		\(O_r \gets \emptyset\)\;
		\ForEach{o \(\in O_\match\) and label \(\ell\)}%
		{
			\If{removing \(\ell\) from \(o\) invalidates \(r\)}%
			{
				\(O_r \gets O_r \cup \{(o, r)\}\) \tcp*{\((o, \ell)\) is essential}
			}
		}
		\(\mathcal{O}_{\match, r} \gets O_{\match} \cup O_r\);
		\(\mathcal{E}_\mu \gets \mathcal{E}_\mu \cup \{\mathcal{O}_{\match, r}\}\)\;
	}
	\Return{\(\mathcal{E}_\mu\)}
\end{algorithm}

\paragraph{Computing Essential Tokens (Step~2b)}
Our algorithm for computing errors for a single path pattern \(\mathcal{P}\) is outlined in Algorithm~\ref{algo:step-3-single-path}.
It uses the automata constructed in Step~2a to identify essential tokens, as follows:
Let \(z\) be the path variable of \(\mathcal{P}\) and \(\autA\) be the automaton constructed for it, and \(\match\) be a match of \(\mathcal{P}\).
Then \(\autA\) accepts the trace of the path \(\match(z)\).
To obtain essential tokens, we consider every accepting run \(r\) of \(\autA\) on the trace of \(\match(z)\): if removing a label \(\ell\) from an object \(o\) on the path \(\match(z)\) invalidates the run, it is \emph{essential for this run}.
For example, the trace \(
\{\pglabel{p}\}
\{\pglabel{w}\}
\{\pglabel{t}\}
\{\pglabel{r}\}
\{\pglabel{d},\pglabel{i}\}
\{\pglabel{r}\}
\{\pglabel{d},\pglabel{i}\}
\) is accepted by the automaton depicted in Figure~\ref{figure:example-automaton}.
Removing the label \(\pglabel{i}\) from the last set in the trace invalidates all accepting runs, because the automaton can no longer take the transition from state~\(5\) to~\(6\) after reading the second \(\{\pglabel{r}\}\) (and if the transition is taken earlier the automaton cannot read the whole trace).

For invalidating the match \(\match\), all accepting runs have to be invalidated.
Thus, for each accepting run \(r\), we extend \(O_\match\) to the error \(O_{\match, r}\) by adding all tokens which are essential for \(r\).
Note that, if the pattern does not contain any labels, there are no essential tokens for any run.
In that case, the original error \(O_\match\) is kept as is.

Suppose now the constraint consists of a graph pattern \(q\) with \(k\) path patterns, with path variables \(z_1,\ldots,z_k\).
Let \(\match\) be a match witnessing that the constraint is not satisfied.
Let further \(\autA_i\) be the automaton for pattern \(i\in\range{k}\),
and \(r_1,\ldots,r_k\) be accepting runs of \(\autA_1,\ldots,\autA_k\) on the traces of the paths \(\match(z_1), \ldots, \match(z_k)\).
While it suffices to invalidate one of these runs, to invalidate \emph{the combination} \(r_1,\ldots,r_k\), we observe that there might be other combinations involving the same runs:
Say we invalidated an accepting run \(r_j\) by deleting some label.
Then the runs \(r_1,\ldots,r_j,\ldots, r_k\) do no longer witness that \(\match\) is a match.
However, there might be another run \(r_j'\) such that \(r_1,\ldots,r_j',\ldots, r_k\) still witness that \(\match\) is a match.

We thus create an error for every combination \(r_1,\ldots,r_k\) of accepting runs.
More precisely, we condense the sets \(O_{\match, {r_1}},\ldots,O_{\match, {r_k}}\) into a single error \(O_{\match, {r_1},\ldots, {r_k}} = O_{\match, {r_1}} \cup \dots \cup O_{\match, {r_k}}\).
The sets \(O_{\match, {r_1},\ldots, {r_k}}\), for each match \(\match\) and all combinations of runs \(r_1,\ldots, r_k\), are then the hyperedges of the \emph{conflict hypergraph}.
This way all combinations have to be invalidated by removing (at least) one object.

Since removing a node or edge implies the removal of all its labels, we adapt the weights accordingly:
For an edge, the weight is the number of its labels plus \(1\), and for a node the sum of the number of its labels, the weights of its incident edges, and \(1\) (for itself).
A token \((o,\ell)\) represents a single label, and thus has weight~\(1\).

We emphasize that Proposition~\ref{result:topological-repair-ind-set} does \emph{not} carry over, since the formulas occurring in an RGPC pattern can contain negations:
Thus, deleting labels can lead to new matches, and hence, to new errors.
We call an RGPC constraint \emph{positive} if \emph{no} label expression in its pattern contains any negation.
We then have the following. %

\begin{proposition}
	A repair with label removals can be computed in non-deterministic exponential time (data complexity), if the given RGPC constraints are positive.
\end{proposition}

\subsection{Optional Step~3: Neighbourhood Errors}\label{section:neighbourhood-repairs}
	The optional Step~3 attempts to reduce the runtime by limiting the size of errors, and thus also the number of errors.
	As a trade-off the repair pipeline yields, in general, an approximate repair.

The size of errors does depend on the size of the considered property graph, and the number of errors can be exponential.
We address this problem by introducing neighbourhood errors: In a nutshell, the idea is to repair a property graph by removing objects \say{close} to the endpoints of the paths inducing errors.

For an integer \(k \ge 1\), which can intuitively be understood as a radius around the endpoints of a path,  the \emph{\(k\)-endpoint neighbourhood} of a path \(v_0e_1v_1e_2\ldots e_nv_n\) consists of the objects which have distance at most \(k\) from one of the endpoints \(v_0\) and \(v_n\), that is, the objects \(v_0,e_1,v_1,\ldots,e_k,v_k\) and \(v_{n-k},e_{n+1-k},v_{n+1-k},\ldots e_n,v_n\).

Given an RGPC constraint \(\query; F \cimp C\) with \(\ell\) path variables \(z_1,\ldots,z_\ell\), and a match \(\mu\) for \(\query\) witnessing that the constraint is \emph{not} satisfied, the \emph{topological \(k\)-neighbourhood error} \(O_\match^k\) consists of all nodes and edges \(o\), for which there is a path variable \(z_i\), \(1\le i \le \ell\), such that \(o\) occurs in the \(k\)-endpoint neighbourhood of \(\mu(z_i)\).

When Step~3 is enabled, the pipeline uses (topological) \(k\)-neigh\-bour\-hood errors in place of (topological) errors.
For example, consider once more the constraint from Example~\ref{example:constraints-intro}, the property graph depicted in Figure~\ref{figure:example-pg-intro}, and recall that there are two matches witnessing that the property graph does not satisfy the constraint:
The first match corresponds to the path \(\pgpath{p_1,w_1,t_1,r_1,d_1,r_3,d_3}\) (cf.\ Example~\ref{example:match}), and the second one to the path \(\pgpath{p_1,w_1,t_1,r_1,d_1,r_4,d_3}\).
The topological \(1\)-neighbourhood errors are \(O_1^{1} = \pgerr{p_1,w_1,t_1}\cup\pgerr{d_1,r_3,d_3}\), and \(O_2^{1} = \pgerr{p_1,w_1,t_1} \cup \pgerr{d_1,r_4,d_3}\).
In contrast to the (full) topological errors \(O_1 = \pgerr{p_1,w_1,t_1,r_1,d_1,r_3,d_3}\) and \(O_2 = \pgerr{p_1,w_1,t_1,r_1,d_1,r_4,d_3}\), the edge \(r_1\) is missing from \(O_1^{1}\) and \(O_2^1\).
Thus, when using \(O_1^{1}\) and \(O_2^1\) instead of \(O_1\) and \(O_2\), the repair obtained by deleting \(r_1\) is not found.
On the other hand, the errors are smaller, and in general their size is bounded linearly in \(k\), and not in the size of the property graph.

\newcommand{\plotscale}{.4}%
\section{Experiments}\label{section:experiments}
To showcase that our approach is applicable in practice we conducted several experiments:
We first assess the base pipeline in terms of runtimes and scalability with respect to the number of errors and constraints, and compare the performances and quality of the ILP and both greedy algorithms.
We then study the effects of the optional steps of the repair pipeline on quality and runtime. %
Finally, we conduct a qualitative study aiming at studying the behaviour of the pipeline on a real-world dataset with real-world constraints.
For the evaluation, we use the four property graphs of different sizes listed in Table~\ref{table:static-pgs}.
The top three are real-world property graphs, and the LDBC property graph is generated with the corresponding benchmark.
We implemented our repair pipeline \cite{artifacts}, illustrated in Figure~\ref{figure:pipeline}, in Python~3.13.
For solving (I)LPs, we use HiGHS \cite{DBLP:journals/mpc/HuangfuH18}.
As graph database system, we use Neo4j~5.26 (community edition).
All experiments were run on a virtual machine with~8 physical cores of an
Intel(R) Xeon(R) Silver 4214 CPU \cite{intel-arc-cpu-spec} and 125GiB RAM. %

\begin{table*}
	\caption{Property graphs used for experiments, number \(\lvert\Sigma\rvert\) of constraints, and number of errors}\label{table:static-pgs}
	\small%
	\begingroup%
	\newcommand{\consStatHeader}{\hspace{.9em}{\footnotesize \(\lvert\Sigma\rvert\)} & {\footnotesize \#errors}}%
	\centering%
	\begin{tabular}{l r r r r r r r r r r r r}
			\toprule
			\multirow{2}{*}{Name} & \multirow{2}{*}{\#Nodes} & \multirow{2}{*}{\#Edges} & \multicolumn{2}{c}{\(1\)-way} & \multicolumn{2}{c}{\(2\)-rep} & \multicolumn{2}{c}{\(2\)-way} & \multicolumn{2}{c}{loop} & \multicolumn{2}{c}{\(3\)-split}\\
			&  &  & \consStatHeader & \consStatHeader & \consStatHeader & \consStatHeader & \consStatHeader\\
			\toprule
			ICIJ Property Graph \cite{InvestigativeJournalists2025}  & 2 016 523 & 3 339 267 & 14 & 558799 & 19 & 397108 & 39 & 7553 & 188 & 56614 & 152 & 876933\\
			Italian Legislative Graph \cite{Colombo2024,colombo_2024_11210265} & 411 787 & 1 117 528 & 20 & 166977 & \multicolumn{2}{c}{--} & 53 & 443695 & 61 & 3783 & \multicolumn{2}{c}{--} \\
			Coreutils Code Property Graph \cite{Yamaguchi2014} & 206 506 & 387 284 & 20 & 45728 & \multicolumn{2}{c}{--} & 52 & 3707 & \multicolumn{2}{c}{--} & \multicolumn{2}{c}{--} \\
			\midrule
			LDBC, scaling factor \(0.3\) \cite{DBLP:journals/sigmod/AnglesBLF0ENMKT14,DBLP:journals/corr/abs-2001-02299} & 940 322 & 5 003 201 & 5 & 79064 & 6 & 72756 & \multicolumn{2}{c}{--} & 56 & 423 & 4 & 2734499\\
			\bottomrule
		\end{tabular}
		\endgroup%
\end{table*}

\subsection{Assessing the Base Pipeline}\label{section:performance-quality}
In this section, we assess the performance and quality of the \emph{base pipeline}, that is our repair pipeline without any optional steps.
For this purpose, we generated constraints with pattern shapes as outlined in Figure~\ref{figure:pattern-shapes} for the property graphs in Table~\ref{table:static-pgs}.
These pattern shapes cover most shapes that have been observed in real-life query logs \cite{Bonifati2020} while also exhibiting recursion.
Notably, the \(1\)-way shape alone covers \(74.61\%\) of the valid path patterns observed \cite{Bonifati2020}, and the loop shapes include triangles, since the orange dashed shape can take the form of a pattern which matches paths of length~\(2\).
An instance of a \(1\)-way pattern with~\(3\) edge and~\(4\) node patterns is \(z = \gpcrep{+}{\gpcnp{a}\gpcep{c}\gpcnp{b}\gpcep{e}\gpcnp{d}\gpcep{f}\gpcnp{h}}\).
A constraint then consists of a pattern combined with an empty filter and the condition \(\{\texttt{false}\}\).
This maximizes the number of errors for our scaling experiments.
Moreover, to increase the number of constraints (and errors), we added, to each graph, an additional \(10\%\) of edges -- labelled with a special label -- randomly between nodes, which were already connected by an edge in the original graph, to preserve the topology.
We emphasize that there is no correlation between the additional edges and errors: they can partake in an error or not, like other edges.
	Their sole purpose is to \emph{increase} the number of paths, and hence errors.
The number of constraints and the total number of errors, for each shape and property graph, are listed in Table~\ref{table:static-pgs}.
	Note that no ground truth is available.
	Instead, the given graphs are regarded as dirty, and our repairs are clean by definition.
	The repairs obtained by the ILP algorithm serve as best possible repairs.

We compare the performance and quality of the two greedy and ILP algorithms.
We use runtimes\footnote{We do not measure error detection via queries, which heavily depends on the database system and is not the focus of our paper.} as performance, and the number of deletions as quality measure -- a smaller number of deletions means a higher quality since it corresponds to less information loss.
Note that measures like the \(F_1\)-score for quality are not meaningful: since our holistic approach always repairs all errors, there are no false negatives, and by definition a repair cannot contain false positives.
Hence, the \(F_1\)-score is always \(1\) (or \(0\) in case of a timeout).%

\begin{figure}
	\includegraphics[width=\linewidth]{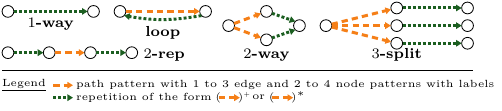}
	\caption{Shapes of RGPC patterns used in Sections~\ref{section:performance-quality} and~\ref{section:ablation}}\label{figure:pattern-shapes}
\end{figure}

\newcommand{\includePlot}[3]{%
	\centering%
	\includegraphics[scale=\plotscale]{#2-#3.pdf}
	\caption{#1}\label{figure:#2-#3}
}

\begin{figure*}
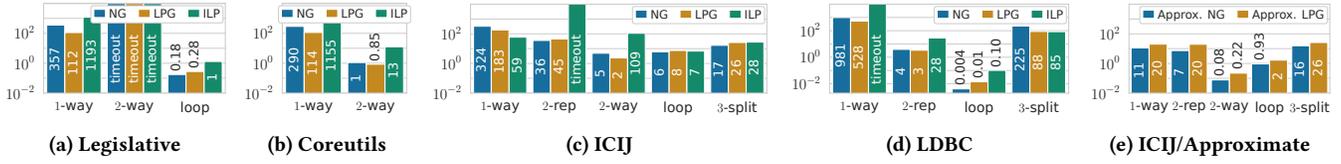

	\centering%
	\begin{subfigure}{.17\textwidth}
		\includePlot{Legislative}{legislative}{classwise-runtime}
	\end{subfigure}
	\hfill%
	\begin{subfigure}{.133\textwidth}
		\includePlot{Coreutils}{coreutils}{classwise-runtime}
	\end{subfigure}
	\hfill%
	\begin{subfigure}{.26\textwidth}
		\includePlot{ICIJ}{icij}{classwise-runtime}
	\end{subfigure}
	\hfill%
	\begin{subfigure}{.22\textwidth}
		\includePlot{LDBC}{ldbc-sf0.3}{classwise-runtime}
	\end{subfigure}
	\hfill%
	\begin{subfigure}{.185\textwidth}
		\includePlot{ICIJ/Approximate}{icij}{classwise-runtime-approx}
	\end{subfigure}
	\hfill%
	\caption{%
		Runtimes in seconds of the naive greedy (NG), LP-guided greedy (LPG), and the ILP algorithms%
	}
\end{figure*}

\begin{figure*}
	\centering%
	\begin{subfigure}{.17\linewidth}
		\includePlot{Legislative}{legislative}{classwise-edge-deletions}
	\end{subfigure}
	\hfill
	\begin{subfigure}{.133\linewidth}
		\includePlot{Coreutils}{coreutils}{classwise-edge-deletions}
	\end{subfigure}
	\hfill
	\begin{subfigure}{.26\linewidth}
		\includePlot{ICIJ}{icij}{classwise-edge-deletions}
	\end{subfigure}
	\hfill
	\begin{subfigure}{.22\linewidth}
		\includePlot{LDBC}{ldbc-sf0.3}{classwise-edge-deletions}
	\end{subfigure}
	\hfill
	\begin{subfigure}{.185\textwidth}
		\includePlot{ICIJ/Approximate}{icij}{classwise-edge-deletions-approx}
	\end{subfigure}
	\hfill
	\caption{%
		Number of edge deletions proposed by the naive greedy (NG), LP-guided greedy (LPG), and the ILP algorithms%
	}
\end{figure*}

\begin{figure*}
	\begin{subfigure}{.27\textwidth}
		\centering%
		\includegraphics[scale=\plotscale]{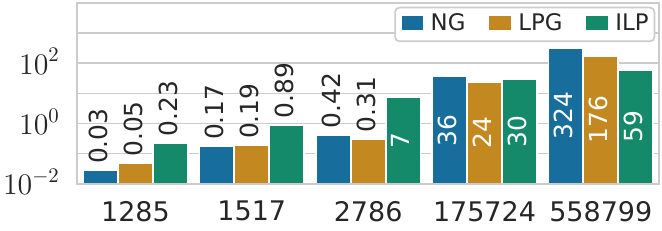}
		\caption{ICIJ/Runtimes/\(1\)-way}\label{figure:icij-err-c-runtime}
	\end{subfigure}
	\hfill
	\begin{subfigure}{.22\textwidth}
		\centering%
		\includegraphics[scale=\plotscale]{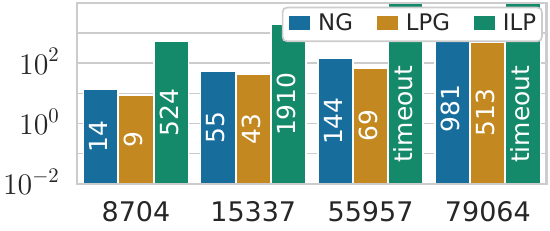}
		\caption{LDBC/Runtimes/\(1\)-way}\label{figure:ldbc-err-c-runtime}
	\end{subfigure}
	\hfill
	\begin{subfigure}{.27\textwidth}
		\centering%
		\includegraphics[scale=\plotscale]{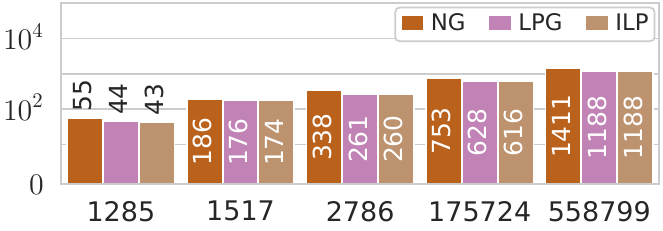}
		\caption{ICIJ/Edge Deletions/\(1\)-way}\label{figure:icij-err-c-deletions}
	\end{subfigure}
	\hfill
	\begin{subfigure}{.22\textwidth}
		\centering%
		\includegraphics[scale=\plotscale]{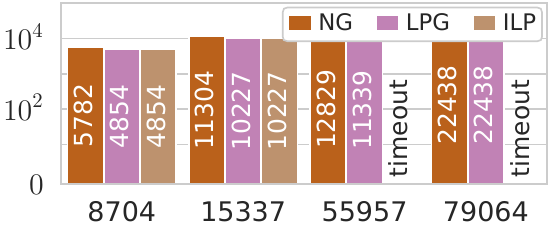}
		\caption{LDBC/Edge Deletions/\(1\)-way}\label{figure:ldbc-err-c-deletions}
	\end{subfigure}
	\hfill
	\caption{%
			Results of the scaling experiments for the naive greedy (NG), LP-guided greedy (LPG), and the ILP algorithms%
	}%
\end{figure*}

	Figures~\ref{figure:legislative-classwise-runtime} to~\ref{figure:ldbc-sf0.3-classwise-runtime} show the runtime of the naive greedy (NG), LP-guided greedy (LPG), and the ILP algorithms.
Note that we used a log-scaled y-axis to obtain readable plots.
For example, Figure~\ref{figure:legislative-classwise-runtime} shows the performance of the base pipeline when run on the legislative property graph with the sets of \(1\)-way, \(2\)-way, and loop constraints outlined in Table~\ref{table:static-pgs}.
The ILP algorithm takes in general more time than the naive greedy algorithm, and as a trade-off, it deletes less objects\footnote{Due to the nature of our constraints, the pipeline always opts for deleting edges without any optional feature enabled.}, cf.\ Figures~\ref{figure:legislative-classwise-edge-deletions} to~\ref{figure:ldbc-sf0.3-classwise-edge-deletions} (which also have a log-scaled y-axis, but a different range than the runtime plots).
The LP-guided greedy achieves the same quality as the ILP algorithm in almost all cases (cf.\ Figures~\ref{figure:legislative-classwise-edge-deletions} to~\ref{figure:ldbc-sf0.3-classwise-edge-deletions}).
At the same time it competes with the naive greedy algorithm performance-wise.
In many cases it even performs better.
Inspecting the solutions of the LP-solver revealed that it converges towards integral solutions in these cases.
Overall, the experiments show that the LP-guided greedy algorithm yields higher-quality repairs with about \(35\%\) fewer deletions than the naive greedy, matching the quality of repairs yielded by the ILP algorithm.
It also offers a runtime advantage of over \(97\%\) (Figure~\ref{figure:icij-classwise-runtime}, \(2\)-way) compared to the ILP strategy, and beats the naive greedy algorithm in a case where the ILP one timed out (Figure~\ref{figure:ldbc-sf0.3-classwise-runtime}, \(1\)-way).
These observations in the general case are however not always true as
demonstrated by Figure~\ref{figure:icij-classwise-runtime} (ICIJ): the ILP algorithm
outperforms both greedy algorithms for \(1\)-way constraints, which exhibit a high number of errors with paths containing many cycles.

	\paragraph*{Approximate Repairs} We now investigate approximate repairs and deepen the comparison between the naive and the LP-guided greedy algorithms by disabling the trimming phase in the pipeline.
	The results for the ICIJ graph are shown in Figures~\ref{figure:icij-classwise-runtime-approx} and~\ref{figure:icij-classwise-edge-deletions-approx}.
	We observe that the selection phase of the naive greedy performs better than the one of the LP-guided greedy, and compared with the full runtimes in Figure~\ref{figure:icij-classwise-runtime}, the trimming phase is responsible for the majority of the runtime for both.
	In particular, the trimming phase of the naive greedy is responsible for making it perform worse overall.
	Correlating the number of proposed edge deletions indicates that this is because, for the naive greedy, more edge deletions are trimmed.
	When comparatively many edge deletions have to be trimmed for the LP-guided greedy, it can be outperformed by the naive greedy algorithm (e.g.\ ICIJ, \(2\)-rep).
	Lastly, the selection phase of the LP-guided greedy algorithm can yield a good approximation or even a proper repair in the majority of cases, with a \(89\%\) reduction in runtime (\(1\)-way).
	Similar results for the other datasets are provided in \refapp{appendix:experiments}{Appendix~C}.

\paragraph*{Scalability} As for the scalability of our pipeline with respect to the numbers of constraints and errors, Figures~\ref{figure:icij-classwise-runtime} and~\ref{figure:legislative-classwise-runtime} attest that our pipeline can handle up to \(188\) loop constraints for the ICIJ graph, while it runs into a timeout for the smaller legislative graph with \(53\) \(2\)-way constraints.
Regarding the size of errors, the maximal size we observed were \(73\) objects (Legislative, loop). %
Figure~\ref{figure:ldbc-sf0.3-classwise-runtime} shows that over \(2.7\) million errors due to \(3\)-split constraint can be repaired comparatively fast, while \(79064\) errors due to \(1\)-way constraints are challenging.
Comparing the results for \(1\)-way constraints for the ICIJ and LDBC graphs, we observe that our pipeline performs better for the former, even though there are fewer constraints (of the same shape) and fewer errors for the smaller LDBC graph (cf.\ Table~\ref{table:static-pgs}, \(1\)-way column).
To investigate this case further, we ran the experiments for the ICIJ and LDBC graphs with subsets of \(1\)-way constraints, and thus fewer errors.
The resulting runtimes are shown in Figures~\ref{figure:icij-err-c-runtime} and~\ref{figure:ldbc-err-c-runtime}, and the number of deletions in Figures~\ref{figure:icij-err-c-deletions} and~\ref{figure:ldbc-err-c-deletions}.
The \(y\)-axes in these plots are labelled with the numbers of errors.
For the LDBC graph,  compared to the number of errors, a high number of edges are deleted -- up to \(67\%\) for the LP-guided greedy and ILP algorithms, and \(73\%\) for naive greedy algorithm -- while only very few deletions are necessary to repair the ICIJ graph.
This indicates that our repair pipeline performs well, if the repair can be achieved with a low number of deletions.
Constraints and number of errors play a relatively negligible role.

\begin{figure*}
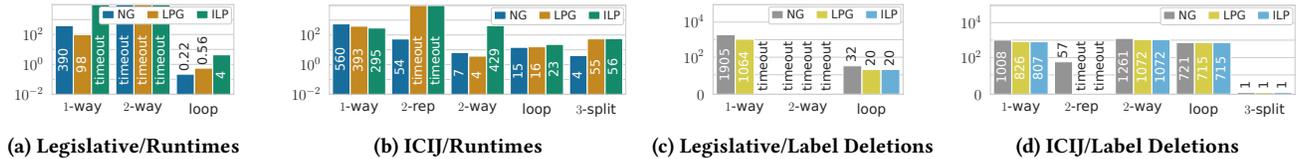

	\centering%
	\hfill%
	\begin{subfigure}{.21\linewidth}
		\includePlot{Legislative/Runtimes}{legislative}{classwise-label-runtime}
	\end{subfigure}
	\hfill
	\begin{subfigure}{.28\linewidth}
		\includePlot{ICIJ/Runtimes}{icij}{classwise-label-runtime}
	\end{subfigure}
	\hfill
	\begin{subfigure}{.21\linewidth}
		\includePlot{Legislative/Label Deletions}{legislative}{classwise-label-deletions}
	\end{subfigure}
	\hfill
	\begin{subfigure}{.28\linewidth}
		\includePlot{ICIJ/Label Deletions}{icij}{classwise-label-deletions}
	\end{subfigure}
	\hfill
	\caption{%
			Runtimes (in seconds) with Step~2 enabled and number of label deletions for the legislative and ICIJ graphs
	}
	\label{figure:classwise-label-deletions}
\end{figure*}

\subsection{Minimizing Deletions of Important Objects}\label{section:custom-weights}
	By default our pipeline assigns uniform weights to edges, and node
weights are determined by degree. To improve robustness, we allow the
pipeline to incorporate varying \emph{custom weights} \(w_c(o) > 0\)
for nodes and edges \(o\in \nodeset \cup \edgeset\), stored as
property values in the property graph, with the aim of reducing the
deletion of objects deemed \say{important}.
Internally, the weight \(w(e)\) of an edge is then \(w_c(e)\), and the
weight \(w(n)\) of a node is \(w_c(n)\) plus the sum of all \(w(e)\), for all
incident edges \(e\).  The objective of the pipeline is then shifted
from minimizing the number of deletions to minimizing the reduction of
the total custom weight across all nodes and edges.

To assess this feature, we used PageRank scores as weights for
nodes and edges\footnote{Page rank for edges are computed using the
line graph: Every edge becomes a node, and there is an edge
\((e_1,e_2)\) if there is a path \(v_0e_1v_1e_2v_2\) in the original
graph.}.
Without custom weights, the total PageRank value decreases by
\(0.09\%\) (\(1\)-way constraints) when applying the LP-guided greedy
algorithm on the ICIJ graph.  With custom weights, the decrease is
limited to \(0.07\%\) (\(2\)-way constraints), corresponding to a
\(24\%\) reduction in weight loss, indicating that fewer
high-weight edges are deleted.  The number of deletions is shown in
Figure~\ref{figure:icij-classwise-edge-deletions-custom-weight}, and
compared with Figure~\ref{figure:icij-classwise-edge-deletions}, showing an
increase ranging from \(1\%\) (\(2\)-rep constraints) to \(43\%\)
(\(2\)-way constraints).
The runtimes are shown in
Figure~\ref{figure:icij-classwise-runtime-custom-weight}. Compared with
the base case in Figure~\ref{figure:icij-classwise-runtime}, custom
weights can have a positive or negative impact on runtime.
More details are provided in \refapp{appendix:experiments}{Appendix~C}.%

\begin{figure*}
\centering%
\hfill
\begin{subfigure}{.2\textwidth}
	\centering%
	\includegraphics[scale=\plotscale]{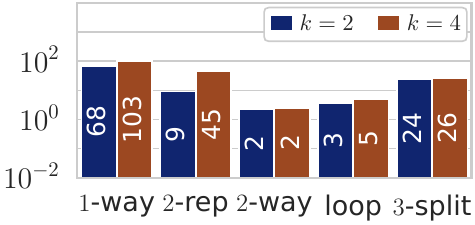}
	\caption{Step~3/Runtimes}
	\label{figure:icij-classwise-nbh-runtime-lpr}
\end{subfigure}
\hfill%
\begin{subfigure}{.2\textwidth}
	\centering%
	\includegraphics[scale=\plotscale]{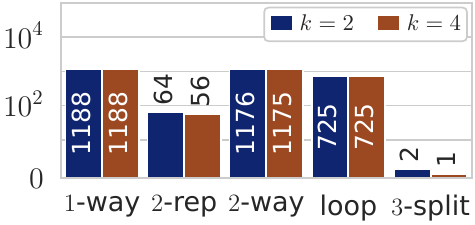}
	\caption{Step~3/Edge Deletions}
	\label{figure:icij-classwise-nbh-deletions-lpr}
\end{subfigure}
\hfill%
\begin{subfigure}{.28\textwidth}
	\includePlot{Custom Weights/Runtimes}{icij}{classwise-runtime-custom-weight}
\end{subfigure}
\hfill%
\begin{subfigure}{.28\textwidth}
	\includePlot{Custom Weights/Edge Deletions}{icij}{classwise-edge-deletions-custom-weight}
\end{subfigure}
\hfill
\caption{Results of the LP-guided greedy algorithm for neighbourhood (Step~3) and custom weight repairs for ICIJ}\label{figure:icij-classwise-nbh-lpr}
\end{figure*}

\begingroup%
\begin{figure}
	\centering%
	\hfill
	\begin{subfigure}{.42\columnwidth}
		\includegraphics[scale=\plotscale]{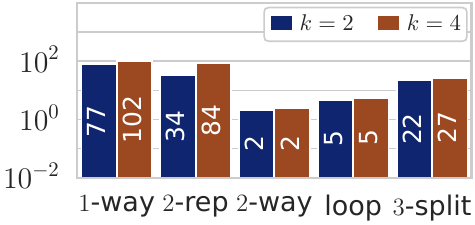}
		\caption{ICIJ/Runtimes}
		\label{figure:icij-classwise-sample-runtime-lpr}
	\end{subfigure}%
	\hfill
	\begin{subfigure}{.42\columnwidth}
		\includegraphics[scale=\plotscale]{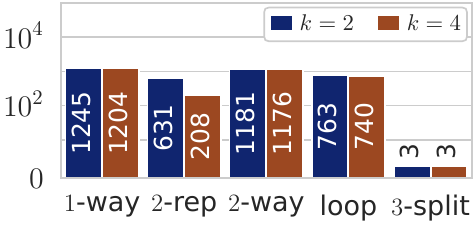}
		\caption{ICIJ/Edge Deletions}
		\label{figure:icij-classwise-sample-deletions-lpr}
	\end{subfigure}%
	\hspace{.1\columnwidth}
	\hfill%
	\caption{LP-guided greedy algorithm on sampled errors}%
\label{figure:icij-classwise-sample-lpr}%
\end{figure}%
\endgroup%
\subsection{Label Removals and Neighbourhood Repairs}\label{section:ablation}
We now evaluate the optional steps of our pipeline, and compare them to the base pipeline established in Section~\ref{section:performance-quality}.

To show the effects of Step~\(2\), we present the runtimes and number of label deletions for the legislative and the ICIJ graph in Figure~\ref{figure:classwise-label-deletions}.
	Recall that extending errors with essential tokens results in more, and larger errors.
	Consequently, the (I)LPs are also larger.
	Furthermore, there tend to be more equally weighted repairs: instead of deleting an edge, there is now the choice between deleting any label from the edge, or one of its endpoints.
	The runtimes increase considerably -- by up to \(500\%\) (Figure~\ref{figure:icij-classwise-runtime} vs.\ Figure~\ref{figure:icij-classwise-label-runtime}, \(1\)-way).
	The ILP algorithm even timed out in one more case (Figure~\ref{figure:legislative-classwise-label-runtime}, \(1\)-way).

However, comparing with Figures~\ref{figure:legislative-classwise-edge-deletions} and~\ref{figure:icij-classwise-edge-deletions}, we observe a reduction of up to \(50\%\) of object deletions for the naive greedy algorithm (Legislative, \(1\)-way), of up to \(47\%\) for the ILP algorithm (Legislative, loop), and of up to \(59\%\) for the LP-guided greedy (Legislative, \(1\)-way, the ILP algorithm times out in this case, cf.\ Figure~\ref{figure:legislative-classwise-label-deletions}).

	Notably, even the naive greedy algorithms can yield a repair with less label deletions than the ILP with (only) edge deletions (Legislative, \(1\)-way).
	Note that the greedy algorithms can sometimes even be faster if label deletions are allowed which has to be attributed to the structure of the errors and the (I)LP (e.g.\ Figure~\ref{figure:legislative-classwise-label-runtime}, \(1\)-way).

A similar trend can be observed for the other datasets; for space reasons they are provided in \refapp{appendix:experiments}{Appendix~C}.

With Step~\(3\) we observe a runtime reduction of up to~\(62\%\) compared to the base pipeline while maintaining almost the same quality, i.e.\ number of deletions.
We present the runtimes and number of proposed edge deletions for ICIJ and neighbourhood \(k = 2\) as well as \(k = 4\) in Figures~\ref{figure:icij-classwise-nbh-runtime-lpr} and~\ref{figure:icij-classwise-nbh-deletions-lpr}, for the LP-guided greedy algorithm.
Compared to the base case shown in Figure~\ref{figure:icij-classwise-runtime}, the performance is better in cases where the base pipeline takes a significant amount of time, i.e.\ more than a few seconds.

The numbers of deletions are quite similar for \(k = 2\) and \(k = 4\), and more importantly, the results for the base pipeline as in Figure~\ref{figure:icij-classwise-edge-deletions}, except for the experiment for \(2\)-rep constraints.
That is, in all cases where matches consisting of \emph{shortest paths} are (almost) completely covered by the \(k\)-endpoint neighbourhood.

	To assess whether picking the neighbourhoods of path endpoints is a suitable strategy for reducing the size of errors, we performed the same experiments with random samples of errors: that is, we sampled \(2k\) edges (and their endpoints) from each error.
	The result for the LP-guided greedy algorithm on the ICIJ graph is shown in Figure~\ref{figure:icij-classwise-sample-lpr}.
	Compared with the result for neighbourhood errors shown in Figure~\ref{figure:icij-classwise-nbh-lpr}, the runtimes do not differ significantly, but the number of deletions can explode if paths are long (\(9.86\)x for \(2\)-rep).
	This suggests that intersections of (topological) errors are indeed more likely to happen near path endpoints, and taking a random sample from a long path is less likely to contain these intersections.%

Similar plots to the ones for Steps~\(2\) and~\(3\) for the other datasets are available in \refapp{appendix:experiments}{Appendix~C}.

\subsection{A Qualitative Study}\label{section:qualitative-study}
	\begingroup%
	\newcommand{\constraintDescription}[1]{\(\mathbf\hookrightarrow\) \textcolor{purple!80!black}{#1}}
	\begin{table*}
		\caption{Real-world like RGPC constraints for the ICIJ property graph}%
		\label{table:icij-qualitative-study-constraints}%
		\footnotesize%
		\begin{tabular}{l c r r}
			\toprule
			ID & RGPC constraint \(\query; F \cimp C\) & \#errors & \#deletions\\
			\midrule
			\addlinespace
			\multirow{2}{*}{\(\gamma_1\)} & \(\gpcnp[x]{\pglabel{Entity}\(\;\lor\;\)\pglabel{Officer}} \gpcep{\pglabel{same\_name\_as}} \gpcnp[y]{\pglabel{Entity}\(\;\lor\;\)\pglabel{Officer}}; \emptyset \cimp \{\propof{x}{name} = \propof{y}{name}\}\) & \multirow{2}{*}{\(18000\)} & \(18000\) edges\\
			&\constraintDescription{All entities and officers in a \pglabel{same\_name\_as} relationship have indeed the same name.}& & \(0\) nodes\\
			\addlinespace
			\multirow{2}{*}{\(\gamma_2\)} & \(\gpcnp[x]{\pglabel{Entity}}; \{\propof{x}{inactivation\_date} \le \textsf{2025-07-01}\} \cimp \{\propof{x}{status} \neq \text{\say{Active}}\}\) & \multirow{2}{*}{\(46\)} & \(177\) edges\\
			&\constraintDescription{All entities with an inactivation date in the past do not have the status \say{Active}.}& & \(46\) nodes\\
			\addlinespace
			\multirow{2}{*}{\(\gamma_3\)} & \(\gpcnp[x]{\pglabel{Officer}} \gpcep{\pglabel{sole\_director\_of}}\gpcnp[y]{\pglabel{Entity}},
			\gpcnp[z]{\pglabel{Officer}}\gpcep{\pglabel{sole\_director\_of}}(y); \emptyset \cimp \{x = z\}\) & \multirow{2}{*}{\(44\)} & \(22\) edges\\
			&\constraintDescription{An entity cannot have more than one sole director.}&& \(0\) nodes\\
			\addlinespace
			\multirow{2}{*}{\(\gamma_4\)} & \hspace{-.5em}\(\gpcnp[x]{\pglabel{Officer}} \gpcep{\pglabel{president\_of}}\gpcnp[y]{\pglabel{Entity}} \gpcrep{+}{
				\gpcnp{\pglabel{Entity}\(~\lor~\)\pglabel{Officer}}\gpcep{\pglabel{officer\_of}\(~\lor~\)\pglabel{same\_name\_as}}
			}\gpcnp[z]{\pglabel{Entity}}, (x) \gpcep{\pglabel{officer\_of}}(z); \emptyset \cimp \{y = z\}\)\hspace{-2.75em} & \multirow{2}{*}{\(19\)} & \(10\) edges\\
			&\constraintDescription{No president is indirectly in control of an entity, e.g.\ through a chain of shell companies, they are directly an officer of.}& & \(0\) nodes\\
			\bottomrule
		\end{tabular}
	\end{table*}%
	\endgroup%
	To showcase that our repair pipeline is viable in practice, we conducted a qualitative study on the real-world ICIJ property graph, which is built and used by the International Consortium of Investigative Journalists.
	Since this property graph was assembled from several data sources, it contains some real-world inconsistencies, detected using the RGPC constraints in Table~\ref{table:icij-qualitative-study-constraints}.

	The constraints are concerned with \emph{entities}, which model companies and organisations, and \emph{officers}, which are persons involved with entities.
	We note that entities can also be in an \pglabel{officer\_of} relationship with other entities, although neither of them is labelled as an officer.
	Constraint \(\gamma_1\) states that entities and/or officers in a \pglabel{same\_name\_as} relationship, have indeed the same name.
	Constraint \(\gamma_2\) is a typical denial constraint forbidding an entity from having an inactivation date, if it is still active.
	Constraint \(\gamma_3\) models an \emph{exclusivity} condition for entities \cite{Angles2021}: no entity can have two \emph{sole} directors.
	Finally, \(\gamma_4\) ensures that a president of an entity cannot gain an advantage by indirectly influencing another entity they are associated with, through their presidency.
	The number of errors for each constraint is given in Column~\(3\) of Table~\ref{table:icij-qualitative-study-constraints}.

	We first discuss the repairs yielded by the base pipeline with the ILP algorithm.
	The numbers of deletions are given in Column~\(4\) of Table~\ref{table:icij-qualitative-study-constraints}.
	For \(\gamma_1\), \pglabel{same\_name\_as} relationships are deleted. This preserves the original names stored in the properties (unifying names could lead to false accusations in this scenario).
	Since errors for this constraint cannot share edges, one deletion per error is expected.
		Constraint~\(\gamma_2\) is an example where the pipeline is forced to delete nodes, since each  error consists of exactly one node.
		Thus, the nodes are deleted, and consequently so are all their \(177\) incident edges.
	For \(\gamma_3\) one of the two relationships in each error is deleted, which shows that symmetries are resolved properly (the errors do not share any objects in this case).
	The repair for \(\gamma_4\) shows that the number of deletions can be significantly less than the number of errors -- our pipeline exploits that errors have edges in common.

		We investigate robustness w.r.t.\ graph topology changes with \(\gamma_4\), since it encompasses a non-trivial path pattern.
		Assuming that sub-patterns describing paths of \pglabel{same\_name\_as}- and \pglabel{president\_of}-labelled edges are important, we double the weight of these edges.
		The repair pipeline then proposes the deletion of \(14\) instead of \(10\) edges but preserves all of the aforementioned edges.
		It also preserves at least \(168\) of \(203\) matches of the first pattern of~\(\gamma_4\).
		Without the custom weights, it preserves only between \(96\) and \(117\) matches of this pattern.

	When enabling label deletions, the graph can be repaired with \(6323\) label deletions, which is only about \(35\%\) of the \(18255\) node and edge deletions without it.
	\(1394\) labels are removed from edges, effectively deleting the concrete relationship but preserving the knowledge that there is some (potential) relationship.

	Using the naive greedy algorithm instead of the ILP algorithm resulted in up to \(2\) more deletions for~\(\gamma_4\) (depending on the order in which the database returned errors), the same amount of deletions for the other constraints.
	The LP-guided greedy yielded the same number of deletions.
	With the ILP algorithm taking under \(4\) seconds for the set of all constraints, there is no advantage of using either greedy algorithm or neighbourhood errors in this scenario.
\section{Related Work}\label{section:related-work}
Error repairing for PG-Constraints has not been addressed in prior work besides the assessment that PG-Constraints can be rewritten into GQL queries for querying errors \cite{Angles2021}.

Error repairing for less expressive graph data models, such as labelled graphs, focuses on chase-based approaches \cite{Fan2023, Fan2019b,Kwashie2019,Cheng2018,Cheng2022,Shimomura2020,Shimomura2022}, which iteratively add properties, or even edges to repair a graph.
Thus, this repair mode is orthogonal to our delete-based approach, and consequently not comparable.
In particular, insertion-only repair modes are not suitable for repairing violations of denial constraints in our examples and experimental settings \cite{Fan2008}.

A different approach for repairing graphs is with respect to neighbourhood constraints, that is, constraints which restrict the labels of a node depending on the labels of its neighbours \cite{DBLP:journals/vldb/SongLCYC17}.
Their repair mode involves changing labels, and, if that is not enough, changing the neighbourhood of a node, which is again incompatible with our repair mode. They introduce a greedy variant of their algorithm to strike a balance between runtime and repair quality.

User-centric repairs for property graphs have been studied \cite{DBLP:journals/pacmmod/PacheraBM25}, albeit w.r.t.\ constraints for less expressive graph data models.
Their repair model focuses on label and property changes, but also allows the deletion of edges as fall-back.
Letting users repair errors consisting of variable-length paths requires carefully crafted human-machine interfaces, which is beyond the scope of this work.

The idea of using vertex covers of conflict hypergraphs originates from theoretical work for repairing relational databases \cite{DBLP:journals/iandc/ChomickiM05,Staworko2012}, but these repairs do not involve dependencies between objects and arbitrary long paths.
	A notable consequence is that, unlike in the relation setting, not every minimal vertex cover corresponds to a repair, which we addressed by using minimum weighted vertex covers.
	We note that modelling dependencies explicitly (like foreign-key constraints) instead of using weights, yields an unreasonable amount of (artificial) errors in the case of property graphs, since they would involve all nodes, edges, and labels.
For the vast literature on error repairing for other database models and repair modes, including the relational model, we refer to surveys \cite{Chu2016, Elmagarmid2007}.

	We note that the vertex cover problem for hypergraphs is even W[2]-hard \cite{DBLP:series/txtcs/FlumG06}, which means that it is very unlikely that there is an FPT algorithm.
	A common strategy for solving ILPs are LP-rounding algorithms \cite[e.g.,][]{DBLP:journals/tcs/OualiFS14}, which solve the LP variant of the ILP and then use \emph{randomized rounding} to obtain an integral solution.
	Our LP-guided greedy algorithm uses a trivial rounding step instead, which is arguably sufficient, as shown by our experiments in Section~\ref{section:experiments}.

Theoretical foundations for RGPC patterns have been investigated, in particular CRPQs with path variables \cite{Barcelo2012,Figueira2022}, which even allow for \say{regular} comparisons of paths. %
Notably, unlike RGPC patterns, CRPQs only match edge labels (i.e.\ no node labels), and do not support label expressions.
The same differences apply to our RGPC automata model and automata constructed for CRPQs \cite{Barcelo2012}.

Lastly, we defined RGPC patterns as a syntactical fragment of GPC patterns to abstract core features of the pattern matching sublanguage of GQL and SQL/PGQ \cite{Francis2023a}.
A useful feature in practice and supported by these languages are patterns which allow for traversing edges in reverse direction.
We outline how our pipeline can be extended with this feature and other differences in \refapp{section:reverse-traversal}{Appendix~A.3}.
It is known that enumerating matches of GPC patterns is in polynomial space in the size of the property graph but not in polynomial space in the size of the pattern \cite{Francis2023a}; making the enumeration of all errors in large graphs intractable, and sets boundaries for repair procedures relying on enumerating all errors.
Finally, PG-Constraints can capture dependencies based on regular path queries for which problems related to graph repair like certain query answering are undecidable \cite[see, e.g.][]{DBLP:conf/icalp/BeeriV81,Francis2017, Salvati2024}.

\section{Conclusion and Future Work}\label{section:conclusion}
In this paper, we addressed error repairing for property graphs under
PG-Constraints. We designed a repair pipeline for RGPC constraints
with user-selectable trade-offs and demonstrated its effectiveness
through extensive experiments.

As future work, we plan to study the practical limits of error
repairing under PG-Constraints. Additional features could be
incorporated such as advanced GQL pattern
matching~\cite{Deutsch2022,Figueira2023}. The repair model may also be
extended to allow property deletions, though this is challenging since
deletions can introduce new errors. Other repair modes could combine
our deletion-based approach with insertion techniques from the state
of the art. Finally, given the large number of errors recursive
constraints may generate, and inspired by work on incremental error detection
in the relational setting~\cite{Kaminsky2024}, we aim to explore how
errors can be efficiently prevented from being inserted in the first
place.
\begin{acks}
	The authors were supported by ANR-21-CE48-0015 VeriGraph.
	A.~Bonifati is also funded by IUF Endowed Chair.
\end{acks}

\pagebreak
\bibliographystyle{ACM-Reference-Format}
\bibliography{main}

\clearpage
\appendix

\section*{Appendix}
This appendix consists of three sections.
In Appendix~{appendix:rgpc} we provide additional material on RGPC patterns and constraints: We compare RGPC patterns with CRPQs and full GPC in more detail than in Section~\ref{section:related-work}, provide more details on how RGPC constraints can be written as PG-Constraints, provide the formal semantics of predicates, and present a possible extension with reverse traversals of edges.
Appendix~{appendix:repairing} provides more details about RGPC automata, the missing proofs for Section~\ref{section:cleaning}, as well as an alternative version of the ILP algorithm.
Finally, in Appendix~\ref{appendix:experiments}, the interested reader can find additional experimental results together with the raw measurements for all experiments and trade-off percentages.

\section{RGPC and RGPC Constraints}\label{appendix:rgpc}
In Appendices~\ref{section:comparison-crpqs} and~\ref{section:comparison-gpc} we compare RGPC with CRPQs and full GPC, respectively.
And in Appendix~\ref{section:definition-literal-semantics} we provide a formal definition for predicates in RGPC constraints.
In Appendix~\ref{section:reverse-traversal} we explain how RGPC can be extended with reverse traversals.

\subsection{Relationship to CRPQs}\label{section:comparison-crpqs}
For establishing theoretical foundations, it can be helpful to understand RGPC patterns as extensions of (full) \emph{conjunctive regular path queries (CRPQ)} \cite[see e.g.][]{Barcelo2012}.
A full CRPQ consists of a set of triples \((x,r,y)\) where \(x\) and \(y\) are node variables, and \(r\) is a regular expression over edge labels.
A match \(\match\) of a triple \((x,r,y)\) in a (property) graph is a mapping that maps \(x, y\) to nodes \(n, m\) of the graph for which there is a path between~\(n\) and~\(m\) such that the \emph{edge labels} of this path form a word in the language of the regular expression \(r\).

For example, \((x, a(bb)^*, y)\) matches all pairs of nodes \(v_0, v_n\) for which there is an odd-length path \(v_0e_1v_1\ldots e_nv_n\) such that \(e_1\) is labelled \(a\) and \(e_2,\ldots e_n\) are labelled \(b\).
It can be easily translated into an RGPC path pattern, namely \((x)\gpcep{a}\gpcrep{*}{\gpcep{b}\gpcep{b}}(y)\).
Other triples \((x, r, y)\) can be translated similarly; and a CRPQ consisting of multiple triples can thus be translated into an RGPC graph pattern.
We note that the resulting RGPC patterns \emph{do not} contain any node labels, since the regular expressions \(r\) can only refer to edge labels.
Furthermore, in an edge pattern \(\gpcep{\(\varphi\)}\), the label expression \(\varphi\) is always a single label (and does not contain any conjunction, disjunction or negation).
Node variables only occur at the beginning and at the end of a path pattern.

\subsection{Formal Semantics of Predicates}\label{section:definition-literal-semantics}

As stated in Section~\ref{section:constraints}, the semantics of predicates that can occur in the filter and condition of an RGPC constraint is as usual, with the addition that all properties occurring in a predicate must be present.
For the sake of completeness, we provide a formal definition below.
Recall that \(\oplus\) ranges over elements from \(\{=,\neq,\le,\ge\}\).
A match \(\match\) for an RGPC pattern \(\query\) in a property graph \(\graph = \pgtuple\) satisfies a predicate \(\ell\) if
\begin{itemize}
	\item \(\ell\) is of the form \(x_i = y_j\) and \(\match(x_i) = \match(y_j)\);
	\item \(\ell\) is of the form \(x_i\neq y_j\) and \(\match(x_i) \neq \match(y_j)\);
	\item \(\ell\) is of the form \(\propof{x_i}{propA} \oplus \propof{y_j}{propB}\),  \(\pgpropmap(\match(x_i), \prop{propA})\) and \(\pgpropmap(\match(y_j), \prop{propB})\) are both defined, and  \[\pgpropmap(\match(x_i), \prop{propA}) \oplus \pgpropmap(\match(y_j), \prop{propB})\] holds;
	\item \(\ell\) is of the form \(\propof{x_i}{propA} \oplus c\), \(\pgpropmap(\match(x_i), \prop{propA})\) is defined, and \[\pgpropmap(\match(x_i), \prop{propA}) \oplus c\] holds.
\end{itemize}

\subsection{Adding Reverse Traversal}\label{section:reverse-traversal}
RGPC patterns (and our repair pipeline) can be extended by reverse traversals, i.e.\ the ability to traverse edges in reverse direction.
For this purpose, the definition of paths has to be extended as well, to allow for edges in reverse directions.
\begin{definition}[Path]
	Let \(\graph = \pgtuple\) be a property graph.
	A \emph{path} in \(\graph\) is a non-empty alternating sequence
	\(
	v_0e_1v_1\ldots e_nv_n
	\)
	where \(v_0,\ldots,v_n\in\nodeset\) are nodes and \(e_1,\ldots,e_n\in\edgeset\) are edges such that \(\pgrelmap(e_i) = (v_{i-1},v_i)\) \emph{or \(\pgrelmap(e_i) = (v_i,v_{i-1})\)} holds for every \(i\in\range{n}\).
\end{definition}

The definition of RGPC can now be augmented with \emph{reverse edge patterns} \(\gpceprev{\(\varphi\)}\).
Unlike (forward) edge patterns \(\gpcep{\(\varphi\)}\), reverse edge patterns require an edge to be traversed in reverse direction in a path to be matched.
For example, \(\gpceprev{a} \cup \gpcep{c}\gpcep{c}\) states that there is an edge labelled \pglabel{a} in reverse direction, \emph{or} a path of length two on which both edges are labelled \pglabel{c}.

\subsection{RGPC vs.\ Full GPC}\label{section:comparison-gpc}
In Section~\ref{section:constraints} we defined RGPC patterns as a subset of GPC patterns.
GPC as defined in the literature \cite{Francis2023a} has several more features.
Notably, it also supports edge variables, and variables to occur below union and repetition operators.
This effectively yields two more types of variables, namely optional and group variables, and the semantics become more involved --
in the literature a type system is used to determine whether a pattern is valid and define the semantics \cite[Section~4]{Francis2023a}.
This is not required for our fragment.
Another feature are conditions to match property values.
Since property values are covered by the other components of our constraints, namely \(P\) and \(C\), we do not require this feature.
GPC also allows for traversing edges in the reverse direction or undirected edges.
Lastly, path patterns are equipped with restrictors like, for instance, \say{simple} and \say{shortest}, which restrict matches to simple paths (no repeated nodes) and shortest paths.
The purpose of restrictors is to ensure that there are (only) finitely many matches.
This is important if patterns are used to define queries, but constraints do not have a return value.
We note that restrictors have also been studied for CRPQs \cite{Figueira2023}.

Lastly, we allow one to specify label constraints using propositional formulas \(\varphi\).
This is not supported by the original definition \cite{Francis2023a}.
Although conjunction and disjunction of labels can be expressed using concatenation and union, expressing the absence of labels (using negation) is, strictly speaking, not possible in the original definition of GPC.

\subsection{RGPC Constraints vs.\ PG-Constraints}
As discussed in the introduction, RGPC constraints form a subset of PG-Constraints, which were proposed as an extension of concrete graph query languages, in particular GQL, with constraints \cite{Angles2021,Angles2023}.

In Section~\ref{section:constraints}, we briefly discussed that the constraint from Example~\ref{example:constraints-intro} is equivalent to the PG-Constraint shown in Figure~\ref{figure:example-pg-intro}, which uses the  \lstinline|MANDATORY| keyword.
Here we provide some more details.

In general, a \emph{PG-Constraint} stipulating mandatory elements is written as
\begin{lstlisting}[numbers=none]
	FOR |\(\tuple{x}\)| |\(p(\tuple{x})\)| MANDATORY |\(\tuple{y}\)| |\(q(\tuple{x}, \tuple{y})\)|.
\end{lstlisting}
Here \(p(\tuple{x})\) and \(q(\tuple{x},\tuple{y})\) are queries called the \emph{scope} and \emph{descriptor query}, respectively.
They refer to variables \(\tuple{x} = (x_1,\ldots,x_n)\) and \(\tuple{y} = (y_1,\ldots,y_m)\),
that are assigned to nodes, paths, or properties.

\begin{example}\label{example:pg-constraints-intro}
	The PG-Constraint shown in Figure~\ref{figure:example-pg-intro}, which corresponds to the RGPC constraint from Example~\ref{example:constraints-intro}, uses GQL-like scope and descriptor queries.
	The \lstinline|MATCH| clause and the first \lstinline|FILTER| clause form the scope query with free variables \lstinline|x| and \lstinline|y|.
	The \lstinline|MATCH| clause is essentially an ASCII art translation of the RGPC graph pattern in Example~\ref{example:constraints-intro}.
	The first \lstinline|FILTER| clause asserts that the project has started by comparing the start (date) with the current time.

	The \lstinline|MANDATORY| clause then states that, for every possible match of the \lstinline|MATCH|-clause's pattern, it is \say{mandatory} that the nodes assigned to \lstinline|x| and \lstinline|y| both have an \pgkey{access\_level} property (more precisely, the records of both nodes are defined for the key \pgkey{access\_level}).

	Finally, it is also \say{mandatory} that the descriptor query, which consists of the second \lstinline|FILTER|-clause in this example, is satisfied.
\end{example}
The translation in Example~\ref{example:pg-constraints-intro} can be generalized to arbitrary constraints \(\query; F \cimp C\): Every RGPC pattern \(\query\) can be expressed in GQL \cite{Francis2023a} and \(C\) and \(F\) can be written as \lstinline|FILTER| clauses.
\section{Additional Material for Section~4}\label{appendix:repairing}
In Appendix~\ref{section:automata-construction} we provide the formal definition and construction of our automata model for RGPC patterns informally explained in Appendix~\ref{section:label-repairs}.
Proofs for Section~\ref{section:cleaning:repairing} are presented in Appendix~{section:error-repairing-proofs}.
In Appendix~\ref{section:explicit-ilps} we briefly discuss an alternative (which performs worse) to the integer linear programs presented in Section~\ref{section:repair-pipeline}.

\subsection{Automata for RGPC Patterns}\label{section:automata-construction}
Step~2 of our repair pipeline relies on an automata construction, which we detail in the following.
We translate RGPC path patterns \(\mathcal{P}\) \emph{without node variables} into automata \(\autA\) which accept precisely those words which are traces of paths matched by \(\mathcal{P}\).
For comprehensibility and readability, we assume in the following that \(\mathcal{P}\) does not contain any reverse traversals as introduced in Appendix~\ref{section:reverse-traversal}, i.e.\ (sub-)patterns of the form \(\gpceprev{\(\varphi\)}\).
We will address reverse traversals towards the end of this section.

We start by defining our automata model, which follows the definition of classical symbolic automata, except that there are two sets \(\statesetA_N\) and \(\statesetA_E\) of states.
A state from \(\statesetA_N\) indicates that the next transitions reads labels of a node, and a state from \(\statesetA_E\) that the next transition reads labels of an edge.
We write \(\mathcal{B}(S)\) for the set of propositional formulas over a set \(S\) of propositions.

\begin{definition}[RGPC Automaton]
	A \emph{(finite) RGPC automaton} \(\autA = (\statesetA_N,\statesetA_E,\mathbf{L},\istateA,\transitionRel,\astateset)\) consists of
	\begin{itemize}
		\item two finite, disjoint sets \(\statesetA_N\), \(\statesetA_E\) of \emph{states},
		\item a finite set \(\mathbf{L}\) of labels,
		\item an \emph{initial state} \(\istateA\in\statesetA_N\),
		\item a finite \emph{transition relation} \[
			\transitionRel\subseteq \left(\statesetA_N\times \mathcal{B}(\mathbf{L})\times \statesetA_E\right) \cup \left(\statesetA_E\times \mathcal{B}(\mathbf{L})\times \statesetA_N\right),
		\]
		and
		\item a set \(\astateset\subseteq\statesetA_E\) of \emph{accepting states}.
	\end{itemize}
\end{definition}
An RGPC automaton reads words over the alphabet \(\Gamma = 2^\mathbf{L}\).
A \emph{run} of \(\autA\) on a sequence of label sets \(L_1\ldots L_n\in \Gamma^*\) is a sequence \(\stateA_0\stateA_1\ldots\stateA_n\) of states starting with the initial state \(\istateA\) such that, for each \(i\in\range{n}\), there is a formula \(\varphi_i\) such that \((\stateA_{i-1},\varphi_i,\stateA_{i})\in\transitionRel\), and \(L_i\models \varphi_i\).
Here \(L_i\models \varphi_i\) holds, if the formula \(\varphi_i\) evaluates to true when interpreting all labels in \(L_i\) as variables set to true, and all other labels as variables set to false.
A run \(\stateA_0\stateA_1\ldots\stateA_n\) is \emph{accepting} if \(\stateA_n\in\astateset\).
An RGPC automaton \emph{accepts} a sequence of label sets \(L_1\ldots L_n\) if there is an accepting run on that sequence, it \emph{accepts a path} \(p = v_0e_1v_1\ldots e_nv_n\) of a property graph if it accepts the trace \(\lambda(p) = \lambda(v_0)\lambda(e_1)\lambda(v_1)\ldots \lambda(e_n)\lambda(v_n)\) of \(p\).

\emph{An automaton for an RGPC path pattern} \(z = \mathcal{P}\) is an RGPC automaton \(\autA\) such that, for every property graph \(\graph\), \(\autA\) accepts a path \(p\) in \(\graph\) if and only if there is a match \(\match\) with \(\match(z) = p\), i.e.\ \(p\) is matched by \(z = \mathcal{P}\).

\begin{example}\label{example:old-automata-construction-example}
	Consider the path pattern
	\[
	\gpcnp{a} \gpcep{b} \gpcnp{c}\gpcrep{*}{%
		\gpcnp{d}\gpcep{b \(\lor\) e} \gpcnp{f \(\land~\lnot\)a}
	}.
	\]
	The automaton for this path pattern is illustrated in Figure~\ref{figure:old-example-automaton}.

	First, note that the transitions are labelled with propositional formula instead of (single) labels of the property graph.
	For instance, the transition with the formula \texttt{f \(\land~\lnot\)a} actually represents several transitions: namely, one for each set of labels that contains the label~\(\pglabel{f}\) but not the label~\(\pglabel{a}\).
	In general, computing (and writing) them explicitly would result in an exponential blow-up of the automaton.

	For instance, the outgoing transition of state~\(1\) originates from the node pattern \(\gpcnp{a}\), and the outgoing transition of state \(2\) originates from the edge pattern \(\gpcep{b}\).
\end{example}

\begin{figure}%
	\includegraphics[width=\columnwidth, trim=.25cm .1cm 0cm .1cm, clip]{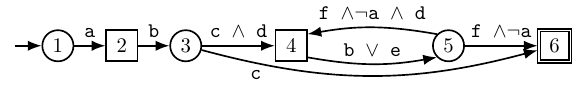}%
	\caption{%
		Automaton for the RGPC pattern from Example~\ref{example:old-automata-construction-example}.%
	}\label{figure:old-example-automaton}%
\end{figure}

We are now prepared to translate an RGPC pattern \(\mathcal{P}\) into an RGPC automaton \(\autA\) for \(\mathcal{P}\).
For this purpose, we adapt the standard construction for translating regular expressions into automata:
Recall that the automaton reads labels from nodes and edges in alternation, and that consecutive node patterns, without edge patterns in between, match the same node (labels).
We observe that RGPC path patterns follow an inductive structure: The base cases are node patterns \(\gpcnp{\(\varphi\)}\) and edge patterns \(\gpcep{\(\psi\)}\), and composed patterns are built by combining path patterns using the union, repetition, and concatenation operators.
We can thus define \(\autA\) by induction over the structure of an RGPC path pattern \(\mathcal{P}\).
During the construction, we maintain the following invariant: \(\autA\) has exactly one accepting state \(\astateA\), and the accepting state \(\astateA\) is \emph{not} the initial state \(\istateA\).
Furthermore, \(\istateA\) has only outgoing transitions and \(\astateA\) has only incoming transitions.
Lastly, let us emphasize that our definition of symbolic automata does \emph{not} permit \(\varepsilon\)-transitions.

Let \(\mathbf{L}\) be the set of labels that occur in the pattern \(\mathcal{P}\).
The base cases can be translated as follows.
\begin{itemize}
	\item \(\mathcal{P} = \gpcnp{\(\varphi\)}\).
	A node pattern asserts the existence of a node whose labels satisfy \(\varphi\).
	We thus set
	\[\autA = (\{\istateA\}, \{\astateA\}, \mathbf{L}, \istateA, \{(\istateA,\varphi,\astateA)\}, \{\astateA\}).\]
	\item \(\mathcal{P} =~\gpcep{\(\varphi\)}\).
	Next to asserting the existence of an edge, \(\mathcal{P}\) also asserts the existence of two nodes, namely the endpoints of the edge.
	Therefore, we define
	\begin{multline*}
		\autA = (\{\istateA,\stateA_t\}, \{\stateA_s,\astateA\}, S, \istateA,\\ \{(\istateA,\top,\stateA_s),(\stateA_s,\varphi,\stateA_t),(\stateA_t,\top,\astateA)\}, \{\astateA\}).
	\end{multline*}
	Here \(\top\) denotes the tautology which is satisfied by every label set.
\end{itemize}

Towards the construction for composed path patterns \(\mathcal{P}\), let \(\autA^1 = (\statesetA^1_N,\statesetA^1_E,\mathbf{L},\istateA^1,\transitionRel^1,\{\astateA^1\})\) and  \(\autA^2 = (\statesetA^2_N,\statesetA^2_E, \mathbf{L},\istateA^2,\transitionRel^2,\{\astateA^2\})\) be the automata constructed for RGPC path patterns \(\mathcal{P}_1\) and \(\mathcal{P}_2\), respectively.
Without loss of generality, we assume that the state sets \(\statesetA^1_N,\statesetA^1_E,\statesetA^2_N,\statesetA^2_E\) are pairwise disjoint.
We construct an automaton \(\autA = (\statesetA_N,\statesetA_E,\mathbf{L},\istateA,\transitionRel,\{\astateA\})\) for \(\mathcal{P}\) composed of \(\mathcal{P}_1\) and \(\mathcal{P}_2\).

\begin{itemize}
	\item \(\mathcal{P} = \mathcal{P}_1\mathcal{P}_2\).
	Recall that the last node matched by \(\mathcal{P}_1\) corresponds to the first node matched by \(\mathcal{P}_2\), following the definition of a concatenation of two paths.
	Therefore, we \say{merge} the incoming transitions of \(\astateA^1\) with the outgoing transitions of \(\istateA^2\).
	\begin{multline*}
		\statesetA_N = \statesetA^1_N \cup \statesetA^2_N, \quad \statesetA_E = \statesetA^1_E \cup \statesetA^2_E, \quad \istateA = \istateA^1, \quad \astateA^2 = \astateA^2,\\
		\transitionRel = \transitionRel^1 \cup \transitionRel^2 \cup \{ (\stateA^1,\varphi\land\psi,\stateA^2) \mid (\stateA^1,\varphi,\astateA^1)\in\transitionRel^1, (\istateA^2,\psi,\stateA^2)\in\transitionRel^2\}
	\end{multline*}

	\item \(\mathcal{P} = \mathcal{P}_1 \cup \mathcal{P}_2\).
	Since initial states have only outgoing and accepting states only incoming transitions, these transitions can simple be copied for a fresh initial state \(\istateA\) and a fresh accepting state \(\astateA\).
	\begin{multline*}
		\statesetA_N = \statesetA^1_N \cup \statesetA^2_N \cup \{\istateA\}\\
		\statesetA_E = \statesetA^1_E \cup \statesetA^2_E \cup \{\astateA\},\\
		\transitionRel = \transitionRel^1 \cup \transitionRel^2 \cup \{(\istateA,\varphi,\stateA) \mid (\istateA^1,\varphi,\stateA)\in\transitionRel^1 \text{ or } (\istateA^2,\varphi,\stateA)\in \transitionRel^2\}\\
		\cup \{(\stateA,\varphi,\astateA) \mid (\stateA,\varphi,\astateA^1)\in\transitionRel^1 \text{ or } (\stateA,\varphi,\astateA^2)\in \transitionRel^2\}
		\\
		\cup \{(\istateA,\varphi,\astateA) \mid (\istateA^1,\varphi,\astateA^1)\in\transitionRel^1 \text{ or } (\istateA^2,\varphi,\astateA^2)\in \transitionRel^2\}
	\end{multline*}

	\item \(\mathcal{P} = \mathcal{P}_1^*\).
	The state sets are \(\statesetA_N = \statesetA^1_N \cup \{\istateA\}\) and \(\statesetA_E = \statesetA^1_E \cup \{\astateA\}\) where \(\istateA\) and \(\astateA\) are the fresh initial and accepting states, respectively.
	The transition relation is \(\transitionRel = \transitionRel^1 \cup \transitionRel_\varepsilon \cup \transitionRel_{\ge 1} \cup \transitionRel_{\ell}\) where \(\transitionRel_{\varepsilon}\), \(\transitionRel_{\ge 1}\), and \(\transitionRel_{\ell}\) are defined as follows.

	The case of \(\mathcal{P}\) matching an empty path, i.e.\ the case of zero repetitions, is covered by \(\transitionRel_\varepsilon\).
	Recall that paths of length zero consists of one node (with an arbitrary label).
	Thus, we set \(\transitionRel_\varepsilon = \{(\istateA,\top,\astateA)\}\).

	For the other cases the transitions of \(\istateA^1\) and \(\astateA^1\) are copied to \(\istateA\) and \(\astateA\).
	\begin{multline*}
		\transitionRel_{\ge 1} =
		\{ (\istateA,\varphi,\stateA^1) \mid (\istateA^1,\varphi,\stateA^1)\in\transitionRel^1\}\\ \cup
		\{ (\stateA^1,\varphi,\astateA) \mid (\stateA^1,\varphi,\astateA^1)\in\transitionRel^1\}
	\end{multline*}
	Finally, we allow the automaton to \say{loop}.
	Similar to our construction for the concatenation, transitions have to be merged.
	\[
	\transitionRel_{\ell} =
	\{ (\stateA^1_s,\varphi\land\psi,\stateA^1_t) \mid (\stateA^1_s,\varphi,\astateA^1)\in\transitionRel^1, (\istateA^1,\psi,\stateA^1_t)\in\transitionRel^1\}
	\]
\end{itemize}

\paragraph{Adding Reverse Traversals}
	For handling reverse traversals, i.e.\ patterns of the form \(\gpceprev{\(\psi\)}\), we will understand the direction of an edge as an additional label.
	More precisely, we consider two special labels \textsf{fw} and \textsf{bw} for forward and reverse traversals, respectively, which do not occur in the pattern \(\mathcal{P}\) (or the underlying property graph).
	For a pattern \(\gpcep{\(\varphi\)}\), we then construct the transition \((\stateA_s, \textsf{fw} \land \varphi, \stateA_t)\) instead of \((\stateA_s,\varphi, \stateA_t)\).
	The automaton for a pattern \(\gpceprev{\(\psi\)}\) can then be constructed analogously, yielding in the transition \((\stateA_s, \textsf{bw} \land \varphi, \stateA_t)\).
	Observe that label expressions of transitions for edge patterns are never combined in the overall construction (only label expressions for node labels are).
	Consequently, all outgoing transitions of states in \(\statesetA_E\), i.e.\ states that indicate that the next transitions reads labels of an edge, have either label expressions of the form \(\textsf{fw} \land \varphi\) or of the form \(\textsf{bw} \land \varphi\), where \(\varphi\) and \(\psi\) contain neither \textsf{fw} nor \textsf{bw}.
	Outgoing transitions of states in \(\statesetA_N\) do not contain the special labels \textsf{fw} and \textsf{bw} at all, since they read only node labels.

	For adapting the semantics, we define the \emph{extended trace} of a path \(P = v_0e_1v_1\ldots e_nv_n\) as \[\lambda^*(P) = \lambda(v_0)\lambda^*(v_0,e_1,v_1)\lambda(v_1)\ldots \lambda^*(v_{n-1},e_n,v_n)\lambda(v_n),\] where \(\lambda(v_i)\) is, as before, the set of labels of node \(v_i\), and \[\lambda^*(v_{i-1}, e_i, v_i) =
		\begin{cases}
			\lambda(e_i) \cup \{\textsf{fw}\}, & \text{ if }\rho(e_i) = (v_{i-1}, v_i),\\
			\lambda(e_i) \cup \{\textsf{bw}\}, & \text{ if }\rho(e_i) = (v_{i}, v_{i-1})
		\end{cases}
	\]
	extends the set \(\lambda(e_i)\) of labels of an edge \(e_i\) by either \textsf{fw} or \textsf{bw}, depending on whether \(e_i\) is traversed in a forward or reverse fashion in the path \(P\).

	An RGPC automaton (with reverse traversals) then accepts a path \(P = v_0e_1v_1\ldots e_nv_n\) if it accepts the extended trace \(\lambda^*(P)\) of that path.

\subsection{Error Repairing}\label{section:error-repairing-proofs}
This section is dedicated to the proofs of Propositions~\ref{result:topological-repair-ind-set} and~\ref{result:topological-repair-greedy}.

\begin{proof}[Proof of Proposition~\ref{result:topological-repair-ind-set}]
	Let \(\Sigma\) be a set of RGPC constraints, \(\graph\) be a property graph, and \(V\) be a minimum weight vertex cover of the (topological) conflict hypergraph of \(\graph\) w.r.t.\ \(\Sigma\).
	Furthermore, let \(\graph'\) be the subgraph of \(\graph\) obtained by removing all nodes and edges in \(V\), as well as all edges incident to the nodes in \(V\).
	We prove that \(\graph'\) is a repair by showing that conditions~a) and~b) from Definition~\ref{definition:repair} hold by contradiction.

	For the sake of contradiction assume that Condition~a) does not hold, that is \(\graph'\) violates a constraint \(\query; F \cimp C\) from \(\Sigma\).
	Then there is a match \(\match\) of \(\query\) in \(\graph'\).
	Let \(O_\match\) be the topological error induced by \(\match\).
	Since all nodes and edges of \(\graph'\) also appear in \(\graph'\) with the same labels and properties -- \(\graph'\) is a topological subgraph -- \(O_\match\) is also a topological error of \(\graph\).
	But then it is also a hyperedge of the conflict hypergraph of \(\graph\), and since all its elements are still present in \(\graph'\), its intersection with \(V\) is empty.
	This is a contradiction to \(V\) being a vertex cover.
	We conclude that Condition~a) holds.

	For proving that Condition~b) holds, we assume, again for the sake of contradiction, that there is a subgraph \(\graph''\) of \(\graph\) such that \(\graph'\) is a proper subgraph of \(\graph''\) and \(\graph''\) satisfies all constraints from~\(\Sigma\).
	Then there is a node or edge \(o\) in \(\graph''\) that is not present in \(\graph'\).
	We make a case distinction.

	In case \(o\in V\), consider all topological errors \(O_1,\ldots,O_m\) of \(\graph\) containing object \(o\).
	Since \(\graph''\) satisfies all constraints, there are objects \(o_1\in O_1,\ldots,o_m\in O_m\), which are distinct from \(o\), and not present in \(\graph''\).
	But since \(\graph'\) is a subgraph of \(\graph''\), the objects \(o_1,\ldots,o_m\) are also not present in \(\graph'\).
	Furthermore, we can assume that every \(o_j\) is in \(V\): If not, then \(o_j\) is an edge which was removed because one of its endpoints, say a node \(n\), was removed.
	But then \(O_j\) also contains \(n\), since errors originate from paths.
	Moreover, \(n\) cannot be \(o\) because the weight of \(n\) is higher than the sum of the weights of its incident edges, and hence higher than the sum of weights of all \(o_j\) not in \(V\).
	Consequently, \(n\) being in \(V\) but some edges \(o_j\) not being in \(V\) would contradict \(V\) having minimum weight.
	Therefore, we can indeed assume that \(o_j\) is \(n\), and hence in \(V\).
	In addition, since all \(o_j\) are in \(V\), we can conclude that \(V\setminus \{o\}\) is a vertex cover.
	And since it has a lesser weight than \(V\) this contradicts the fact that \(V\) is a minimal vertex cover.

	It remains to consider the case \(o\notin V\).
	Then \(o\) is an edge that was removed because one of its endpoints \(n\) was removed due to being in~\(V\).
	Since \(\graph''\) contains \(o\), it also has to contain its endpoint \(n\).
	Therefore, applying the first case for \(n\in V\) yields a contradiction.

	Overall, we can conclude that Condition~b) holds.
\end{proof}

\begin{proof}[Proof of Proposition~\ref{result:topological-repair-greedy}]
	We proceed similarly as in the proof of Proposition~\ref{result:topological-repair-ind-set}.
	Let \(\Sigma\) be a set of RGPC constraints, \(\graph\) be a property graph, and \(V_g\) be a minimum vertex cover of the (topological) conflict hypergraph of \(\graph\) w.r.t.\ \(\Sigma\), computed by the greedy algorithm.
	Furthermore, let \(\graph'\) be the subgraph of \(\graph\) obtained by removing all nodes and edges from \(V_g\), as well as all edges incident to nodes in \(V_g\).
	We prove that \(\graph'\) is a repair by showing that conditions~a) and~b) from Definition~\ref{definition:repair} hold.
	The proof of Condition~a) being satisfied is exactly the same as in the proof for Proposition~\ref{result:topological-repair-ind-set}, since it only relies on \(V_g\) being a vertex cover.

	To prove that Condition~b) holds, we assume, for the sake of contradiction, that there is a subgraph \(\graph''\) of \(\graph\) such that \(\graph'\) is a proper subgraph of \(\graph''\) and \(\graph''\) satisfies all constraints from~\(\Sigma\).
	Then there is a node or edge \(o\) in \(\graph''\) that is not present in \(\graph'\).

	In case \(o\in V_g\), consider all topological errors \(O_1,\ldots,O_m\) of \(\graph\) with \(V_g \cap O_i = \{o\}\).
	Note that there is at least one such \(O_i\), because otherwise \(o\) would have been removed from \(V_g\) in the second phase of the greedy algorithm.
	Furthermore, for at least one of these errors \(O_j\) we have that \(o\) is the object with the least weight among all objects in \(O_j\).
	If this would not be the case, the greedy algorithm would have picked elements \(o_1\in O_1,\ldots,o_m\in O_m\) with a smaller weight than \(o\) in the first phase and consequently removed \(o\) from \(V_g\) in the second phase prior to removing \(o_1,\ldots,o_m\), since \(o\) has a greater weight.
	We fix one of these errors \(O_j\) for which \(o\) has minimal weight among the objects in \(O_j\).
	Since \(V_g \cap O_j = \{o\}\), all objects \(O_j\setminus \{o\}\) are in \(\graph'\), and hence also in \(\graph''\) since \(\graph'\) is a subgraph of \(\graph''\).
	Note that this is particularly true if \(o\) is a node: Since \(o\) has minimal weight, \(O_j\) cannot contain an edge in that case.
	But since \(\graph''\) also contains \(o\), it contains all objects of \(O_j\).
	This is a contradiction because \(O_j\) originates from a match \(\match\) of the pattern \(\query\) of a constraint \(\query; F \cimp C\) in \(\Sigma\) that passes \(F\) but does not satisfy \(C\).
	Thus, \(\graph''\) violates the constraint.

	It remains to consider the case \(o\notin V_g\).
	Then \(o\) is an edge that was removed because one of its endpoints \(n\) was removed due to being in~\(V_g\).
	Since \(\graph''\) contains \(o\), it also has to contain its endpoint \(n\).
	Therefore, applying the first case for \(n\in V_g\) yields a contradiction.

	Overall, we can conclude that Condition~b) holds.
\end{proof}

\subsection{ILPs with Explicit Dependencies}\label{section:explicit-ilps}
In this section, we discuss an alternative to the ILPs presented in Section~\ref{section:cleaning:repairing}, and why we did \emph{not} choose it.
Thanks to the use of weights in the ILPs presented in Section~\ref{section:cleaning:repairing}, it is ensured that nodes are only deleted if necessary, and deleting (a subset of its) incident edges takes precedence.
One potential issue with this approach is, however, that the weights are \emph{static} (i.e.\ hardcoded in the ILP):
For example, consider a constraint that forbids the presence of two different nodes with the same value for the \pgkey{name} property.
Suppose there are two nodes \(n_1,n_2\) with two and five incident edges, respectively, violating this constraint.
Since \(n_1\) has weight \(3\) and \(n_2\) weight \(6\), an ILP solver will always mark \(n_1\) for deletion.
This is, in particular, true even if, due to another constraint, the ILP also marked all incident edges of \(n_2\) for deletion (and no incident edge of \(n_1\)).
But in this case, the \say{effective weight} of \(n_2\) is only \(1\), and thus \(n_2\) should be marked for deletion instead of \(n_1\).
This can be addressed by constructing and solving an alternative ILP, in which dependencies are encoded explicitly.
That is, if \(n_1\) has the two incident edges \(e_1\) and \(e_2\) in the property graph, then we add the constraint \(2x_{n_1} \le x_{e_1} + x_{e_2}\) to the ILP.
Thanks to this constraint, the variables for \(x_{e_1}\) and \(x_{e_2}\) for the edges have to be assigned value \(1\) in a solution whenever the variable \(x_{n_1}\) for \(n_1\) is assigned value \(1\) (recall that value \(1\) means that the object is deleted).
Weights are then no longer necessary and the ILP takes the following form, where \(E_n\) is the set of incident edges of a node \(n\in\nodeset\) of the property graph.\footnote{Strictly speaking, the elements of each \(E_n\) with \(n\in\hnodeset\) have to be added to the set \(\hnodeset\) of vertices of the conflict hypergraph.}
\begin{align*}
	\textnormal{minimize } & \quad\sum_{v\in \hnodeset} x_v\\
	\textnormal{subject to }& \quad\sum_{v\in O} x_v \ge 1 \quad \textnormal{ for all } O\in\hedgeset\\
	\textnormal{and } & \quad \abs{E_n}\cdot x_n \le \sum_{e\in E_n} x_e \quad \textnormal{ for all } n\in \nodeset \cap \hnodeset\\
	\textnormal{and } & \quad  0 \le x_v \le 1 \quad \textnormal{ for all } v\in\hnodeset
\end{align*}
This variant trades a larger ILP for a \emph{potentially} better repair.
However, in our experiments, it took significantly longer to solve the ILP with explicit dependencies in all cases, while the repairs obtained were the same as for the ILP with weights.

\section{Additional Material for Experiments}\label{appendix:experiments}
In Appendix~\ref{section:additional-plots} we present plots for combinations of optional steps/datasets/constraints not shown in Section~\ref{section:experiments}.
Appendix~\ref{section:raw-numbers} contains the raw numbers used for the plots, as well as the runtime increase and deletion reduction percentages, in particular those reported in the paper.

\subsection{Additional Experimental Results}\label{section:additional-plots}
We provide additional plots akin to those in Section~\ref{section:experiments}, for the remaining combinations of optional steps/datasets/constraints.

\paragraph{Scaling}
In Section~\ref{section:performance-quality} we provide plots for showcasing how the base pipeline scales on the ICIJ and LDBC graphs with \(1\)-way constraints.
Figure~\ref{figure:scaling-1way-missing} shows similar plots for the legislative and Coreutils datasets.
As observed in the paper for the ICIJ and LDBC datasets, the number of deletions is comparatively small, and the pipeline scales well.

\begin{figure*}
	\centering%
	\begingroup%
	\renewcommand{\plotscale}{.33}%
	\begin{subfigure}{.245\linewidth}
		\centering%
		\includegraphics[scale=\plotscale]{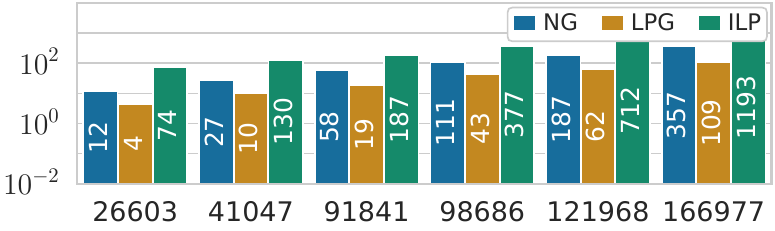}
		\caption{Legislative: Runtime}\label{figure:legislative-err-c-runtime}
	\end{subfigure}
	\hfill
	\begin{subfigure}{.245\linewidth}
		\centering%
		\includegraphics[scale=\plotscale]{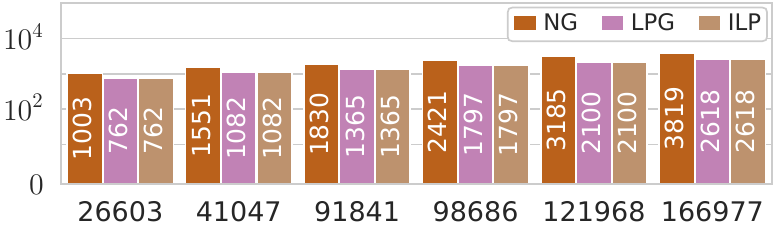}
		\caption{Legislative: Edge deletions}\label{figure:legislative-err-c-deletions}
	\end{subfigure}
\hfill
\begin{subfigure}{.245\linewidth}
	\centering%
	\includegraphics[scale=\plotscale]{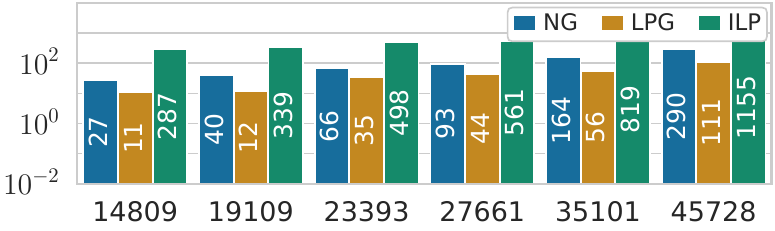}
	\caption{Coreutils: Runtime}\label{figure:coreutils-err-c-runtime}
\end{subfigure}
\hfill
\begin{subfigure}{.245\linewidth}
	\centering%
	\includegraphics[scale=\plotscale]{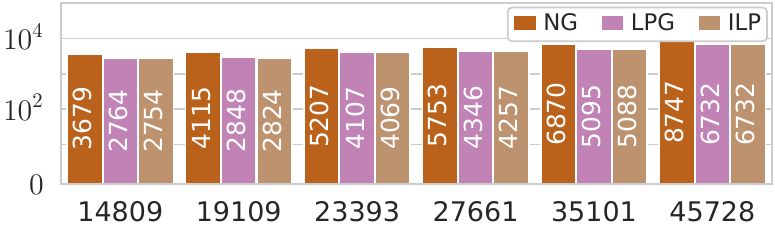}
	\caption{Coreutils: Edge deletions}\label{figure:coreutils-err-c-deletions}
\end{subfigure}
\endgroup%
\caption{Runtime and number of edge deletions for LDBC and the Coreutils graph for subsets of \(1\)-way constraints}%
\label{figure:coreutils-err-c}\label{figure:legislative-err-c}\label{figure:scaling-1way-missing}
\end{figure*}

We also performed scaling experiments for \(1\)-way constraints with Step~\(2\) enabled.
Figures~\ref{figure:runtimes-err-c-label} and~\ref{figure:deletions-err-c-label} show the runtime and number of label deletions, if Step~2 is enabled.
They assert that, as for the base pipeline, our pipeline performs well when the number of required deletions is small if Step~\(2\) is enabled, although the runtime is consistently higher than for the base pipeline, confirming our observations in Section~\ref{section:ablation}.

\paragraph{Approximate Repairs}
In Section~\ref{section:performance-quality} we investigated the approximate repairs yielded by the selection phase of the naive and LP-guided greedy algorithms.
For the ICIJ graph we discussed that the selection phase of the naive greedy performs generally better than the one for the LP-guided greedy algorithm, but overall the naive greedy algorithm performs worse, since it has to trim more objects in the trimming phase.
The same applies to the other datasets as shown in Figures~\ref{figure:classwise-runtimes-approx} and~\ref{figure:classwise-deletions-approx}, and in comparison with Figures~\ref{figure:legislative-classwise-runtime}, \ref{figure:coreutils-classwise-runtime}, \ref{figure:icij-classwise-runtime}, and~\ref{figure:ldbc-sf0.3-classwise-runtime}, which show total runtime, and Figures~\ref{figure:legislative-classwise-edge-deletions}, \ref{figure:coreutils-classwise-edge-deletions}, \ref{figure:icij-classwise-edge-deletions}, and~\ref{figure:ldbc-sf0.3-classwise-edge-deletions}, which show the edge deletions, in the paper for the full algorithms computing repairs.
As for the ICIJ dataset in the paper, we can also observe that the LP-guided greedy algorithms often achieves optimal quality (and thus actually computes a repair), while the naive greedy does not.
However, when the naive greedy comes close to an optimal solution in the selection phase, it can outperform the LP-guided greedy when computing a repair (and not only an approximation).

\begin{table*}
		\caption{Total edge weight for repairs proposed by the LP-guided greedy algorithm when the pipeline does or does not use the custom weights. Here the weight of an edge is its PageRank in the line graph. For readability all weights are rounded.}\label{table:custom-weights-page-rank}
	\begin{tabular}{l c c c c c c r}
		\toprule%
		&&& \multicolumn{2}{c}{\textbf{Repair w/o custom weights}} & \multicolumn{2}{c}{\textbf{Repair with custom weights}} & \\
		Dataset & Shape & Original Weight & Repair Weight & Weight Loss & Repair Weight & Weight Loss & Loss Reduction \\
		\midrule
		\multirow{5}{*}{ICIJ} & \(1\)-way & \multirow{5}{*}{579593.37} & 579060.61 & 0.0919\% & 579188.60 & 0.0698\% & 24\% \\
		& \(2\)-way & & 579233.32 & 0.0621\% & 579261.99 & 0.0572\% & 7\% \\
		& \(2\)-rep & & 579548.82 & 0.0077\% & 579569.60 & 0.0041\% & 46\% \\
		& loop & & 579143.61 & 0.0776\% & 579219.41 & 0.0645\% & 16\% \\
		& \(3\)-split & & 579590.94 & 0.0004\% & 579590.94 & 0.0004\% & 0\% \\
		\midrule
		\multirow{2}{*}{Coreutils} & \(1\)-way & \multirow{2}{*}{112369.10} & 110873.74 & 1.3308\% & 110920.13 & 1.2895\% & 3\% \\
		& \(2\)-way &  & 112203.69 & 0.1472\% & 112225.11 & 0.1281\% & 13\% \\

		\bottomrule
	\end{tabular}
\end{table*}

\paragraph{Minimizing Deletions of Important Objects}
In Section~\ref{section:custom-weights} we assessed the robustness of our pipeline with respect to custom weights.
More precisely, we used the PageRank scores for nodes and edges as custom weights for repairing the ICIJ graph, and discussed that the weight loss of a repair can be reduced by \(24\%\) when using custom weights.
This indicates that more important edges, i.e.\ edges with a high PageRank score, are preserved.
Table~\ref{table:custom-weights-page-rank} shows the total weights, weight losses, and loss reductions for all sets of constraints.
The loss reduction is computed as
\[
	100 \cdot \left(
		1 - \frac{w_\texttt{orig} - w_\texttt{cw-repair}}%
					{w_\texttt{orig} - w_{\texttt{repair}}}
	\right)
\]
where \(w_\texttt{orig}\) is the total weight of the input graph, \(w_\texttt{cw-repair}\) is the total weight of a repair when custom weights are provided, and \(w_{\texttt{repair}}\) is the total weight of a repair when \emph{no} custom weights are provided.
For \(2\)-rep constraints the loss reduction is even \(46\%\), albeit the weight loss is lower for repairs without custom weights than for the case of \(1\)-way constraints discussed in the paper.

Table~\ref{table:custom-weights-page-rank} also shows the experimental results for the Coreutils graph, and
Figures~\ref{figure:coreutils-classwise-runtime-custom-weights} and~\ref{figure:coreutils-classwise-edge-deletions-custom-weights} show the runtime and edge deletions for the Coreutils graph when PageRank scores are provided as custom weights.
Compared with Figures~\ref{figure:coreutils-classwise-runtime} and~\ref{figure:coreutils-classwise-edge-deletions} in the paper they confirm our findings for the ICIJ graph.

We used the Graph Data Science Library of Neo4j to compute the PageRank scores.
Since PageRank scores are, as usual, computed for nodes, we computed the PageRank scores for edges using line graphs, in which the role of nodes and edges is, intuitively, swapped: The \emph{line graph} \(\graph_L = (V_L, E_L)\) of a graph \(\graph = (V, E)\) is the graph with \(V_L = E\), and a pair \((e_1,e_2)\in V_L\times V_L\) is in \(E_L\) if and only if there are \(v_0,v_1,v_2 \in V\) such that \(v_0e_1v_1e_2v_2\) is a path in \(\graph\).
For the legislative and LDBC graph computing the line graph for computing edge weights was not possible with our hardware using Neo4j.

\begin{figure}
	\centering%
	\begin{subfigure}{.45\linewidth}
		\centering%
		\includegraphics[scale=\plotscale]{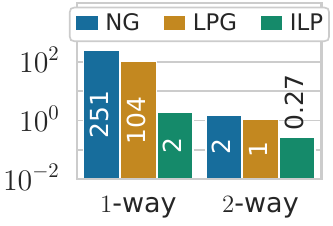}
		\caption{Runtime}\label{figure:coreutils-classwise-runtime-custom-weights}
	\end{subfigure}
	\hfill
	\begin{subfigure}{.45\linewidth}
		\centering%
		\includegraphics[scale=\plotscale]{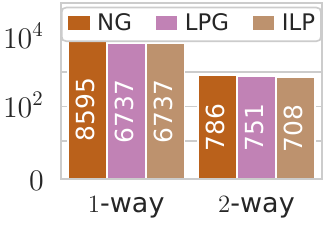}
		\caption{Edge Deletions}\label{figure:coreutils-classwise-edge-deletions-custom-weights}
	\end{subfigure}
	\hfill
	\caption{Runtimes (in seconds) and deletions for the Coreutils graph when using PageRank scores as custom weights}%
\label{figure:classwise-runtimes-custom-weights}
\end{figure}

\paragraph{Step~2/Label Deletions}
Figure~\ref{figure:classwise-label-deletions-appendix} shows the runtimes and number of edge deletions of our repair pipeline with Step~\(2\) enabled for the Coreutils and LDBC datasets and the sets of constraints outlined in Table~\ref{table:static-pgs}.
Compared to the corresponding plots for the base pipeline in Section~\ref{section:qualitative-study}, they provide the same insight as the analogous plots for the legislative and ICIJ graphs discussed in Section~\ref{section:ablation}: With Step~2 enabled the repair pipeline yields higher-quality repairs (measured in the number of deletions), at the expense of a higher runtime.

\paragraph{Step~3/Neighbourhood Errors}
Figures~\ref{figure:legislative-classwise-nbh}, \ref{figure:coreutils-classwise-nbh}, \ref{figure:coreutils-classwise-nbh}, and~\ref{figure:ldbc-sf0.3-classwise-nbh} show the plots for the neighbourhood error repairs (Step~\(3\)), except for the combination of ICIJ graph and the LP-guided greedy algorithm, for which they are shown in Section~\ref{section:ablation}.
They are supporting our insights discussed in Section~\ref{section:ablation}: If \(k\) is chosen well, a runtime improvement can be achieved in almost all cases, while retaining a high quality.
For \(2\)-rep constraints (which are only available for the ICIJ and LDBC graphs), \(k = 4\) yields a quality close (or identical) to the base pipeline, while for the other constraints \(k = 2\) already suffices.
This effect can be explained by the shape of our constraints: The minimal length of a path within a match of the \(2\)-rep constraints is~\(9\), while it is at most~\(6\) for the other constraints.
For matches where all paths are of length at most~\(6\) the complete match is within the \(4\)-neighbourhood, and the \(2\)-neighbourhood covers most of it.
Furthermore, since our constraints exhibit repetitions, it is likely that larger errors are resolved by repairing smaller ones, in particular those containing paths of length at most~\(6\).
For \(2\)-rep constraints, the \(4\)-neighbourhood covers matches with minimal-length paths almost completely but the \(2\)-neighbourhood does not, explaining why for \(k = 4\) the result is very close to the base case, and worse for \(k = 2\) in this case.
This suggests that a good choice for \(k\) is \(k = \frac{\ell_\text{min}}{2}\) where \(\ell_\text{min}\) is the minimal length of a path matched by a path pattern of the constraint.

	Another noteworthy observation is that the naive greedy algorithm can yield a better quality with Step~3 enabled than without.
	For example, this happens for \(1\)-way constraints and \(k = 4\) for the legislative graph (Figure~\ref{figure:legislative-classwise-nbh-deletions-greedy} vs.\ Figure~\ref{figure:legislative-classwise-edge-deletions} in the paper).
	The underlying reason is that the naive greedy selects \emph{any} object from an error for removal: Truncating an error can thus change which object is selected and can lead to a higher-quality repair (but also to lower-quality repairs in others).
	For the LP-guided greedy and ILP strategies we cannot observe this effect, since these strategies can better avoid bad selections thank to the (I)LP solutions.

Finally, Figure~\ref{figure:classwise-sample-runtimes} shows the results when running the pipeline with random samples of \(2k\) edges from each topological error, except for the ICIJ graph, for which they are shown in Section~\ref{section:ablation} in the paper.
They confirm our observation discussed in the paper for the ICIJ graph: taking random samples can results in significantly more deletions.
We conclude that picking neighbourhoods around path endpoints is a more appropriate strategy.

\begin{figure}
	\centering%
	\begin{subfigure}{.32\linewidth}
		\centering%
		\includegraphics[scale=\plotscale]{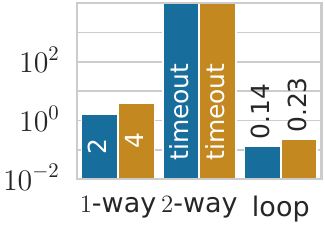}
		\caption{Legislative}\label{figure:legislative-classwise-runtime-approx}
	\end{subfigure}
	\hfill
	\begin{subfigure}{.29\linewidth}
		\centering%
		\includegraphics[scale=\plotscale]{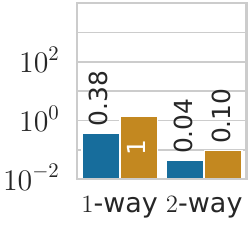}
		\caption{Coreutils}\label{figure:coreutils-classwise-runtime-approx}
	\end{subfigure}
	\hfill
	\begin{subfigure}{.35\linewidth}
		\centering%
		\includegraphics[scale=\plotscale]{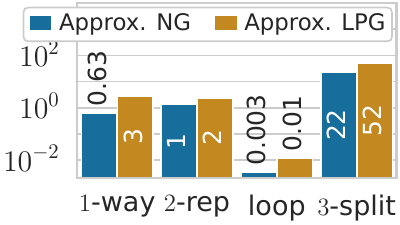}
		\caption{LDBC}\label{figure:ldbc-classwise-runtime-approx}
	\end{subfigure}
	\caption{Runtimes in seconds of the selection phase of the naive greedy (NG) and LP-guided greedy (LPG)}%
\label{figure:classwise-runtimes-approx}
\end{figure}

\begin{figure}
	\centering%
	\begin{subfigure}{.32\linewidth}
		\centering%
		\includegraphics[scale=\plotscale]{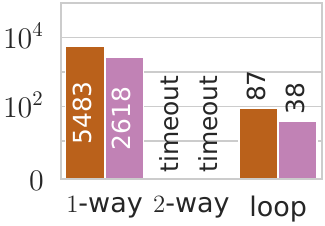}
		\caption{Legislative}\label{figure:legislative-classwise-deletions-approx}
	\end{subfigure}
	\hfill
	\begin{subfigure}{.29\linewidth}
		\centering%
		\includegraphics[scale=\plotscale]{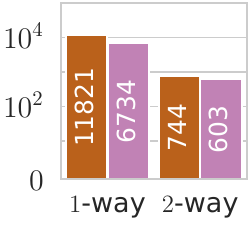}
		\caption{Coreutils}\label{figure:coreutils-classwise-deletions-approx}
	\end{subfigure}
	\hfill
	\begin{subfigure}{.35\linewidth}
		\centering%
		\includegraphics[scale=\plotscale]{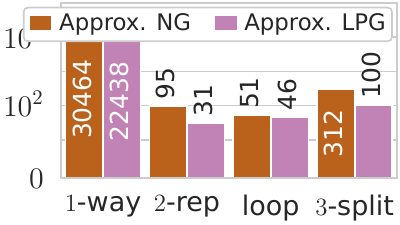}
		\caption{LDBC}\label{figure:ldbc-classwise-deletions-approx}
	\end{subfigure}
	\caption{Number of edge deletions of the selection phase of the naive greedy (NG) and LP-guided greedy (LPG)}%
\label{figure:classwise-deletions-approx}
\end{figure}

\begin{figure*}
	\centering%
	\begingroup%
	\renewcommand{\plotscale}{.36}%
	\begin{subfigure}{.27\linewidth}
		\centering%
		\includegraphics[scale=\plotscale]{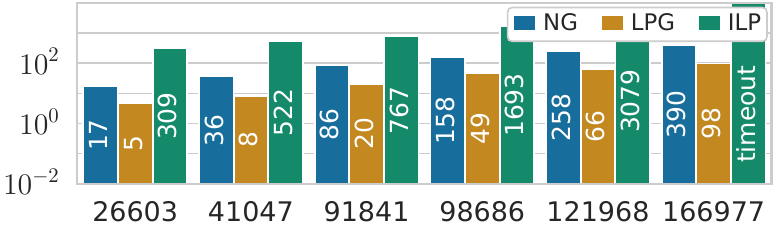}
		\caption{Legislative}\label{figure:legislative-err-c-label-runtime}
	\end{subfigure}
	\hfill
	\begin{subfigure}{.27\linewidth}
		\centering%
		\includegraphics[scale=\plotscale]{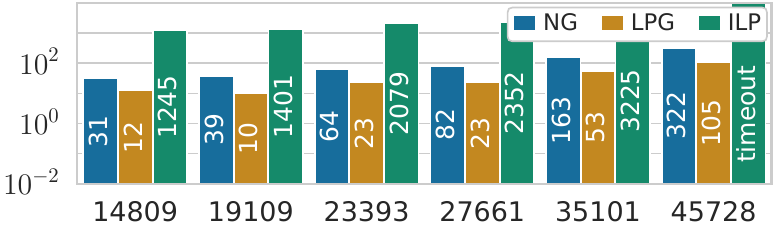}
		\caption{Coreutils}\label{figure:coreutils-err-c-label-runtime}
	\end{subfigure}
	\hfill
	\begin{subfigure}{.23\linewidth}
		\centering%
		\includegraphics[scale=\plotscale]{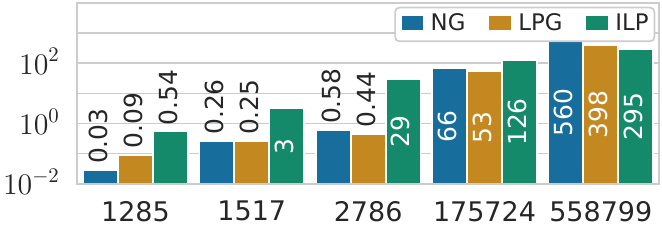}
		\caption{ICIJ}\label{figure:icij-err-c-label-runtime}
	\end{subfigure}
	\hfill
	\begin{subfigure}{.19\linewidth}
		\centering%
		\includegraphics[scale=\plotscale]{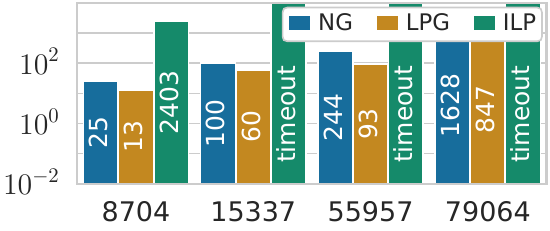}
		\caption{LDBC}\label{figure:ldbc-err-c-label-runtime}
	\end{subfigure}
	\endgroup%
	\caption{Runtimes of the repair pipeline with Step~\(2\) (label deletions) enabled for subsets of \(1\)-way constraints}%
	\label{figure:runtimes-err-c-label}
\end{figure*}

\begin{figure*}
	\centering%
	\begingroup%
	\renewcommand{\plotscale}{.36}%
	\begin{subfigure}{.27\linewidth}
		\centering%
		\includegraphics[scale=\plotscale]{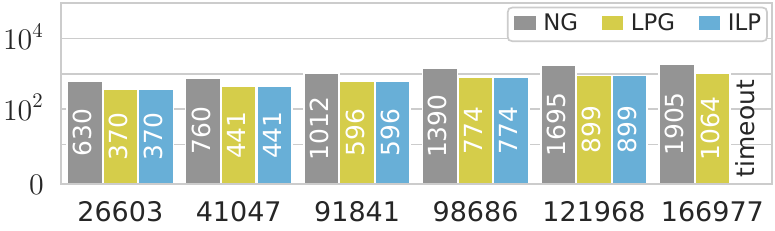}
		\caption{Legislative}\label{figure:legislative-err-c-label-deletions}
	\end{subfigure}
	\hfill
	\begin{subfigure}{.27\linewidth}
		\centering%
		\includegraphics[scale=\plotscale]{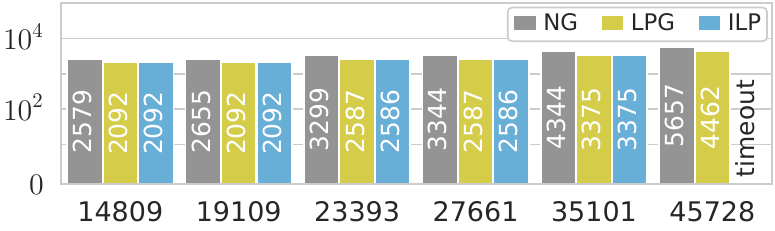}
		\caption{Coreutils}\label{figure:coreutils-err-c-label-deletions}
	\end{subfigure}
	\hfill
	\begin{subfigure}{.23\linewidth}
		\centering%
		\includegraphics[scale=\plotscale]{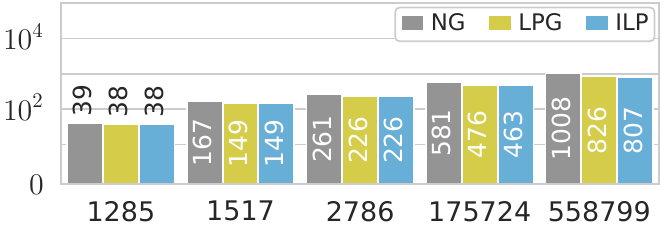}
		\caption{ICIJ}\label{figure:icij-err-c-label-deletions}
	\end{subfigure}
	\hfill
	\begin{subfigure}{.19\linewidth}
		\centering%
		\includegraphics[scale=\plotscale]{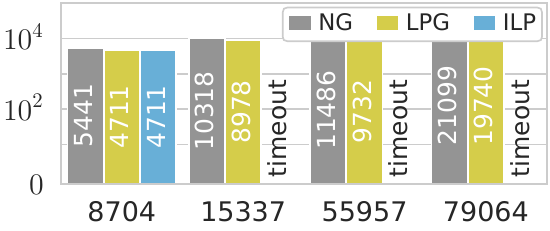}
		\caption{LDBC}\label{figure:ldbc-err-c-label-deletions}
	\end{subfigure}
	\endgroup%
	\caption{Label deletions proposed by the repair pipeline with Step~\(2\) (label deletions) enabled for subsets of \(1\)-way constraints}%
\label{figure:deletions-err-c-label}
\end{figure*}

\begin{figure*}
	\centering%
	\begin{subfigure}{.245\linewidth}
		\includePlot{Coreutils: Runtimes}{coreutils}{classwise-label-runtime}
	\end{subfigure}
	\hfill
	\begin{subfigure}{.245\linewidth}
		\includePlot{LDBC: Runtimes}{ldbc-sf0.3}{classwise-label-runtime}
	\end{subfigure}
	\hfill
	\begin{subfigure}{.245\linewidth}
		\includePlot{Coreutils: Label deletions}{coreutils}{classwise-label-deletions}
	\end{subfigure}
	\hfill
	\begin{subfigure}{.245\linewidth}
		\includePlot{LDBC: Label deletions}{ldbc-sf0.3}{classwise-label-deletions}
	\end{subfigure}
	\caption{Runtime with Step~2 enabled and number of label deletions for the Coreutils graph and LDBC}%
\label{figure:classwise-label-deletions-appendix}
\end{figure*}

\begin{figure*}
	\centering%
	\begin{subfigure}{.16\linewidth}
		\centering%
		\includegraphics[scale=\plotscale]{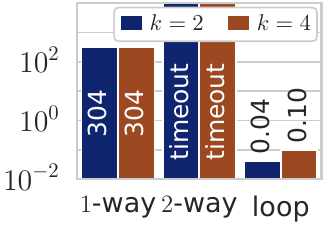}
		\caption{NG: Runtime}
		\label{figure:legislative-classwise-nbh-runtime-greedy}
	\end{subfigure}
	\hfill%
	\begin{subfigure}{.16\linewidth}
		\centering%
		\includegraphics[scale=\plotscale]{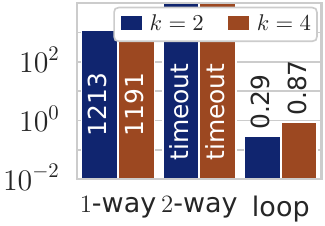}
		\caption{ILP: Runtime}
		\label{figure:legislative-classwise-nbh-runtime-ilp}
	\end{subfigure}
	\hfill%
	\begin{subfigure}{.16\linewidth}
		\centering%
		\includegraphics[scale=\plotscale]{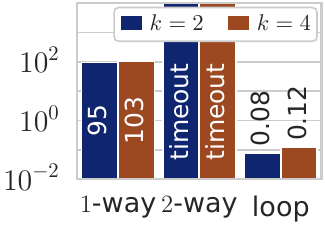}
		\caption{LPG: Runtime}
		\label{figure:legislative-classwise-nbh-runtime-lpr}
	\end{subfigure}
	\hfill%
	\begin{subfigure}{.16\linewidth}
		\centering%
		\includegraphics[scale=\plotscale]{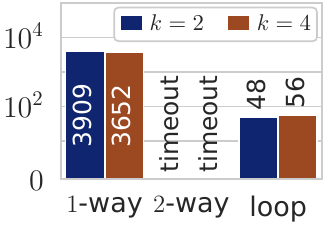}
		\caption{NG: Edge deletions}
		\label{figure:legislative-classwise-nbh-deletions-greedy}
	\end{subfigure}
	\hfill%
	\begin{subfigure}{.16\linewidth}
		\centering%
		\includegraphics[scale=\plotscale]{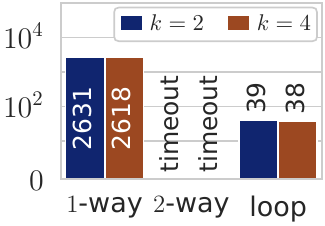}
		\caption{ILP: Edge deletions}
		\label{figure:legislative-classwise-nbh-deletions-ilp}
	\end{subfigure}
	\hfill%
	\begin{subfigure}{.16\linewidth}
		\centering%
		\includegraphics[scale=\plotscale]{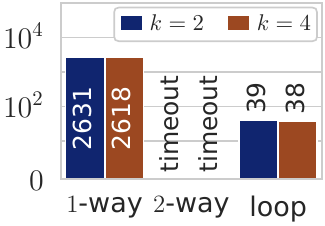}
		\caption{LPG: Edge deletions}
		\label{figure:legislative-classwise-nbh-edge-deletions-lpr}
	\end{subfigure}
	\hfill%
	\caption{Runtime (in seconds) and edge deletions for neighbourhood errors for the legislative graph}\label{figure:legislative-classwise-nbh}
\end{figure*}

\begin{figure*}
	\centering%
	\begin{subfigure}{.16\linewidth}
		\centering%
		\includegraphics[scale=\plotscale]{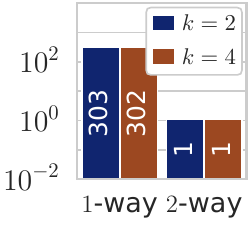}
		\caption{NG: Runtime}
		\label{figure:coreutils-classwise-nbh-runtime-greedy}
	\end{subfigure}
	\hfill%
	\begin{subfigure}{.16\linewidth}
		\centering%
		\includegraphics[scale=\plotscale]{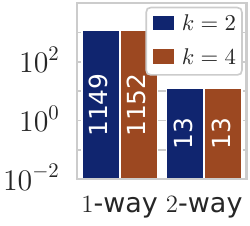}
		\caption{ILP: Runtime}
		\label{figure:coreutils-classwise-nbh-runtime-ilp}
	\end{subfigure}
	\hfill%
	\begin{subfigure}{.16\linewidth}
		\centering%
		\includegraphics[scale=\plotscale]{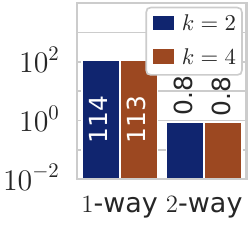}
		\caption{LPG: Runtime}
		\label{figure:coreutils-classwise-nbh-runtime-lpr}
	\end{subfigure}
	\hfill%
	\begin{subfigure}{.16\linewidth}
		\centering%
		\includegraphics[scale=\plotscale]{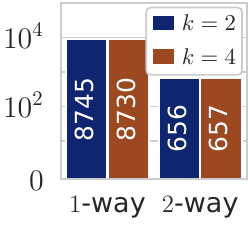}
		\caption{NG: Edge deletions}
		\label{figure:coreutils-classwise-nbh-deletions-greedy}
	\end{subfigure}
	\hfill%
	\begin{subfigure}{.16\linewidth}
		\centering%
		\includegraphics[scale=\plotscale]{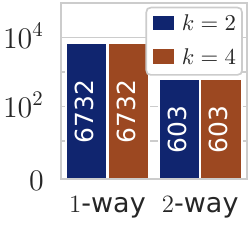}
		\caption{ILP: Edge deletions}
		\label{figure:coreutils-classwise-nbh-deletions-ilp}
	\end{subfigure}
	\hfill%
	\begin{subfigure}{.16\linewidth}
		\centering%
		\includegraphics[scale=\plotscale]{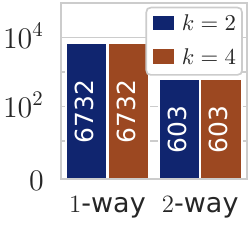}
		\caption{LPG: Edge deletions}
		\label{figure:coreutils-classwise-nbh-edge-deketions-lpr}
	\end{subfigure}
	\hfill%
	\caption{Runtime (in seconds) and edge deletions for neighbourhood errors for the Coreutils graph}\label{figure:coreutils-classwise-nbh}
\end{figure*}

\begin{figure*}
	\centering%
	\begin{subfigure}{.2\linewidth}
		\centering%
		\includegraphics[scale=\plotscale]{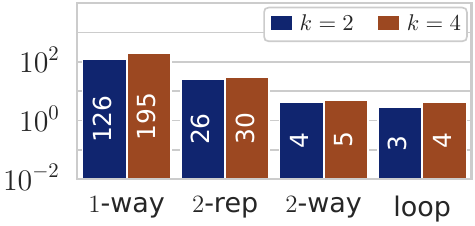}
		\caption{NG: Runtime}
		\label{figure:icij-classwise-nbh-runtime-greedy}
	\end{subfigure}
	\hfill%
	\begin{subfigure}{.2\linewidth}
		\centering%
		\includegraphics[scale=\plotscale]{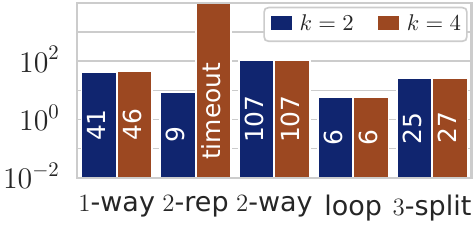}
		\caption{ILP: Runtime}
		\label{figure:icij-classwise-nbh-runtime-ilp}
	\end{subfigure}
	\hfill%
	\begin{subfigure}{.2\linewidth}
		\centering%
		\includegraphics[scale=\plotscale]{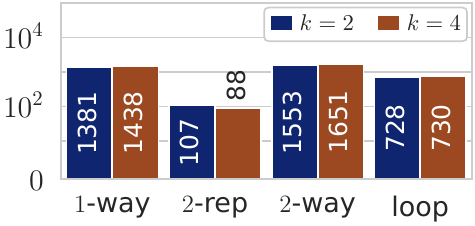}
		\caption{NG: Edge deletions}
		\label{figure:icij-classwise-nbh-deletions-greedy}
	\end{subfigure}
	\hfill%
	\begin{subfigure}{.2\linewidth}
		\centering%
		\includegraphics[scale=\plotscale]{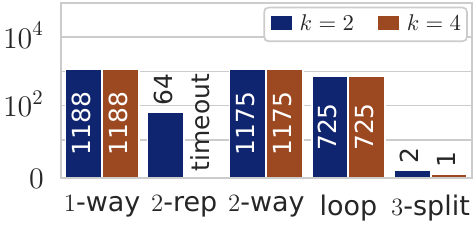}
		\caption{ILP: Edge deletions}
		\label{figure:icij-classwise-nbh-deletions-ilp}
	\end{subfigure}
	\hfill%
	\caption{Runtime (in seconds) and edge deletions for neighbourhood errors for the ICIJ graph}\label{figure:icij-classwise-nbh}
\end{figure*}

\begin{figure*}
	\centering%
	\begin{subfigure}{.16\linewidth}
		\centering%
		\includegraphics[scale=\plotscale]{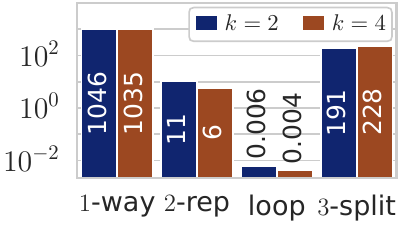}
		\caption{NG: Runtime}
		\label{figure:ldbc-sf0.3-classwise-nbh-runtime-greedy}
	\end{subfigure}
	\hfill%
	\begin{subfigure}{.16\linewidth}
		\centering%
		\includegraphics[scale=\plotscale]{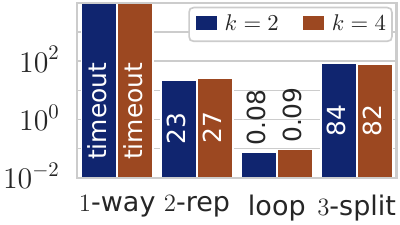}
		\caption{ILP: Runtime}
		\label{figure:ldbc-sf0.3-classwise-nbh-runtime-ilp}
	\end{subfigure}
	\hfill%
	\begin{subfigure}{.16\linewidth}
		\centering%
		\includegraphics[scale=\plotscale]{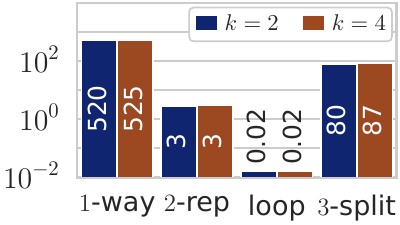}
		\caption{LPG: Runtime}
		\label{figure:ldbc-classwise-nbh-runtime-lpr}
	\end{subfigure}
	\hfill%
	\begin{subfigure}{.16\linewidth}
		\centering%
		\includegraphics[scale=\plotscale]{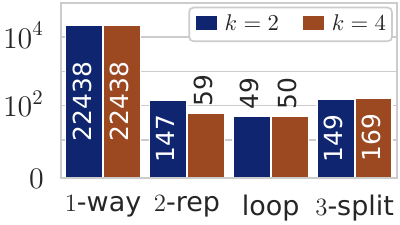}
		\caption{NG: Edge deletions}
		\label{figure:ldbc-sf0.3-classwise-nbh-deletions-greedy}
	\end{subfigure}
	\hfill%
	\begin{subfigure}{.16\linewidth}
		\centering%
		\includegraphics[scale=\plotscale]{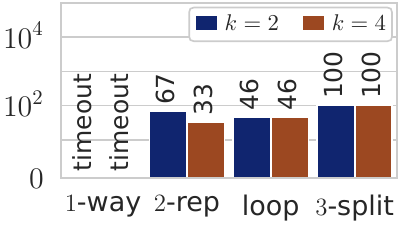}
		\caption{ILP: Edge deletions}
		\label{figure:ldbc-sf0.3-classwise-nbh-deletions-ilp}
	\end{subfigure}
	\hfill%
	\begin{subfigure}{.16\linewidth}
		\centering%
		\includegraphics[scale=\plotscale]{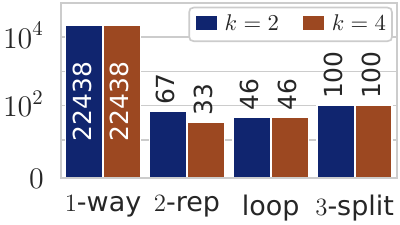}
		\caption{LPG: Edge deletions}
		\label{figure:ldbc-classwise-nbh-edge-deletions-lpr}
	\end{subfigure}
	\hfill%
	\caption{Runtime (in seconds) and edge deletions for neighbourhood errors for the LDBC graph}\label{figure:ldbc-sf0.3-classwise-nbh}
\end{figure*}

\begin{figure*}
	\centering%
	\begin{subfigure}{.164\linewidth}
		\centering%
		\includegraphics[scale=\plotscale]{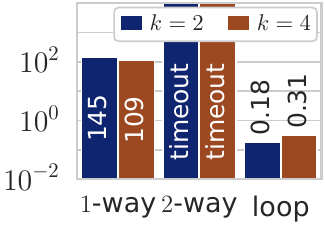}
		\caption{Legislative: Runtime}
		\label{figure:legislative-classwise-sample-runtime-lpr}
	\end{subfigure}
	\hfill%
	\begin{subfigure}{.17\linewidth}
		\centering%
		\includegraphics[scale=\plotscale]{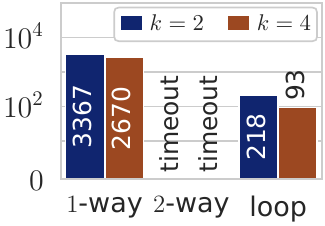}
		\caption{Legislative: Deletions}
		\label{figure:legislative-classwise-sample-edge-deletions-lpr}
	\end{subfigure}
	\hfill%
	\begin{subfigure}{.153\linewidth}
		\centering%
		\includegraphics[scale=\plotscale]{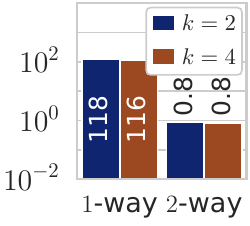}
		\caption{Coreutils: Runtime}
		\label{figure:coreutils-classwise-sample-runtime-lpr}
	\end{subfigure}
	\hfill%
	\begin{subfigure}{.159\linewidth}
		\centering%
		\includegraphics[scale=\plotscale]{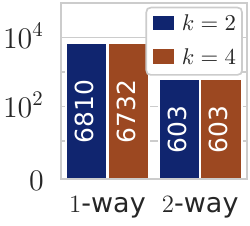}
		\caption{Coreutils: Deletions}
		\label{figure:coreutils-classwise-sample-edge-deketions-lpr}
	\end{subfigure}
	\hfill%
	\begin{subfigure}{.16\linewidth}
		\centering%
		\includegraphics[scale=\plotscale]{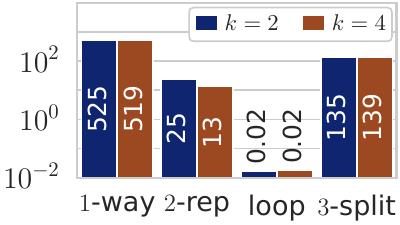}
		\caption{LDBC: Runtime}
		\label{figure:ldbc-classwise-sample-runtime-lpr}
	\end{subfigure}
	\hfill%
	\begin{subfigure}{.16\linewidth}
		\centering%
		\includegraphics[scale=\plotscale]{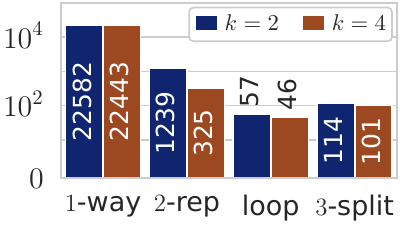}
		\caption{LDBC: Deletions}
		\label{figure:ldbc-classwise-sample-edge-deletions-lpr}
	\end{subfigure}
	\caption{Runtimes (in seconds) and edge deletions for sample errors of the LP-guided greedy algorithm}\label{figure:classwise-sample-runtimes}
\end{figure*}

\subsection{Raw Measurements and Trade-Off Percentages}\label{section:raw-numbers}
Table~\ref{table:runtimes-classwise} shows the runtimes recorded  for the datasets and sets of constraints listed in Table~\ref{table:static-pgs} for the base pipeline and with Step~2 (label deletions) enabled in the fourth and sixth column, respectively.
All numbers are averages of four runs.
Column \(5\) contains the runtime reduction, in percentage, achieved for the base pipeline by the naive greedy and the LP-guided greedy algorithms relative to the ILP algorithm.
To be precise, the runtime reduction is
\[100\cdot\frac{r_{\texttt{ilp}} - r_{\texttt{g}}}{r_{\texttt{ilp}}},\]
where \(r_{\texttt{ilp}}\) is the runtime of the ILP algorithm, and \(r_{\texttt{g}}\) is the runtime of the naive or LP-guided greedy algorithm.
For instance, the first three rows state that, for \(1\)-way constraints on the Coreutils graph, the runtime of the naive greedy algorithm is \(74.86\%\) less than the reference value \(1154.86s\) for the ILP algorithm, and the LP-guided greedy algorithm even achieves a runtime reduction of \(90.13\%\).

Similarly, the last column shows the runtime \emph{increase}
\[100\cdot\frac{r_\texttt{label}}{r_\texttt{base}},\]
when enabling Step~2.
Here \(r_\texttt{base}\) is the runtime of the naive greedy, LP-guided greedy, or ILP algorithm for the base pipeline, and \(r_\texttt{label}\) the runtime with Step~2 enabled.
For example, the first row states that the naive greedy algorithm takes \(290.32s\) for the base pipeline, and with Step~2 enabled, it takes \(320.46s\), which is \(110.38\%\) of the reference value \(290.32s\).
Extrema, and in particular numbers reported in the main paper, are printed in bold face in a (blue) frame.
Tables~\ref{table:runtimes-err-coreutils-legislative} and~\ref{table:runtimes-err-icij-ldbc} show the corresponding numbers for the scaling plots, and Table~\ref{table:approx-runtimes} for the selection phase of the greedy algorithms to compute approximate repairs, as discussed in Section~\ref{section:performance-quality} and Appendix~\ref{section:additional-plots}.
For the experiments with custom weights in Section~\ref{section:custom-weights} and Appendix~\ref{section:additional-plots}, they are shown in Table~\ref{table:custom-weight-runtimes}.

Table~\ref{table:deletions-classwise} show the recorded numbers of deletions for the base pipeline and with Step~2 enabled.
For example, the first three rows of Table~\ref{table:deletions-classwise} show that the naive greedy algorithm yields \(8747\) edge deletions, and the LP-guided greedy  and ILP algorithm both yield \(6732\) edge deletions for the base pipeline.
That is, the LP-guided greedy and ILP strategies yield \(23.04\%\) less deletions, which is reported in the fifth column, and hence a higher quality.
In general, the reduction in deletions is calculated as
\[100\cdot\frac{d_{\texttt{ng}} - d_{\texttt{lp}}}{d_{\texttt{ng}}},\]
where \(d_{\texttt{ng}}\) is the number of edge deletions yielded by the naive greedy algorithm, and \(d_{\texttt{lp}}\) the number of deletions yielded by the LP-guided greedy or ILP strategy.
The last column shows the reduction (or quality gain) in the number of deletions with Step~2 enabled, calculated by
\[
	100\cdot\frac{d_{\texttt{base}} - d_{\texttt{label}}}{d_{\texttt{base}}},
\]
where \(d_{\texttt{base}}\) is the number of edge deletions yielded by the base pipeline, and \(d_{\texttt{label}}\) is the number of label deletions yielded with Step~2 enabled.
For example, the first row states that, with Step~2 enabled, the naive greedy algorithm opts for deleting about \(5655\) labels (column \(6\)), which is \(35.35\%\) less than the reference value of \(8747\) edge deletions in column \(4\).
As for the runtime table, all values are averaged over four runs.
Since the greedy algorithms can pick a different number of deletions in each run, the average is not necessarily an integer value.
Tables~\ref{table:deletions-err-coreutils-legislative} and~\ref{table:deletions-err-icij-ldbc} shows the number of deletions for the scaling experiments and Table~\ref{table:approx-deletions} for the approximate repairs discussed in Section~\ref{section:performance-quality} and Appendix~\ref{section:additional-plots}.
For the experiments with custom weights in Section~\ref{section:custom-weights} and Appendix~\ref{section:additional-plots}, they are shown in Table~\ref{table:custom-weight-runtimes}.

Table~\ref{table:neighbourhood-classwise} shows the runtimes and runtime reductions over the base pipeline for Step~3 (neighbourhood error repairs); analogously to Table~\ref{table:runtimes-classwise} explained above.
The runtime reductions are computed with the formula
\[
	100\cdot\frac{r_{\texttt{base}} - r_{\texttt{nbh}}}{r_{\texttt{base}}},
\]
where \(r_{\texttt{nbh}}\) and \(r_{\texttt{base}}\) are the runtimes with and without Step~3 enabled, respectively.
The runtime (reductions) for the experiments in Section~\ref{section:ablation} with sampled errors are shown in Table~\ref{table:sample-classwise}.

Finally, Table~\ref{table:neighbourhood-classwise-deletions} shows the deletions with and without Step~3 enabled.
The percentage is computed as
\[
100\cdot\frac{d_{\texttt{nbh}}}{d_{\texttt{base}}},
\]
where \(d_{\texttt{nbh}}\) and \(d_{\texttt{base}}\) are the numbers of deletions with Step~3 enabled and without Step~3, respectively.
The numbers of edge deletions for the experiments in Section~\ref{section:ablation} with sampled errors are shown in Table~\ref{table:sample-classwise-deletions}.

\begin{table*}
	\caption{Raw runtimes and runtime comparisons for the datasets and sets of constraints listed in Table~\ref{table:static-pgs}}\label{table:runtimes-classwise}
\begin{tabular}{c c l l l l l l}
	\toprule
	&&&&\multicolumn{2}{c}{\bfseries Base Pipeline}&\multicolumn{2}{c}{\bfseries Repair with label deletions (Step~2)}\\
	Dataset & Shape &\#Errors&Algorithm&Runtime&Runtime Reduction&Runtime&Runtime Increase\\
	\toprule
	\multirow{6}{*}{Coreutils}
	&$1$-way&45728&Naive Greedy&290.32&74.86&320.46&110.38\\
	&$1$-way&45728&LP-Guided Greedy&113.95&90.13&109.13&95.78\\
	&$1$-way&45728&ILP&1154.86&&&\\
	&$2$-way&3707&Naive Greedy&1.11&91.19&2.11&190.02\\
	&$2$-way&3707&LP-Guided Greedy&0.85&93.24&1.56&182.73\\
	&$2$-way&3707&ILP&12.59&&49.92&396.47\\
	\midrule
	\multirow{9}{*}{Legislative}
	&$1$-way&166977&Naive Greedy&356.64&70.10&389.81&109.30\\
	&$1$-way&166977&LP-Guided Greedy&111.94&90.62&97.90&87.46\\
	&$1$-way&166977&ILP&1192.84&&&\\
	&$2$-way&419439&Naive Greedy&&&&\\
	&$2$-way&443695&LP-Guided Greedy&&&&\\
	&$2$-way&419439&ILP&&&&\\
	&loop&3783&Naive Greedy&0.18&86.57&0.22&126.09\\
	&loop&3783&LP-Guided Greedy&0.28&78.61&0.56&200.18\\
	&loop&3783&ILP&1.31&&4.37&334.05\\
	\midrule
	\multirow{15}{*}{ICIJ}
	&$1$-way&558799&Naive Greedy&324.12&\highlightcell{-445.65}&559.79&172.71\\
	&$1$-way&558799&LP-Guided Greedy&183.48&-208.89&393.31&\highlightcell{214.36}\\
	&$1$-way&558799&ILP&59.40&&294.69&\highlightcell{496.11}\\
	&$2$-rep&397108&Naive Greedy&35.68&&53.62&150.28\\
	&$2$-rep&397108&LP-Guided Greedy&45.27&&&\\
	&$2$-rep&397108&ILP&&&&\\
	&$2$-way&7553&Naive Greedy&4.89&\highlightcell{95.51}&6.52&133.22\\
	&$2$-way&7553&LP-Guided Greedy&2.34&\highlightcell{97.85}&3.75&159.83\\
	&$2$-way&7553&ILP&108.90&&429.15&394.07\\
	&loop&56614&Naive Greedy&6.14&10.13&14.69&239.29\\
	&loop&56614&LP-Guided Greedy&7.51&-9.97&16.04&213.45\\
	&loop&56614&ILP&6.83&&23.21&339.75\\
	&$3$-split&876933&Naive Greedy&3.03&89.02&4.11&135.57\\
	&$3$-split&876933&LP-Guided Greedy&25.98&5.91&54.69&210.54\\
	&$3$-split&876933&ILP&27.61&&56.13&203.29\\

	\midrule
	\multirow{12}{*}{LDBC}
	&$1$-way&79064&Naive Greedy&980.82&&1627.58&165.94\\
	&$1$-way&79064&LP-Guided Greedy&527.74&&846.50&160.40\\
	&$1$-way&79064&ILP&&&&\\
	&$2$-rep&72756&Naive Greedy&3.90&86.29&9.99&\highlightcell{256.01}\\
	&$2$-rep&72756&LP-Guided Greedy&3.36&88.18&6.35&188.79\\
	&$2$-rep&72756&ILP&28.46&&109.26&383.86\\
	&loop&423&Naive Greedy&0.0041&\highlightcell{95.97}&0.01&182.72\\
	&loop&423&LP-Guided Greedy&0.01&86.28&0.02&162.33\\
	&loop&423&ILP&0.10&&0.35&347.20\\
	&$3$-split&1358499&Naive Greedy&225.46&-165.99&477.74&211.90\\
	&$3$-split&1358499&LP-Guided Greedy&87.64&-3.40&173.72&198.22\\
	&$3$-split&1358499&ILP&84.76&&143.25&169.00\\
	\bottomrule
\end{tabular} \end{table*}

\begin{table*}
	\caption{Raw number of deletions and comparisons for the datasets and sets of constraints listed in Table~\ref{table:static-pgs}}\label{table:deletions-classwise}
\begin{tabular}{c c l l l l l l}
	\toprule
	&&&&\multicolumn{2}{c}{\bfseries Base Pipeline}&\multicolumn{2}{c}{\bfseries Repair with label deletions (Step~2)}\\
	Dataset & Shape &\#Errors&Algorithm&\#Edge Deletions&Reduction in Deletions&\#Label Deletions&Delete Reduction\\
	\toprule
	\multirow{6}{*}{\rotatebox{90}{Coreutils}}
	&$1$-way&45728&Naive Greedy&8747.00&&5655.29&35.35\\
	&$1$-way&45728&LP-Guided Greedy&6732.00&23.04&4462.00&33.72\\
	&$1$-way&45728&ILP&6732.00&23.04&&\\
	&$2$-way&3707&Naive Greedy&656.00&&626.33&4.52\\
	&$2$-way&3707&LP-Guided Greedy&603.00&8.08&593.00&1.66\\
	&$2$-way&3707&ILP&603.00&8.08&593.00&1.66\\
	\midrule
	\multirow{9}{*}{\rotatebox{90}{Legislative}}
	&$1$-way&166977&Naive Greedy&3818.67&&1905.33&\highlightcell{50.10}\\
	&$1$-way&166977&LP-Guided Greedy&2618.00&31.44&1064.33&\highlightcell{59.35}\\
	&$1$-way&166977&ILP&2618.00&31.44&&\\
	&$2$-way&419439&Naive Greedy&&&&\\
	&$2$-way&443695&LP-Guided Greedy&&&&\\
	&$2$-way&419439&ILP&&&&\\
	&loop&3783&Naive Greedy&49.67&&32.33&34.90\\
	&loop&3783&LP-Guided Greedy&38.00&23.49&20.00&47.37\\
	&loop&3783&ILP&38.00&23.49&20.00&\highlightcell{47.37}\\
	\midrule
	\multirow{15}{*}{ICIJ}
	&$1$-way&558799&Naive Greedy&1411.00&&1008.33&28.54\\
	&$1$-way&558799&LP-Guided Greedy&1188.33&15.78&826.33&30.46\\
	&$1$-way&558799&ILP&1188.00&15.80&807.00&32.07\\
	&$2$-rep&397108&Naive Greedy&79.67&&57.00&28.45\\
	&$2$-rep&397108&LP-Guided Greedy&50.67&\highlightcell{36.40}&&\\
	&$2$-rep&397108&ILP&&&&\\
	&$2$-way&7553&Naive Greedy&1588.33&&1261.00&20.61\\
	&$2$-way&7553&LP-Guided Greedy&1176.33&25.94&1072.00&8.87\\
	&$2$-way&7553&ILP&1175.00&26.023&1072.00&8.77\\
	&loop&56614&Naive Greedy&727.33&&720.67&0.92\\
	&loop&56614&LP-Guided Greedy&725.00&0.32&715.33&1.33\\
	&loop&56614&ILP&725.00&0.32&715.00&1.38\\
	&$3$-split&876933&Naive Greedy&2.00&&1.00&50.00\\
	&$3$-split&876933&LP-Guided Greedy&1.00&50.0&1.00&0.00\\
	&$3$-split&876933&ILP&1.00&\highlightcell{50.0}&1.00&0.00\\
	\midrule
	\multirow{12}{*}{LDBC}
	&$1$-way&79064&Naive Greedy&22438.00&&21099.00&5.97\\
	&$1$-way&79064&LP-Guided Greedy&22438.00&&19740.00&12.02\\
	&$1$-way&79064&ILP&&&&\\
	&$2$-rep&72756&Naive Greedy&42.33&&51.33&-21.26\\
	&$2$-rep&72756&LP-Guided Greedy&31.00&26.77&30.00&3.23\\
	&$2$-rep&72756&ILP&31.00&26.77&30.00&3.23\\
	&loop&423&Naive Greedy&48.67&&48.00&1.37\\
	&loop&423&LP-Guided Greedy&46.00&5.48&44.00&4.35\\
	&loop&423&ILP&46.00&5.48&44.00&4.35\\
	&$3$-split&1358499&NG&155.00&&142.33&8.17\\
	&$3$-split&1358499&LPG&100.00&\highlightcell{35.48}&93.00&7.00\\
	&$3$-split&1358499&ILP&100.00&\highlightcell{35.48}&93.00&7.00\\
	\bottomrule
\end{tabular} \end{table*}

\begin{table*}
	\caption{Raw runtimes and runtime comparisons for Step~3 (neighbourhood errors)}\label{table:neighbourhood-classwise}
\begin{tabular}{c c l l l l l l l}
	\toprule
	&&&&{\bfseries Base Pipeline}&\multicolumn{2}{c}{\(\mathbf{k = 2}\)}&\multicolumn{2}{c}{\(\mathbf{k = 4}\)}\\
	Dataset & Shape &\#Errors&Algorithm&Runtime&Runtime&Runtime Reduction&Runtime&Runtime Reduction\\
	\toprule
	\multirow{6}{*}{\rotatebox{90}{Coreutils}}
	&$1$-way&45728&Naive Greedy&290.32&302.98&-4.36&301.72&-3.93\\
	&$1$-way&45728&LP-Guided Greedy&113.95&114.42&-0.42&113.00&0.83\\
	&$1$-way&45728&ILP&1154.86&1148.77&0.53&1152.01&0.25\\
	&$2$-way&3707&Naive Greedy&1.11&1.13&-1.77&1.12&-0.61\\
	&$2$-way&3707&LP-Guided Greedy&0.85&0.87&-1.6445&0.84&1.01\\
	&$2$-way&3707&ILP&12.59&12.54&0.41&12.76&-1.32\\
	\midrule
	\multirow{9}{*}{\rotatebox{90}{Legislative}}
	&$1$-way&166977&Naive Greedy&356.64&304.48&14.63&303.79&14.82\\
	&$1$-way&166977&LP-Guided Greedy&111.94&95.14&15.00&102.70&8.25\\
	&$1$-way&166977&ILP&1192.84&1213.44&-1.73&1190.86&0.17\\
	&$2$-way&443695&Naive Greedy&&&&&\\
	&$2$-way&443695&LP-Guided Greedy&&&&&\\
	&$2$-way&443695&ILP&&&&&\\
	&loop&3783&Naive Greedy&0.18&0.04&75.99&0.10&43.61\\
	&loop&3783&LP-Guided Greedy&0.28&0.08&71.90&0.12&57.24\\
	&loop&3783&ILP&1.31&0.29&\highlightcell{78.04}&0.87&33.86\\
	\midrule
	\multirow{15}{*}{ICIJ}
	&$1$-way&558799&Naive Greedy&324.12&126.14&61.08&194.55&39.98\\
	&$1$-way&558799&LP-Guided Greedy&183.48&68.21&\highlightcell{62.83}&103.09&43.81\\
	&$1$-way&558799&ILP&59.40&41.21&30.62&46.21&22.21\\
	&$2$-rep&397108&Naive Greedy&35.68&26.33&26.20&30.24&15.26\\
	&$2$-rep&397108&LP-Guided Greedy&45.27&9.19&79.69&44.98&0.63\\
	&$2$-rep&397108&ILP&&8.78&&&\\
	&$2$-way&7553&Naive Greedy&4.89&4.33&11.53&5.07&-3.60\\
	&$2$-way&7553&LP-Guided Greedy&2.34&2.27&3.35&2.47&-5.12\\
	&$2$-way&7553&ILP&108.90&106.91&1.83&107.20&1.57\\
	&loop&56614&Naive Greedy&6.14&2.90&52.84&4.23&31.08\\
	&loop&56614&LP-Guided Greedy&7.51&3.50&53.42&5.01&33.36\\
	&loop&56614&ILP&6.83&5.66&17.21&5.87&14.03\\
	&$3$-split&876933&Naive Greedy&3.03&4.46&-47.12&2.01&33.65\\
	&$3$-split&876933&LP-Guided Greedy&25.98&23.84&8.24&25.79&0.74\\
	&$3$-split&876933&ILP&27.61&25.35&8.19&26.83&2.84\\
	\midrule
	\multirow{12}{*}{LDBC}
	&$1$-way&79064&Naive Greedy&980.82&1045.71&-6.62&1034.62&-5.48\\
	&$1$-way&79064&LP-Guided Greedy&527.74&519.94&1.48&524.52&0.61\\
	&$1$-way&79064&ILP&&&&&\\
	&$2$-rep&72756&Naive Greedy&3.90&10.76&-175.78&5.58&-42.98\\
	&$2$-rep&72756&ILP&28.46&22.60&20.61&26.62&6.47\\
	&loop&423&Naive Greedy&0.0041&0.006&-46.73&0.0041&1.02\\
	&loop&423&LP-Guided Greedy&0.01&0.06&-13.40&0.02&-14.46\\
	&loop&423&ILP&0.10&0.08&24.36&0.09&8.13\\
	&$3$-split&1358499&Naive Greedy&225.46&191.02&15.2738&227.6318&-0.9631\\
	&$3$-split&1358499&LP-Guided Greedy&87.64&79.87&8.86&87.07&0.65\\
	&$3$-split&1358499&ILP&84.76&83.97&0.93&82.45&2.73\\
	\bottomrule
\end{tabular} \end{table*}

\begin{table*}
	\caption{Raw edge deletions and deletion comparisons for Step~3 (neighbourhood errors)}\label{table:neighbourhood-classwise-deletions}
\begin{tabular}{c c l l l l l l l}
	\toprule
	&&&&{\bfseries Base Pipeline}&\multicolumn{2}{c}{\(\mathbf{k = 2}\)}&\multicolumn{2}{c}{\(\mathbf{k = 4}\)}\\
	Dataset & Shape &\#Errors&Algorithm&Deletions&Deletions&Percentage&Deletions&Percentage\\
	\toprule
	\multirow{6}{*}{Coreutils}
	&$1$-way&45728&Naive Greedy&8747.00&8744.67&99.97&8730.00&99.81\\
	&$1$-way&45728&LP-Guided Greedy&6732.00&6732.00&100.00&6732.00&100.00\\
	&$1$-way&45728&ILP&6732.00&6732.00&100.00&6732.00&100.00\\
	&$2$-way&3707&Naive Greedy&656.00&656.00&100.00&656.67&100.10\\
	&$2$-way&3707&LP-Guided Greedy&603.00&603.00&100.00&603.00&100.00\\
	&$2$-way&3707&ILP&603.00&603.00&100.00&603.00&100.00\\
	\midrule
	\multirow{9}{*}{Legislative}
	&$1$-way&166977&Naive Greedy&3818.67&3909.33&102.37&3652.00&95.64\\
	&$1$-way&166977&LP-Guided Greedy&2618.00&2631.00&100.50&2618.00&100.00\\
	&$1$-way&166977&ILP&2618.00&2631.00&100.50&2618.00&100.00\\
	&$2$-way&443695&Naive Greedy&&&&&\\
	&$2$-way&443695&LP-Guided Greedy&&&&&\\
	&$2$-way&443695&ILP&&&&&\\
	&loop&3783&Naive Greedy&49.67&48.33&97.32&56.00&112.75\\
	&loop&3783&LP-Guided Greedy&38.00&39.00&102.63&38.00&100.00\\
	&loop&3783&ILP&38.00&39.00&102.63&38.00&100.00\\
	\midrule
	\multirow{15}{*}{ICIJ}
	&$1$-way&558799&Naive Greedy&1411.00&1381.33&97.90&1438.33&101.94\\
	&$1$-way&558799&LP-Guided Greedy&1188.33&1188.33&100.00&1188.33&100.00\\
	&$1$-way&558799&ILP&1188.00&1188.00&100.00&1188.00&100.00\\
	&$2$-rep&397108&Naive Greedy&79.67&107.33&134.73&88.33&110.88\\
	&$2$-rep&397108&LP-Guided Greedy&50.67&64.00&126.32&56.00&110.53\\
	&$2$-rep&397108&ILP&&64.00&&&\\
	&$2$-way&7553&Naive Greedy&1588.33&1552.67&97.75&1650.67&103.92\\
	&$2$-way&7553&LP-Guided Greedy&1176.33&1176.33&100.00&1175.33&99.91\\
	&$2$-way&7553&ILP&1175.00&1175.00&100.00&1175.00&100.00\\
	&loop&56614&Naive Greedy&727.33&727.67&100.05&730.33&100.41\\
	&loop&56614&LP-Guided Greedy&725.00&725.00&100.00&725.00&100.00\\
	&loop&56614&ILP&725.00&725.00&100.00&725.00&100.00\\
	&$3$-split&876933&Naive Greedy&2.00&5.00&250.00&3.33&166.67\\
	&$3$-split&876933&LP-Guided Greedy&1.00&2.33&233.33&1.00&100.00\\
	&$3$-split&876933&ILP&1.00&2.00&200.00&1.00&100.00\\
	\midrule
	\multirow{12}{*}{LDBC}
	&$1$-way&79064&Naive Greedy&22438.00&22438.00&100.00&22438.00&100.00\\
	&$1$-way&79064&LP-Guided Greedy&22438.00&22438.00&100.00&22438.00&100.00\\
	&$1$-way&79064&ILP&&&&&\\
	&$2$-rep&72756&Naive Greedy&42.33&147.33&348.03&59.33&140.16\\
	&$2$-rep&72756&LP-Guided Greedy&31.00&67.00&216.13&33.00&106.45\\
	&$2$-rep&72756&ILP&31.00&67.00&216.13&33.00&106.45\\
	&loop&423&Naive Greedy&48.67&49.00&100.68&50.33&103.42\\
	&loop&423&LP-Guided Greedy&46.00&46.00&100.00&46.00&100.00\\
	&loop&423&ILP&46.00&46.00&100.00&46.00&100.00\\
	&$3$-split&1358499&Naive Greedy&155.00&149.33&96.34&169.33&109.25\\
	&$3$-split&1358499&LP-Guided Greedy&100.00&100.00&100.00&100.00&100.00\\
	&$3$-split&1358499&ILP&100.00&100.00&100.00&100.00&100.00\\
	\bottomrule
\end{tabular} \end{table*}

\begin{table*}
	\caption{Raw runtimes and runtime comparisons for repairing with sampled errors}\label{table:sample-classwise}
\begin{tabular}{c c l l l l l l l}
	\toprule
	&&&&{\bfseries Base Pipeline}&\multicolumn{2}{c}{\(\mathbf{k = 2}\)}&\multicolumn{2}{c}{\(\mathbf{k = 4}\)}\\
	Dataset & Shape &\#Errors&Algorithm&Runtime&Runtime&Runtime Reduction&Runtime&Runtime Reduction\\
	\toprule
	\multirow{2}{*}{Coreutils}
	&$1$-way&45728&LP-Guided Greedy&113.9463&117.9029&-3.4724&116.1778&-1.9584\\
	&$2$-way&3707&LP-Guided Greedy&0.8512&0.8423&1.0449&0.8353&1.8685\\
	\midrule
	\multirow{3}{*}{Legislative}
	&$1$-way&166977&LP-Guided Greedy&111.9393&144.8334&-29.3857&109.1641&2.4792\\
	&$2$-way&419439&LP-Guided Greedy&&&&&\\
	&loop&3783&LP-Guided Greedy&0.2798&0.1843&34.1506&0.3096&-10.6573\\
	\midrule
	\multirow{5}{*}{ICIJ}
	&$1$-way&558799&LP-Guided Greedy&183.4825&76.8191&58.1328&102.3670&44.2089\\
	&$2$-rep&397108&LP-Guided Greedy&45.2682&34.3127&24.2014&83.7004&-84.8988\\
	&$2$-way&7553&LP-Guided Greedy&2.3449&2.1603&7.8712&2.3681&-0.9890\\
	&loop&56614&LP-Guided Greedy&7.5130&4.7326&37.0083&5.4140&27.9375\\
	&$3$-split&876933&LP-Guided Greedy&25.9784&21.9824&15.3820&26.6876&-2.7303\\
	\midrule
	\multirow{3}{*}{LDBC}
	&$1$-way&79064&LP-Guided Greedy&527.7366&524.6183&0.5909&518.9552&1.6640\\
	&$2$-rep&72756&LP-Guided Greedy&3.3646&24.9187&-640.6203&13.4611&-300.0839\\
	&loop&423&LP-Guided Greedy&0.0140&0.0167&-19.1819&0.0177&-25.9705\\
	\bottomrule
\end{tabular} \end{table*}

\begin{table*}
	\caption{Raw edge deletions and deletion comparisons for repairing with sampled errors}\label{table:sample-classwise-deletions}
\begin{tabular}{c c l l l l l l l}
	\toprule
	&&&&{\bfseries Base Pipeline}&\multicolumn{2}{c}{\(\mathbf{k = 2}\)}&\multicolumn{2}{c}{\(\mathbf{k = 4}\)}\\
	Dataset & Shape &\#Errors&Algorithm&Deletions&Deletions&Percentage&Deletions&Percentage\\
	\toprule
	\multirow{2}{*}{Coreutils}
	&$1$-way&45728&LP-Guided Greedy&6732.00&6810.33&101.16&6732.00&100.00\\
	&$2$-way&3707&LP-Guided Greedy&603.00&603.00&100.00&603.00&100.00\\
	\midrule
	\multirow{3}{*}{Legislative}
	&$1$-way&166977&LP-Guided Greedy&2618.00&3367.33&128.62&2670.00&101.99\\
	&$2$-way&419439&LP-Guided Greedy&&&&&\\
	&loop&3783&LP-Guided Greedy&38.00&218.00&573.68&92.67&243.86\\
	\midrule
	\multirow{5}{*}{ICIJ}
	&$1$-way&558799&LP-Guided Greedy&1188.33&1245.33&104.80&1204.33&101.35\\
	&$2$-rep&397108&LP-Guided Greedy&50.67&630.67&1244.74&207.67&409.87\\
	&$2$-way&7553&LP-Guided Greedy&1176.33&1180.67&100.37&1175.67&99.94\\
	&loop&56614&LP-Guided Greedy&725.00&762.67&105.20&740.33&102.11\\
	&$3$-split&876933&LP-Guided Greedy&1.00&3.00&300.00&3.00&300.00\\
	\midrule
	\multirow{3}{*}{LDBC}
	&$1$-way&79064&LP-Guided Greedy&22438.00&22582.00&100.64&22443.00&100.02\\
	&loop&423&LP-Guided Greedy&46.00&57.00&123.91&46.00&100.00\\
	&$2$-rep&72756&LP-Guided Greedy&31.00&1239.33&3997.85&325.00&1048.39\\
	\bottomrule
\end{tabular} \end{table*}

\begin{table*}
	\caption{Raw runtimes and runtime comparisons for the scaling experiments for the Coreutils and legislative graphs (\(1\)-way constraints)}\label{table:runtimes-err-coreutils-legislative}
\begin{tabular}{c l l l l l l}
\toprule
\multirow{2}{*}{\parbox{7em}{Dataset and\newline Constraint Shape}}&&&\multicolumn{2}{c}{\bfseries Base Pipeline}&\multicolumn{2}{c}{\bfseries Repair with label deletions (Step~2)}\\
&\#Errors&Algorithm&Runtime&Runtime Reduction&Runtime&Runtime Increase\\
\toprule
\multirow{18}{*}{\parbox{2cm}{Coreutils\newline \(1\)-way}}&14809&Naive Greedy&26.78&90.68&30.98&115.69\\
&14809&LP-Guided Greedy&10.78&96.24&12.37&114.76\\
&14809&ILP&287.15&&1245.41&433.72\\
&19109&Naive Greedy&40.08&88.18&38.56&96.20\\
&19109&LP-Guided Greedy&12.19&96.41&10.48&86.01\\
&19109&ILP&339.057&&1401.39&413.32\\
&23393&Naive Greedy&66.43&86.65&64.43&97.00\\
&23393&LP-Guided Greedy&35.27&92.91&22.94&65.05\\
&23393&ILP&497.72&&2078.88&417.68\\
&27661&Naive Greedy&93.34&83.36&81.91&87.76\\
&27661&LP-Guided Greedy&43.98&92.16&22.92&52.12\\
&27661&ILP&561.02&&2352.06&419.25\\
&35101&Naive Greedy&163.55&80.04&163.26&99.82\\
&35101&LP-Guided Greedy&56.32&93.13&52.72&93.61\\
&35101&ILP&819.24&&3224.69&393.62\\
&45728&Naive Greedy&290.32&74.86&322.12&110.95\\
&45728&LP-Guided Greedy&110.73&90.41&105.31&95.11\\
&45728&ILP&1154.86&&timeout&\\
\midrule\multirow{18}{*}{\parbox{2cm}{Legislative\newline
\(1\)-way}}&26603&Naive Greedy&11.52&84.53&17.28&150.00\\
&26603&LP-Guided Greedy&4.31&94.21&4.75&110.18\\
&26603&ILP&74.48&&309.06&414.95\\
&41047&Naive Greedy&26.81&79.30&35.93&134.03\\
&41047&LP-Guided Greedy&9.81&92.43&8.29&84.55\\
&41047&ILP&129.51&&521.96&403.02\\
&91841&Naive Greedy&57.60&69.23&85.64&148.67\\
&91841&LP-Guided Greedy&18.84&89.94&20.37&108.12\\
&91841&ILP&187.20&&766.74&409.57\\
&98686&Naive Greedy&110.54&70.67&158.43&143.33\\
&98686&LP-Guided Greedy&43.41&88.48&48.60&111.97\\
&98686&ILP&376.91&&1692.88&449.15\\
&121968&Naive Greedy&187.30&73.68&258.30&137.91\\
&121968&LP-Guided Greedy&61.78&91.32&65.68&106.31\\
&121968&ILP&711.57&&3079.11&432.72\\
&166977&Naive Greedy&356.64&70.10&389.81&109.30\\
&166977&LP-Guided Greedy&109.13&90.85&98.02&89.82\\
&166977&ILP&1192.84&&timeout&\\
\bottomrule
\end{tabular} \end{table*}

\begin{table*}
	\caption{Raw runtimes and runtime comparisons for the scaling experiments for the ICIJ and LDBC graphs (\(1\)-way constraints)}\label{table:runtimes-err-icij-ldbc}
\begin{tabular}{c l l l l l l}
\toprule
\multirow{2}{*}{\parbox{7em}{Dataset and\newline Constraint Shape}}&&&\multicolumn{2}{c}{\bfseries Base Pipeline}&\multicolumn{2}{c}{\bfseries Repair with label deletions (Step~2)}\\
&\#Errors&Algorithm&Runtime&Runtime Reduction&Runtime&Runtime Increase\\
\toprule
\multirow{18}{*}{\parbox{2cm}{ICIJ\newline
\(1\)-way}}&1285&Naive Greedy&0.027&87.97&0.03&100.03\\
&1285&LP-Guided Greedy&0.05&78.91&0.09&190.18\\
&1285&ILP&0.23&&0.54&235.78\\
&1517&Naive Greedy&0.17&80.79&0.26&148.87\\
&1517&LP-Guided Greedy&0.19&79.20&0.25&135.68\\
&1517&ILP&0.89&&3.14&351.59\\
&2786&Naive Greedy&0.42&94.35&0.58&138.17\\
&2786&LP-Guided Greedy&0.31&95.89&0.44&143.73\\
&2786&ILP&7.49&&28.53&380.64\\
&175724&Naive Greedy&36.19&-22.40&66.26&183.08\\
&175724&LP-Guided Greedy&24.13&18.41&52.79&218.80\\
&175724&ILP&29.57&&125.52&424.50\\
&558799&Naive Greedy&324.12&-445.65&559.79&172.71\\
&558799&LP-Guided Greedy&176.13&-196.51&397.95&225.95\\
&558799&ILP&59.40&&294.69&\highlightcell{496.11}\\
&2052352&Naive Greedy&timeout&&timeout&\\
&2052352&LP-Guided Greedy&timeout&&timeout&\\
&2052352&ILP&timeout&&timeout&\\
\midrule\multirow{12}{*}{\parbox{2cm}{LDBC\newline
\(1\)-way}}&8704&Naive Greedy&14.2085507934292&\highlightcell{97.29}&25.46&179.22\\
&8704&LP-Guided Greedy&8.61&98.36&12.77&148.31\\
&8704&ILP&524.46&&2402.83&458.15\\
&15337&Naive Greedy&54.99&97.12&99.98&181.80\\
&15337&LP-Guided Greedy&42.52&97.77&60.18&141.54\\
&15337&ILP&1909.62&&timeout&\\
&55957&Naive Greedy&143.65&&244.35&170.10\\
&55957&LP-Guided Greedy&69.12&&92.72&134.14\\
&55957&ILP&timeout&&timeout&\\
&79064&Naive Greedy&980.82&&1627.58&165.94\\
&79064&LP-Guided Greedy&513.14&&847.47&165.15\\
&79064&ILP&timeout&&timeout&\\
\bottomrule
\end{tabular} \end{table*}

\begin{table*}
	\caption{Raw number of deletions and comparisons for the scaling experiments of the Coreutils and legislative graphs (\(1\)-way constraints)}\label{table:deletions-err-coreutils-legislative}
\begin{tabular}{c l l l l l l}
	\toprule
	\multirow{2}{*}{\parbox{7em}{Dataset and\newline Constraint Shape}}&&&\multicolumn{2}{c}{\bfseries Base Pipeline}&\multicolumn{2}{c}{\bfseries Repair with label deletions (Step~2)}\\
&\#Errors&Algorithm&\#Edge Deletions&Reduction in Deletions&\#Label Deletions&Delete Reduction\\
\toprule
\multirow{18}{*}{\parbox{2cm}{Coreutils\newline \(1\)-way}}&14809&Naive Greedy&3679&&2579&29.90\\
&14809&LP-Guided Greedy&2763.67&24.87&2092.00&24.30\\
&14809&ILP&2754&25.14&2092&24.04\\
&19109&Naive Greedy&4115&&2655&35.49\\
&19109&LP-Guided Greedy&2848.33&30.79&2092.00&26.55\\
&19109&ILP&2824&31.38&2092&25.92\\
&23393&Naive Greedy&5207&&3299&36.65\\
&23393&LP-Guided Greedy&4107.33&21.12&2586.67&37.02\\
&23393&ILP&4069&21.86&2586&36.45\\
&27661&Naive Greedy&5753&&3344&41.87\\
&27661&LP-Guided Greedy&4346.00&24.46&2586.67&40.48\\
&27661&ILP&4257&26.00&2586&39.25\\
&35101&Naive Greedy&6870&&4344&36.77\\
&35101&LP-Guided Greedy&5095.33&25.84&3375.00&33.76\\
&35101&ILP&5088&25.94&3375&33.67\\
&45728&Naive Greedy&8747&&5657&35.33\\
&45728&LP-Guided Greedy&6732.00&23.04&4462.00&33.72\\
&45728&ILP&6732&23.04&timeout&\\
\midrule\multirow{18}{*}{\parbox{2cm}{Legislative\newline
		\(1\)-way}}&26603&Naive Greedy&1003&&630&37.21\\
&26603&LP-Guided Greedy&762.00&24.05&370.00&51.44\\
&26603&ILP&762&24.05&370&51.44\\
&41047&Naive Greedy&1551&&760&\highlightcell{50.98}\\
&41047&LP-Guided Greedy&1082.00&30.24&441.00&59.24\\
&41047&ILP&1082&30.24&441&\highlightcell{59.24}\\
&91841&Naive Greedy&1830&&1012&44.70\\
&91841&LP-Guided Greedy&1365.00&25.41&596.00&56.34\\
&91841&ILP&1365&25.41&596&56.34\\
&98686&Naive Greedy&2421&&1390&42.59\\
&98686&LP-Guided Greedy&1797.00&25.76&774.00&56.93\\
&98686&ILP&1797&25.76&774&56.93\\
&121968&Naive Greedy&3185&&1695&46.78\\
&121968&LP-Guided Greedy&2100.00&34.06&899.00&57.19\\
&121968&ILP&2100&34.06&899&57.19\\
&166977&Naive Greedy&3819&&1905&50.10\\
&166977&LP-Guided Greedy&2618.00&31.44&1064.00&59.36\\
&166977&ILP&2618&31.44&timeout&\\
\bottomrule
\end{tabular} \end{table*}

\begin{table*}
	\caption{Raw number of deletions and comparisons for the scaling experiments of the ICIJ and LDBC graphs (\(1\)-way constraints)}\label{table:deletions-err-icij-ldbc}
\begin{tabular}{c l l l l l l}
	\toprule
	\multirow{2}{*}{\parbox{7em}{Dataset and\newline Constraint Shape}}&&&\multicolumn{2}{c}{\bfseries Base Pipeline}&\multicolumn{2}{c}{\bfseries Repair with label deletions (Step~2)}\\
&\#Errors&Algorithm&\#Edge Deletions&Reduction in Deletions&\#Label Deletions&Delete Reduction\\
\toprule
\multirow{18}{*}{\parbox{2cm}{ICIJ\newline
		\(1\)-way}}&1285&Naive Greedy&55&&39&28.05\\
&1285&LP-Guided Greedy&44.33&18.90&38.00&14.29\\
&1285&ILP&43&21.34&38&11.63\\
&1517&Naive Greedy&186&&167&10.38\\
&1517&LP-Guided Greedy&176.00&5.55&149.33&15.15\\
&1517&ILP&174&6.62&149&14.37\\
&2786&Naive Greedy&338&&261&22.88\\
&2786&LP-Guided Greedy&261.33&22.68&226.33&13.39\\
&2786&ILP&260&23.08&226&13.08\\
&175724&Naive Greedy&753&&581&22.92\\
&175724&LP-Guided Greedy&627.67&16.68&476.00&24.16\\
&175724&ILP&616&18.23&463&24.84\\
&558799&Naive Greedy&1411&&1008&28.54\\
&558799&LP-Guided Greedy&1188.00&15.80&826.00&30.47\\
&558799&ILP&1188&15.80&807&32.07\\
&2052352&Naive Greedy&timeout&&timeout&\\
&2052352&LP-Guided Greedy&timeout&&timeout&\\
&2052352&ILP&timeout&&timeout&\\
\midrule\multirow{8}{*}{\parbox{2cm}{LDBC\newline
		\(1\)-way}}&8704&Naive Greedy&5782&&5441&5.89\\
&8704&LP-Guided Greedy&4854.00&16.045&4711.00&2.95\\
&8704&ILP&4854&16.04&4711&2.95\\
&15337&Naive Greedy&11304&&10318&8.72\\
&15337&LP-Guided Greedy&10227.00&9.52&8978.00&12.21\\
&15337&ILP&10227&9.52&timeout&\\
&55957&Naive Greedy&12829&&11486&10.47\\
&55957&LP-Guided Greedy&11339.00&11.62&9732.00&14.17\\
&55957&ILP&timeout&&timeout&\\
&79064&Naive Greedy&22438&&21099&5.97\\
&79064&LP-Guided Greedy&22438.00&&19740.00&12.02\\
&79064&ILP&timeout&&timeout&\\
\bottomrule
\end{tabular} \end{table*}

\begin{table*}
	\caption{Runtimes for computing repairs with custom weights (PageRank scores)}\label{table:custom-weight-runtimes}
	\begin{tabular}{c r r l r r r r }
		\toprule
		Dataset&Shape&\#Errors&Algorithm&Runtime & Runtime Increase & Edge Deletions & Deletion Reduction\\
		\toprule
		\multirow{15}{*}{ICIJ}&$1$-way&558799&Naive Greedy&439.04&-1655.31&2643.00&\\
		&$1$-way&558799&LP-Guided Greedy&185.26&-640.69&1266.00&52.10\\
		&$1$-way&558799&ILP&25.01&&1264.00&52.18\\
		&$2$-rep&397108&Naive Greedy&212.25&&153.00&\\
		&$2$-rep&397108&LP-Guided Greedy&40.14&&73.00&52.29\\
		&$2$-rep&397108&ILP&&&&\\
		&$2$-way&7553&Naive Greedy&4.30&-693.78&1912.00&\\
		&$2$-way&7553&LP-Guided Greedy&2.43&-349.08&1189.67&37.78\\
		&$2$-way&7553&ILP&0.54&&1188.00&37.87\\
		&loop&56614&Naive Greedy&7.32&-312.05&732.00&\\
		&loop&56614&LP-Guided Greedy&7.63&-329.29&726.00&0.82\\
		&loop&56614&ILP&1.78&&726.00&0.82\\
		&$3$-split&876933&Naive Greedy&14.79&45.80&2.00&\\
		&$3$-split&876933&LP-Guided Greedy&26.10&4.39&1.00&50.0\\
		&$3$-split&876933&ILP&27.29&&1.00&50.0\\
		\midrule
		\multirow{6}{*}{Coreutils}&$1$-way&45728&Naive Greedy&250.91&-13373.86&8595.00&\\
		&$1$-way&45728&LP-Guided Greedy&103.88&-5478.28&6737.00&21.62\\
		&$1$-way&45728&ILP&1.86&&6737.00&21.62\\
		&$2$-way&3707&Naive Greedy&1.57&-473.44&786.00&\\
		&$2$-way&3707&LP-Guided Greedy&1.13&-310.75&751.00&4.45\\
		&$2$-way&3707&ILP&0.27&&708.00&9.92\\
		\bottomrule
	\end{tabular}
\end{table*}

\begin{table*}
	\caption{Runtimes for computing approximate repairs for the selection phase of the naive greedy (NG) and LP-guided greedy (LPG)}\label{table:approx-runtimes}
	\begin{tabular}{c r r l r}
		\toprule
		Dataset&Shape&\#Errors&Algorithm&Runtime\\
		\toprule
		\multirow{4}{*}{Coreutils}&$1$-way&45728&Approximate NG&0.38\\
		&$1$-way&45728&Approximate LPG&1.39\\
		&$2$-way&3707&Approximate NG&0.04\\
		&$2$-way&3707&Approximate LPG&0.10\\
		\midrule
		\multirow{6}{*}{Legislative}&$1$-way&166977&Approximate NG&1.69\\
		&$1$-way&166977&Approximate LPG&3.99\\
		&$2$-way&443695&Approximate NG&timeout\\
		&$2$-way&443695&Approximate LPG&timeout\\
		&loop&3783&Approximate NG&0.14\\
		&loop&3783&Approximate LPG&0.23\\
		\midrule
		\multirow{10}{*}{ICIJ}&$1$-way&558799&Approximate NG&11.15\\
		&$1$-way&558799&Approximate LPG&20.11\\
		&$2$-rep&397108&Approximate NG&7.41\\
		&$2$-rep&397108&Approximate LPG&19.72\\
		&$2$-way&7553&Approximate NG&0.08\\
		&$2$-way&7553&Approximate LPG&0.22\\
		&loop&56614&Approximate NG&0.93\\
		&loop&56614&Approximate LPG&1.68\\
		&$3$-split&876933&Approximate NG&15.69\\
		&$3$-split&876933&Approximate LPG&26.00\\
		\midrule
		\multirow{8}{*}{LDBC}&$1$-way&79064&Approximate NG&0.63\\
		&$1$-way&79064&Approximate LPG&2.77\\
		&$2$-rep&72756&Approximate NG&1.32\\
		&$2$-rep&72756&Approximate LPG&2.37\\
		&loop&423&Approximate NG&0.00\\
		&loop&423&Approximate LPG&0.01\\
		&$3$-split&1358499&Approximate NG&22.32\\
		&$3$-split&1358499&Approximate LPG&51.70\\
		\bottomrule
	\end{tabular}
\end{table*}

\begin{table*}
	\caption{Edge deletions of approximate repairs suggested by the selection phase of the naive greedy (NG) and LP-guided greedy (LPG)}\label{table:approx-deletions}
	\begin{tabular}{c r r l r}
		\toprule
		Dataset&Shape&\#Errors&Algorithm&Runtime\\
		\toprule
		\multirow{4}{*}{Coreutils}&$1$-way&45728&Approximate NG&11820.67\\
		&$1$-way&45728&Approximate LPG&6733.67\\
		&$2$-way&3707&Approximate NG&744.33\\
		&$2$-way&3707&Approximate LPG&603\\
		\midrule
		\multirow{6}{*}{Legislative}&$1$-way&166977&Approximate NG&5483.33\\
		&$1$-way&166977&Approximate LPG&2618\\
		&$2$-way&443695&Approximate NG&\\
		&$2$-way&443695&Approximate LPG&\\
		&loop&3783&Approximate NG&87\\
		&loop&3783&Approximate LPG&38\\
		\midrule
		\multirow{10}{*}{ICIJ}&$1$-way&558799&Approximate NG&1580.33\\
		&$1$-way&558799&Approximate LPG&1188\\
		&$2$-rep&397108&Approximate NG&113.67\\
		&$2$-rep&397108&Approximate LPG&92\\
		&$2$-way&7553&Approximate NG&1743.33\\
		&$2$-way&7553&Approximate LPG&1180\\
		&loop&56614&Approximate NG&731.67\\
		&loop&56614&Approximate LPG&725\\
		&$3$-split&876933&Approximate NG&9.67\\
		&$3$-split&876933&Approximate LPG&1\\
		\midrule
		\multirow{8}{*}{LDBC}&$1$-way&79064&Approximate NG&30464.33\\
		&$1$-way&79064&Approximate LPG&22438\\
		&$2$-rep&72756&Approximate NG&95.33\\
		&$2$-rep&72756&Approximate LPG&31\\
		&loop&423&Approximate NG&51.33\\
		&loop&423&Approximate LPG&46\\
		&$3$-split&1358499&Approximate NG&311.67\\
		&$3$-split&1358499&Approximate LPG&100\\
		\bottomrule
	\end{tabular}
\end{table*}

\end{document}